\documentclass[pra,amsmath,amssymb]{revtex4}
\usepackage{graphicx}
\usepackage{amssymb,amsthm,amsfonts,amstext}
\usepackage{amsmath}

\newtheorem{lemma}{Lemma}
\newtheorem{theorem}{Theorem}
\newcommand{\ket}[1]{|{#1}\rangle}
\newcommand{\bra}[1]{\langle{#1}|}

\begin{document}

\title{Hilbert space structure induced by quantum probes}
\author{Go Kato$^1$, Masaki Owari$^2$, Koji Maruyama$^{3}$}
\affiliation{$^{1}$NTT Communication Science Laboratories, NTT Corporation, Atsugi-Shi, Kanagawa,
243-0198 Japan}
\affiliation{$^{2}$Department of Computer Science, Shizuoka University, Hamamatsu 432-8011, Japan}
\affiliation{$^{3}$Department of Chemistry and Materials Science, Osaka City University, Osaka,
558-8585 Japan}

\begin{abstract}
In the general setting of quantum controls, it is unrealistic to control all of the degrees of freedom of a quantum system.
We consider a scenario where our direct access is restricted to a small subsystem $S$ that is constantly interacting with the rest of the system $E$. 
What we investigate here is the fundamental structure of the Hilbert space that is caused solely by the restrictedness of the direct control.
We clarify the intrinsic space structure of the entire system and that of the operations which could be activated through $S$. 
The structures hereby revealed would help us make quantum control problems more transparent and provide a guide for 
understanding what we can implement.
They can be deduced by considering an algebraic structure, which is the \textit{Jordan algebra} formed from Hermitian operators, naturally 
induced by the setting of limited access.
From a few very simple assumptions about direct operations, we elucidate rich structures of the operator algebras and  Hilbert spaces 
that manifest themselves in quantum control scenarios. 
\end{abstract}
\maketitle



\section{Introduction}

Understanding the dynamics of many-body quantum systems under artificial control is by no means
easy. As the race towards the realization of quantum computer is growing in momentum, a solid
theoretical foundation is desired more than ever in order to tame complex quantum dynamics
systematically. The principal difficulty is in the necessity of controlling exponentially many
degrees of freedom of a large quantum system through a limited number of controllable
parameters.

Since it is unrealistic to control all such degrees of freedom, the number of the modulable parameters is limited no matter 
what physical control scheme is employed. Thus,  natural questions would be what 
we can do to a given physical system under severe limitations on our artificial control and how  can it be done \cite{LLS04,BMM10,KP10}. 
Although there has been a widely accepted control method in the quantum information processing community, i.e., using a combination of one- and two-qubit operations, its prospects still look rocky in terms of scalability. This hiatus of the development in this direction encourages us to explore the problem from a more fundamental, or mathematical, point of view.

A major obstacle when scaling up a quantum system is the noise induced to the system through interactions with its environment. 
We thus consider a setting in which the system interacts with its environment minimally; most of the system is insulated from its surroundings 
and only a small subsystem is the subject of our direct control. The insulated part $E$ is connected only with the controllable subsystem $S$ 
through the \textit{drift Hamiltonian} $h_0$, and any operation can be applied to $S$ at will.
This type of scenario has recently been studied, mainly for systems of spins-1/2~\cite{BMM10,KP10,GLST17,LARR18}.

The most noteworthy tool for analyzing the dynamics in such a setting is the dynamical Lie algebra, which is a set of all realizable operators under the 
given condition \cite{JS72,RSDRP95,DAlessandroBook,MBBookChapter}. It can be calculated as the maximum 
set of independent operators that are generated by the 
drift Hamiltonian $h_0$ and Hamiltonians $\{h_k\}$ corresponding to modulable field parameters.

In order to 
make the setting realistic and mathematically tractable, we assume that $\{h_k\}$ forms a Lie algebra su(dim$\mathcal{H}_S$)
acting on $\mathcal{H}_S$, where $\mathcal{H}_S$ is the Hilbert space for a small subsystem $S$ 
of dimension dim$\mathcal{H}_S$ (Fig. \ref{fig:SandE}). The $S$ subsystem interacts through $h_0$ with the rest of the system, $E$,
which we also assume is finite-dimensional. 
\begin{figure}
\includegraphics[scale=0.8]{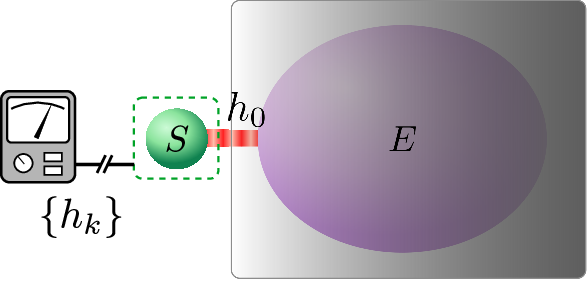}
\caption{A schematic view of the problem setting. A small subsystem $S$ can be directly accessed, while the rest of the 
system $E$ is beyond direct artificial control. The intrinsic dynamics of $S$ and $E$, including interactions between them, 
is governed by the drift Hamiltonian $h_0$. Any operation in su(dim$\mathcal{H}_S$) is applicable by modulating 
the Hamiltonians $\{h_k\}$ acting on $\mathcal{H}_S$.
}\label{fig:SandE}
\end{figure}

It is  clear that the dynamical Lie algebra does not necessarily span the Lie algebra 
su(dim$\mathcal{H}_S\cdot\mathrm{dim}\mathcal{H}_E)$ for the entire Hilbert space of the system. 
The dynamical Lie algebra has mostly been calculated and analyzed in an ad hoc fashion, depending on the specific physical system.
In fact, calculating the dynamical Lie algebra from a given set of Hamiltonians is hard; its  complexity is $O(d^8)$ for 
a $d$-dimensional system \cite{ZSH11}. This makes it extremely difficult to discuss the general properties of the controllability except in
some special cases, such as XY or Heisenberg spin chains that have high symmetry.
When the dynamical Lie algebra does not coincide with a simple Lie algebra on the whole Hilbert space, it is an unlucky case: 
the system is not fully controllable; thus, the $S$ part may need to be expanded, in the hope of making the controllable space larger. 

Now we ask ourselves whether there are intrinsic structures in the dynamical Lie algebra when artificial
controls are applied only to a small subsystem of a many-body system? In other words, what does
the structure of the Hilbert space look like, especially when it is not fully controllable? 
Also, what is the precise effect of expanding the accessible part $S$, namely, that of appending an ancillary system $\mathcal H_A$ 
to $\mathcal H_S$? Does it always help to enlarge the controllable space in $\mathcal H_E$?

In this paper, we classify the structure of the dynamical Lie algebra, which is induced by the restricted
access, as well as the Hilbert space structure that manifests itself accordingly. We then find that there is a clear distinction
between the cases of ${\rm dim}\mathcal H_S=2$ and ${\rm dim}\mathcal H_S\geq 3$. 
On the one hand, when ${\rm dim}\mathcal H_S\geq 3$, there appear only direct sums of
${\rm su}(\cdot)$. On the other hand, a structure of {\it formally real Jordan algebra}
explicitly emerges in the dynamical Lie algebra if ${\rm dim}\mathcal H_S=2$. Although the Jordan algebra was
introduced by Jordan et al. \cite{JNW34} as a  mathematical formulation of quantum mechanics, it has attracted 
relatively little attention in the quantum community. 

Further, we can see how the structures of these two cases correspond to each other, when an additional dimension(s) 
is appended to $S$. Looking into this correspondence allows us to answer the
question about the effect of ancilla: enlarging $\mathcal{H}_S$ does enhance the
controllability of quantum state of $E$ if ${\rm dim}\mathcal H_S=2$, while it does not otherwise. 
This is a somewhat unexpected result; one may envision that appending an ancilla to $S$ would not be of use at all
because what is interacting with $E$ is still only the original $S$ itself. 
One's   intuition may be  opposite to such a view; as reported in \cite{OMTK15}, the size of 
the ancillary system could help make the probable subspace in $E$ larger. Our result proves that these ideas are over-naive.

Investigating spatial structures will also have direct and important consequences with respect to the 
system identifiability. There has been intensive 
research on the problem of quantum system identification under limited access
\cite{BMN09,FPK09,BM09,BMN11,Shabani11,JS14,KY14,SC17a,SC17b}, since the knowledge of the system Hamiltonian 
is crucially important for control. 
A number of identification schemes have been discovered so far, and at the same time it is becoming
clearer that there may exist limitations on what we can observe through $S$. The Hilbert space structures 
we elucidate here will provide a useful toolbox to address all these key issues systematically.

\section{Main results}
\label{sec:main} The physical setup we consider is as follows. 
We suppose there is a 
 quantum system ${\mathcal H}_S$,
on which arbitrary control can be applied at will, that  interacts with an external system
${\mathcal H}_E$ coherently. The dynamics of $\mathcal{H}_E$, including the interaction with
$\mathcal{H}_S$, are described by the drift Hamiltonian $h_0$, and $\mathcal{H}_E$ is not subject to our direct
control (See Fig. \ref{fig:SandE}). That is, we can access $\mathcal{H}_E$ only indirectly through $\mathcal{H}_S$. 
Also, we assume that the Hilbert spaces ${\mathcal H}_E$ and ${\mathcal H}_S$ are both finite dimensional.

The dynamical Lie algebra $L$ is crucial in the analysis of the controllability of a quantum system.
 It is a Lie algebra generated by $ih_0$ and a set $\{Id_E\}\otimes {\rm su}({\rm
dim}\mathcal H_S)$ of operators. Here, $Id_E$ is the identity operator on ${\mathcal H}_E$, $\{Id_E\}$ is a one-dimensional space generated by $Id_E$, and ${\rm su}({\rm dim}\mathcal H_S)$ is the set of all traceless skew-Hermitian operators acting on ${\mathcal H}_S$, thus representing a set of arbitrary controls. A direct product of the operator sets $\mathcal S_1\otimes \mathcal S_2$ is a set of $s_1\otimes s_2$ for all
$s_b\in\mathcal S_b\;(b=\{1,2\})$, and $i\mathcal S$ means the set of elements $i \cdot s$ for all $s\in\mathcal{S}$.

We now present five central theorems about the structure of the dynamical Lie algebra, as well as
that of the space ${\mathcal H}_E$. Before presenting them, let us introduce a few terms.
\begin{itemize}
\item The {\it connected algebra} $L_c$ is the smallest ideal\footnote{The ideal $L'$ of a Lie algebra $L$ is a subspace of $L$ such that it cannot be expanded by taking commutators between $L'$ and $L$, i.e. $L'\supseteq[L,L']$.}
 of $L$ which includes
$\{Id_E\}\otimes {\rm su}({\rm dim}\mathcal H_S)$, i.e.
\begin{equation}\label{Lc_def}
L_c := \mathcal L(\{[\cdots[[g^\prime,g_1],g_2],\cdots,g_n]|n\in {\mathbb Z}_{\geq 1}\land g_m\in L\land g^\prime\in \{Id_E\}\otimes {\rm su}({\rm dim}\mathcal H_S) \}),
\end{equation}
where
 $\mathcal L(\mathcal S)$ indicates a set of all real linear combinations of the elements in $\mathcal S$.

\item The {\it disconnected algebra} $L_d$ is  the set of all skew-Hermitian operators which commute with
any element in $L_c$, i.e.
\begin{equation}\label{Ld_def}
L_d := \{g|g\in {\rm u}({\rm dim}\mathcal H_E\cdot {\rm dim}\mathcal H_S )\land \forall g^\prime\in L_c, [g,g^\prime]=0 \},
\end{equation}
where ${\rm u}({\rm dim}\mathcal H_E\; {\rm dim}\mathcal H_S )$ is the set of all skew-Hermitian
operators on ${\mathcal H}_E\otimes{\mathcal H}_S$.
From the Jacobi relation, we can verify that
$L_d$ is also a Lie algebra.
\end{itemize}

That our direct access is restricted to $\mathcal{H}_S$ necessarily imposes a nontrivial structure on the dynamical Lie
algebra. Let us summarize the rough 
ideas behind
 the main theorems before presenting them in a rigorous manner.
Throughout this paper, the structure of the Hilbert space $\mathcal{H}_E$ is the structure in the context of quantum control; namely it is what we shall ``see and control'' through $S$. 
\begin{enumerate}
\item[Theorem 1: ] Any element in the dynamical Lie algebra $L$ is a sum of two elements, one of which is controllable from operations on $S$ and the other is uncontrollable \footnote{Even in the case where $L$ is equal to su(dim$\mathcal{H}_E\cdot\mathrm{dim}\mathcal{H}_S)$, the disconnected algebra $L_d$ can formally be identified as a one-dimensional Lie algebra $\{i\cdot Id\}$. The connected algebra $L_c$ is then equal to $L$.}. These two are the elements of subalgebras $L_c$ and $L_d$, respectively.
\item[Theorem 2: ] When dim$\mathcal{H}_S\ge 3$, the Hilbert space $\mathcal{H}_E$ can have a direct sum structure with subspaces, each of which may be a direct product of two spaces, $\mathcal{H}_R$ and $\mathcal{H}_B$. The dynamics on $\mathcal{H}_R$ are driven by $L_c$, while those on 
 $\mathcal{H}_B$ are driven by $L_d$. Thus, $\mathcal{H}_B$ cannot be controlled through operations on $\mathcal{H}_S$. In other words, the limitedness of direct access to $S$ induces a natural basis structure in $E$. 
\item[Theorem 3: ] When dim$\mathcal{H}_S= 2$, $\mathcal{H}_E$ has a direct sum structure, similarly to the case of dim$\mathcal{H}_S\ge 3$; however, there may be a restriction on $L_c$. 
\item[Theorem 4: ] The algebraic structures shown in Theorems 2 and 3 are sufficient conditions for $L$ to be a Lie algebra that contains su(dim$\mathcal{H}_S$).
\item[Theorem 5: ] This theorem shows how the space structure changes when an additional dimension(s) is appended to a two-dimensional $\mathcal{H}_S$.  
\end{enumerate}

The theorems are not restricted to the setting with a single drift Hamiltonian $i h_0$. This is because we do not impose 
any specific constraints on the combination of physical Hamiltonians to obtain the dynamical Lie algebra, that is, 
there could be multiple drift Hamiltonians $\{i h_0^{(p)}\}_p$, instead of one. What we classify is the structure of the dynamical 
Lie algebra $L$, which contains $Id\otimes \mathrm{su(dim}\mathcal{H}_S)$, so the theorems are valid for such cases as well.

\noindent
\subsection{Induced structure of the dynamical Lie algebra $L$} 
The following three theorems describe the precise structure of the Hilbert space of $E$ as well as that of the dynamical Lie algebra $L$, 
and how it depends on the dimensionality of $\mathcal{H}_S$. 

\begin{theorem}
\label{th:con_discon_Lie_alg}
 The algebra $L$ is a subspace of the direct sum of $L_d$ and $L_c$:
\begin{eqnarray}
& L\subseteq {\mathcal L}(L_d\cup L_c),
\label{eq:str_abs_L}\\
& L_d\cap L_c=\{0\}.
\label{eq:str_abs_L_2}
\end{eqnarray}
This, together with the relation $L_c\subseteq L$, implies $L={\mathcal L}((L_d\cap L)\cup L_c)$.
\end{theorem}


\begin{theorem}
\label{th2:str_ds3_dem}
When ${\rm dim}\mathcal H_S\ge 3$, the space ${\mathcal H}_E$ has the structure of  a direct sum of subspaces, each of which is a direct product of two spaces,
\begin{equation}
{\mathcal H}_E = \bigoplus_j \mathcal{H}_{E_j} = \bigoplus_j{\mathcal H}_{B_j}\otimes {\mathcal H}_{R_j},
\label{eq:structure_e_ds3}
\end{equation}
and the precise nature of these subspaces depends on $L$.

In accordance with the decomposition {\rm (\ref{eq:structure_e_ds3})}, $L_d$ and $L_c$ are written as direct
sums of subalgebras as 
\begin{eqnarray}
L_d &=&  \bigoplus_j {\rm u}({\rm dim}\mathcal H_{B_j})\otimes \{Id_{R_j}\otimes Id_{S}\}
\label{eq:str_sub_Lu} \makebox{ and} \\
L_c &=& \bigoplus_j \{Id_{B_j}\}\otimes {\rm su}({\rm dim}\mathcal H_{R_j}\cdot{\rm dim}\mathcal H_S).
\label{eq:str_sub_Lc}
\end{eqnarray}

Moreover, this intrinsic structure stays the same, even if an ancillary space $\mathcal H_{S'}$ is appended to $\mathcal{H}_S$ to enlarge the directly accessible space. That is, if we let $L'$ be the `expanded' Lie algebra generated by $\{Id\}\otimes {\rm su}( {\rm dim}\mathcal H_S\cdot {\rm dim}\mathcal H_{S'})$ and $i h_0\otimes Id_{S'}$, 
where $h_0$ is the drift Hamiltonian, the corresponding connected and disconnected algebras, $L_c'$ and $L_d'$, are 
${\rm u}({\rm dim}\mathcal H_{B_j})\otimes \{Id_{R_j}\otimes Id_{S}\otimes Id_{S'}\}$, and 
$ \{Id_{B_j}\}\otimes {\rm su}({\rm dim}\mathcal H_{R_j}\cdot {\rm dim}\mathcal H_S\cdot {\rm dim}\mathcal H_{S'})$, respectively.
 
\end{theorem}


\begin{theorem}
\label{th3:str_ds2_dem}
When ${\rm dim}\mathcal H_S=2$, the space ${\mathcal H}_E$
 has the structure of  a direct sum of subspaces
 ${\mathcal H}_{E_j}^*$, i.e.
\begin{equation}
{\mathcal H}_E = \bigoplus_j{\mathcal H}_{E_j}^*,
\label{eq:structure_e_ds2}
\end{equation}
such that the disconnected and connected algebras, $L_d$ and $L_c$, can be written as direct
sums of subalgebras, each of which acts on a subspace ${\mathcal H}_{E_j}^*\otimes {\mathcal H}_S$.
Similarly to Theorem \ref{th2:str_ds3_dem}, the detail of each  subspaces in Eq.(\ref{eq:structure_e_ds2}) is determined  by $L$.

Further, these subalgebras of $L_d$ and $L_c$  have the forms,
\begin{eqnarray}
&  i\hat J_j \otimes \{Id_S\} \label{eq:str_sub_Lu_2} \mbox{ and} \\
&  {\mathcal L}(i\bar J_j\otimes \{Id_S\} \cup J_j\otimes {\rm su}({\rm dim}\mathcal H_S)),
\label{eq:str_sub_Lc_2}
\end{eqnarray}
respectively, where the triple of the operator sets $(J_j,\bar J_j,\hat J_j)$ is equal to one of the
following three types: $({\mathfrak{R}}, \bar {\mathfrak{R}}, \hat {\mathfrak{R}})$, 
$({\mathfrak{M}}_{\gamma}^{(k)},\bar {\mathfrak{M}}_{\gamma}^{(k)}, \hat{\mathfrak{M}}_{\gamma}^{(k)})$
or $({\mathfrak{S}}_{n},\bar {\mathfrak{S}}_{n}, \hat{\mathfrak{S}}_{n})$. Depending on the type of $J_j$ among $\mathfrak{R}, \mathfrak{M}_\gamma^{(k)}$, and $\mathfrak{S}_n$, $\mathcal{H}_{E_j}^*$ has a finer structure shown below in Eq. (\ref{eq:str_sub_e}).

\end{theorem}

The notations for the sets, $\mathfrak R$, $\mathfrak M_\gamma^{(k)}$, and $\mathfrak S_n$, are after \cite{JNW34}, and their details will be given later in this section (from Eq. (\ref{eq:JordanAlg_1}) to Eq. (\ref{eq:JordanAlg_6})). The indices $\gamma, k,$ and $n$ that specify the structure of operator sets $\mathfrak M_\gamma^{(k)}$, $\mathfrak S_n$ are integers such that $\gamma\geq 3$, $k\in\{1,2,4\}$ and $n\geq 3$. 
Also, we will introduce  sets with accent signs, $\hat{\bullet}$ and $\bar{\bullet}$, in Eqs. (\ref{eq:Rbar})-(\ref{eq:str_hat_6}), which are defined in 
correspondence to each of $\mathfrak R$, $\mathfrak M_\gamma^{(k)}$, and $\mathfrak S_n$.

The subspaces ${\mathcal H}_{E_j}$ or $\mathcal{H}_{E_j}^*$ have a fine structure depending on the type of $J_j$:
\begin{equation}
\mathcal{H}_{E_j}^* =  
\left\{\begin{array}{ll} 
\makebox[5cm][l]{$
{\mathcal H}_{A_j} 
$} &
\makebox{when }J_j=\mathfrak{R}
\\
\makebox[5cm][l]{$
{\mathcal H}_{A_j}\otimes{\mathcal H}_{Q_j}
$} &
\makebox{when }J_j=\mathfrak{M}_{\gamma}^{(k)}
\makebox{ for $k\in\{1,2\}$}
\\
\makebox[5cm][l]{$
{\mathcal H}_{A_j}\otimes{\mathcal H}_{Q_j^{(1)}}\otimes{\mathcal H}_{Q_j}
$} &
\makebox{when }J_j=\mathfrak{M}_{\gamma}^{(4)}
\\
{\mathcal H}_{A_j}\otimes{\mathcal H}_{Q_j^{(\lceil n/2\rceil-1)}}\otimes{\mathcal H}_{Q_j^{(\lceil n/2\rceil-2)}}\otimes \cdots \otimes{\mathcal H}_{Q_j^{(1)}}
&
\makebox{when }J_j=\mathfrak{S}_{n}
\end{array}
\right.
\label{eq:str_sub_e}
\end{equation}
where ${\rm dim}\mathcal H_{A_j}\ge 1$,
${\rm dim}\mathcal H_{Q_j}=\gamma$, and all other spaces, ${\mathcal H}_{Q_j^{(1)}}$, ${\mathcal H}_{Q_j^{(2)}}$, $\cdots$, are two-dimensional. 

If $J_j=\mathfrak{S}_{2n'}$ or $\mathfrak{M}_{\gamma}^{(2)}$ for $n'\in\mathbb N_{>1}$, there appears a Hermitian operator $Z^*_j$ in the
representations of $({\mathfrak{S}}_{2n'},\bar {\mathfrak{S}}_{2n'}, \hat{\mathfrak{S}}_{2n'})$ and
$({\mathfrak{M}}_{\gamma}^{(2)},\bar {\mathfrak{M}}_{\gamma}^{(2)}, \hat{\mathfrak{M}}_{\gamma}^{(2)})$ (see Eqs. {\rm (\ref{eq:JordanAlg_1})-(\ref{eq:str_hat_6})}).
The operators $Z^*_j$ acting on the space ${\mathcal H}_{A_j}$ have eigenvalues $+1$ and/or $-1$. The dimensions of $\mathcal H_{A_j}$ and $\mathcal H_{Q_j}$, as well as the precise form of $Z^*_j$, may differ for each $j$, even if $J_j$ could be of the same type for all $j$, e.g., $J_j={\mathfrak{M}}_{\gamma}^{(2)} \;(\forall j)$.

Theorem \ref{th:con_discon_Lie_alg} states that, we can uniquely divide any drift Hamiltonian $h_0$, which describes the (unmodulable) interaction between the systems $E$ and $S$, into two parts $h_d\in L_d$ and $h_c\in L_c$. This division is done such that the $h_d$ part has no effect on the dynamics in the space $\mathcal H_S$, and the other part
$h_c$ represents the interaction between $\mathcal H_S$ and  $\mathcal H_E$.

Theorem \ref{th2:str_ds3_dem} conveys a somewhat strong message. It claims that, when ${\rm dim}\mathcal H_S\geq 3$, even if we attach an additional quantum system $S'$ to $S$, intending to enlarge the effective work space, it does not expand the set of executable operations for $\mathcal{H}_E$. That is, if we let $L'$ denote the Lie algebra generated by $L\otimes \{Id_{S'}\}$ and $\{Id_E\}\otimes {\rm su}({\rm dim}\mathcal H_S\cdot{\rm dim}\mathcal H_{S'})$, the set of all generators in $E$ and $S$ that are possible under the expansion $S^\prime$ is still the same as $L$;
\begin{equation}\label{possibleOperationsInE}
\{g|g\otimes Id_{S'}\in L'\}=L.
\end{equation}

One common message from Theorems \ref{th2:str_ds3_dem} and \ref{th3:str_ds2_dem} is that, regardless of the dimension of the system $S$, the system $E$ would have a direct sum structure as in Eqs. (\ref{eq:structure_e_ds3}) and (\ref{eq:structure_e_ds2}). Thus, the quantum dynamics cannot make a state jump between different subspaces in the sum, which is already a significant consequence of the limited access. Theorems \ref{th2:str_ds3_dem} and \ref{th3:str_ds2_dem} then state further that there are substantial differences in the fine structures of each subspace, depending on whether $\mathrm{dim}\mathcal{H}_S$ is larger than or equal to 2. 

\noindent
\subsection{Sufficient conditions required for $L$, $L_d$, and $L_c$} 
Next, we give sufficient conditions for the operator sets $L$, $L_d$ and $L_c$ to be a Lie algebra, disconnected and connected algebras for the Lie algebra $L$, respectively.
As a matter of fact, having the structures stated in Theorems \ref{th:con_discon_Lie_alg} and \ref{th2:str_ds3_dem}, as well as the rather trivial property of $L\cap L_d$  being closed under the commutator, are sufficient for them to have the necessary properties mentioned with Eqs. (\ref{Lc_def}) and (\ref{Ld_def}).


\begin{theorem}
\label{th:str_32_suf}
Suppose that ${\rm dim}\mathcal H_S\ge 3$ and $\mathcal{H}_E$ can be decomposed into $\mathcal H_{\tilde B_j}\otimes\mathcal H_{\tilde R_j}$ such that $\mathcal{H}_E=\bigoplus_j \mathcal H_{\tilde B_j} \otimes \mathcal H_{\tilde R_j}$. Also, define $\tilde{L}_d$ and $\tilde{L}_c$ according to this space decomposition as 
\begin{eqnarray}
\tilde L_d &:=& \bigoplus_j
  {\rm u}({\rm dim}\mathcal H_{\tilde B_j})\otimes \{Id_{\tilde R_j}\otimes Id_{S}\},
\label{eq:def_tilde_L_d_3}
\\
\tilde L_c &:=& \bigoplus_j
 \{Id_{\tilde B_j}\}\otimes {\rm su}({\rm dim}\mathcal H_{\tilde R_j}\cdot {\rm dim}\mathcal H_S).
\label{eq:def_tilde_L_c_3}
\end{eqnarray}

If the set of operators $\tilde L$ on $\left(\bigoplus_j{\mathcal H}_{\tilde B_j}\otimes {\mathcal H}_{\tilde R_j}\right) \otimes{\mathcal H}_S$ satisfy 
\begin{eqnarray}
\tilde L  &:=& \mathcal L\left(\tilde L_d'\cup \tilde L_c\right),
\label{eq:def_tilde_L_3}
\\
\tilde L_d' &\subseteq& \tilde L_d,
\label{eq:def_tilde_L'_3}
\end{eqnarray}
such that $\tilde L_d'$ is closed under the commutator, then so is $\tilde L$, and  $\tilde L_d$ and $\tilde L_c$ are the disconnected and the connected algebras for $\tilde L$.

If $\mathrm{dim}\mathcal{H}_S=2$ and $\mathcal{H}_E$ can be decomposed into $\mathcal{H}_{\tilde{E}_j}^\diamond$, i.e., $\mathcal{H}_E=\bigoplus_j \mathcal{H}_{\tilde{E}_j}^\diamond$, the above statement still holds with the following modifications to the definitions of $\tilde{L}_d$ and $\tilde{L}_c$. Namely, 
\begin{eqnarray}
\tilde{L}_d &:=& \bigoplus_j i\hat J_j\otimes\{Id_S\},
\label{eq:def_tilde_L_d_2}
\\
\tilde{L}_c &:=& \bigoplus_j  \mathcal L\left(i\bar J_j\otimes\{ Id_S\}\cup J_j\otimes {\rm su}({\rm dim}\mathcal H_S)\right),
\label{eq:def_tilde_L_c_2}
\end{eqnarray}
where $(J_j,\bar J_j, \hat J_j)$ is equal to one of the triples of operator sets,
 $({\mathfrak R},\bar{\mathfrak R}, \hat{\mathfrak R})$,
 $({\mathfrak M}_\gamma^{(k)},\bar{\mathfrak R}_\gamma^{(k)}, \hat{\mathfrak R}_\gamma^{(k)})$ and
 $({\mathfrak S}_n,\bar{\mathfrak S}_n, \hat{\mathfrak S}_n)$. 
Naturally, $\tilde{L}$ in Eq. (\ref{eq:def_tilde_L_3}) should be considered to be an operator set acting on $\left(\bigoplus_j \mathcal{H}_{\tilde{E}_j}^\diamond \right) \otimes \mathcal{H}_S$.
\end{theorem}

Theorem \ref{th:str_32_suf} reveals the structure of the dynamical Lie algebra $L$, which contains arbitrary generators on the space $\mathcal H_S$.
It implies that the structure of the space in $\mathcal H_E$ may not be trivial at all. By a trivial structure, we mean that $\mathcal H_E$ is a simple direct product of two spaces $\mathcal H_{E_1}$ and $\mathcal H_{E_2}$, i.e., $\mathcal H_E=\mathcal H_{E_1} \otimes \mathcal H_{E_2}$. If the Hamiltonian had the form, $h_0=Id_{E_1} \otimes h_{E_2}\otimes h_S$,  then obviously the space $\mathcal H_{E_1}$ cannot be accessed from $\mathcal H_S$, while $\mathcal H_{E_2}$ can. 
What is claimed above is, however, that the accessible and inaccessible spaces in $\mathcal H_E$ would have more complex and rich structure because of the restrictedness of our physical access. 

\noindent
\subsection{Relation between structures when ${\rm dim}\mathcal H_S =2$ and ${\rm dim}\mathcal H_S\geq 3$}
From the quantum control perspective, one might naively think of enlarging the controllable space in $E$ by introducing an additional system $S'$ that interacts with $S$. We have mentioned above that this is not possible when $\mathrm{dim}\mathcal{H}_S\ge 3$, but what happens if we append an ancillary system $S'$ when $\mathrm{dim}\mathcal{H}_S=2$? The following theorem depicts the transition that occurs when an ancillary system $S'$ (obviously, dim$\mathcal{H}_{S^\prime}\ge 2$) is added to the two-dimensional $S$.


\begin{theorem}
\label{th:str_rel}
Let $L'$ be an \textit{expanded} Lie algebra generated by $h_0\otimes Id_{S'}$ and $\{Id_E\}\otimes {\rm su}({\rm dim}\mathcal H_S\cdot {\rm dim}\mathcal H_{S'})$. The expansion of the accessible space from $S$ to $S'$ causes a change in the structure of $\mathcal{H}_E$ from that of Eq. (\ref{eq:structure_e_ds2}) to Eq. (\ref{eq:structure_e_ds3}). (Below, primed indices are for the spaces after expanding $\mathrm{dim}\mathcal{H}_S=2$ to $\mathrm{dim}(\mathcal{H}_S\otimes \mathcal{H}_{S^\prime})\ge 3$.)

If $J_j$ in Eq. (\ref{eq:str_sub_Lc_2}) is equal to one of $\mathfrak{R}$, $\mathfrak{M}_{\gamma}^{(1)}$, $\mathfrak{M}_{\gamma}^{(4)}$
or  $\mathfrak{S}_{2n'-1}$ with $n'\in \mathbb N_{>1}$, there is a one-to-one correspondence between $j$ and $j'$ 
such that 
\[
\mathcal{H}_{E_j}^*=\mathcal H_{E_{j^\prime}}=\mathcal H_{B_{j'}}\otimes \mathcal H_{R_{j'}}.
\]
If $J_j$ is equal to either
$\mathfrak{M}_{\gamma}^{(2)}$
or $\mathfrak{S}_{2n'}$, the subspace $\mathcal{H}_{E_j}$ splits into a direct sum of two direct products:
\begin{eqnarray*}
\mathcal{H}_{E_j}^* &=& \mathcal{H}_{E_{j^\prime+}}\oplus \mathcal{H}_{E_{j^\prime-}} \\
&=&  (\mathcal H_{B_{j^\prime+}} \otimes \mathcal H_{R_{j^\prime+}}) \oplus (\mathcal H_{B_{j^\prime-}} \otimes \mathcal H_{R_{j^\prime-}})   \\
&=&\left\{
\begin{array}{ll}
(\mathcal H_{A_j^{(+1)}}\otimes \mathcal H_{Q_{j}})\oplus (\mathcal H_{A_j^{(-1)}}\otimes \mathcal H_{Q_{j}}), \;\;
&\makebox{when } J_j=\mathfrak{M}_\gamma^{(2)} \\
(\mathcal H_{A_j^{(+1)}}\otimes \mathcal H_{Q_{j}^{(n^\prime-1)}}\otimes\cdots\otimes \mathcal{H}_{Q_j^{(1)}}) 
\\
\quad\quad\oplus 
(\mathcal H_{A_j^{(-1)}}\otimes \mathcal H_{Q_{j}^{(n^\prime-1)}}\otimes\cdots\otimes \mathcal{H}_{Q_j^{(1)}}), \;\;
&\makebox{when } J_j=\mathfrak{S}_{2n^\prime} \\
\end{array}
\right.
\end{eqnarray*}
where $\mathcal{H}_{A_j^{(\pm 1)}}$ are the eigenspaces of the $Z_j^*$ operator on $\mathcal{H}_{A_j}$ corresponding to its eigenvalues $\pm 1$,
and $j^\prime\pm$ are the indices for distinguishing these subspaces.  
\end{theorem}

The structures of $\mathcal H_{E_j}^*$ in Eq. (\ref{eq:str_sub_e}) are related to those in Eq. (\ref{eq:structure_e_ds3}) as follows:
\begin{eqnarray}
&&\left\{
\begin{array}{ll}
{\mathcal H}_{B_{j'}}={\mathcal H}_{A_j}, &\makebox{when } J_j=\mathfrak{R}, \mathfrak{M}_{\gamma}^{(1)}, \mathfrak{M}_{\gamma}^{(4)}, 
\mathfrak{S}_{2n'-1} \; (n'>1), \\
{\mathcal H}_{B_{j^\prime \pm}}={\mathcal H}_{A_j^{(\pm 1)}}, &\makebox{when } J_j=\mathfrak{M}_{\gamma}^{(2)}, \mathfrak{S}_{2n'},
\end{array}
\right. \\
&& \mathcal{H}_{R_j^\prime}\;\;\mbox{or}\;\; {\mathcal H}_{R_{j^\prime \pm}}=
\left\{
\begin{array}{ll}
\makebox[3cm][l]{$
{\mathcal H}_{Q_j}
$} &
\makebox{when }J_j=\mathfrak{R}, \mathfrak{M}_{\gamma}^{(k)} \; k\in\{1,2\},
\\
\makebox[3cm][l]{$
{\mathcal H}_{Q_j^{(1)}}\otimes{\mathcal H}_{Q_j}
$} &
\makebox{when }J_j=\mathfrak{M}_{\gamma}^{(4)},
\\
{\mathcal H}_{Q_j^{(\lceil n/2\rceil-1)}}\otimes{\mathcal H}_{Q_j^{(\lceil n/2\rceil-2)}}\otimes \cdots \otimes{\mathcal H}_{Q_j^{(1)}} &
\makebox{when }J_j=\mathfrak{S}_{n},
\end{array}
\right.
\end{eqnarray}
for $b\in\{+1,-1\}$. When $J_j=\mathfrak R$, we consider $\mathcal H_{A_j}$ to be a direct product of itself and a one-dimensional space $\mathcal H_{Q_j}$.

Also, the connected algebra $L_c^\prime$ of $L^\prime$ after appending $S^\prime$ will be of the form in Eq. (\ref{eq:str_sub_Lc}), 
i.e., su(dim$\mathcal H_{R_j}\cdot \mathrm{dim}\mathcal H_S)$ on each block subspace, and the disconnected algebra $L_d^\prime$ is related to the original $L_d$ as 
\begin{equation}
L_d' = L_d\otimes \{Id_{S'}\}.
\label{eq:str_Lcp}
\end{equation}

\subsection{Physical examples}
Expansion of the controllable space is a topic in the study of quantum controllability of specific physical systems. For example, in \cite{BMM10}, 
indirect control was discussed for a one-dimensional chain of $N$ spin-1/2 particles whose dynamics are governed by the drift Hamiltonian 
\begin{equation}\label{xxham}
ih_0^\mathrm{XX} = \frac{i}{2}\sum_{k=1}^N c_k[(1+\gamma)X_k X_{k+1} + (1-\gamma)Y_k Y_{k+1}]+b_k Z_k,
\end{equation}
where the last term represents the Zeeman interaction with a static magnetic field in the $z$-direction and $\gamma$ is the anisotropy parameter. 
Despite what it may imply, the order of the spin spaces is the opposite to our convention, e.g., that in Eq. (\ref{eq:str_sub_Lu}) or (\ref{eq:str_sub_e}); the $S$ subsystem
is spin 1, which is at the left end, while in Eq. (\ref{eq:str_sub_e}),  it is assumed to be attached to the right end.

\begin{figure}
\includegraphics[scale=0.4]{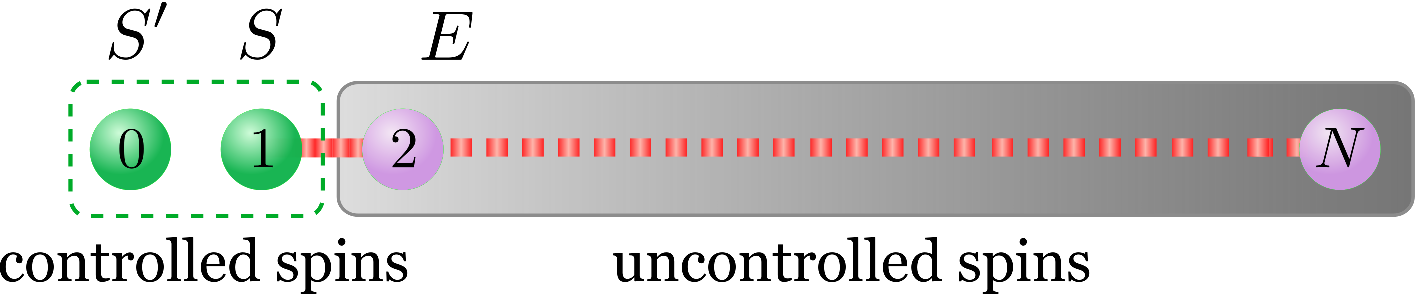}
\caption{A one-dimensional spin chain considered for control in \cite{BMM10}. The two spins at the chain end are in the directly
accessible subsystem, and the rest of the chain, $E$, only evolves through the drift Hamiltonian 
$h_0^\mathrm{XX}$ of Eq. (\ref{xxham}). Spins 1 and 0 are  labeled $S$ and $S^\prime$ in line with the description of 
algebra expansion in the main text. Any su(4) operation can be applied to spins 0 and 1; this applicability of arbitrary su(4) operations
is achieved by assuming the same $h_0^\mathrm{XX}$-type interaction between them in \cite{BMM10}.
}\label{fig:spinchain}
\end{figure}

The Hamiltonian Equation (\ref{xxham}) describes the so-called XX-type interaction between neighboring spins, 
and the paper \cite{BMM10} presented a specific and efficient scheme 
to control the entire chain through $S$ containing two end spins, i.e., those labeled as $k=1$ and 2 (See Fig. \ref{fig:spinchain}).
The inclusion of two spins in $S$ is necessary, since having direct controllability of only one spin at the chain end does not lead to full controllability over the entire chain with the above drift Hamiltonian $ih_0^\mathrm{XX}$. 
More precisely, with $ih_0^\mathrm{XX}$ and su(dim$\mathcal{H}_S$)=su(2) for spin 1, the connected algebra is equal to 
\[
L_c=\mathcal{L}(i\bar{\mathfrak{S}}_{2N-1}\otimes \{Id_S\}\cup \mathfrak{S}_{2N-1}\otimes \mathrm{su(dim}\mathcal{H}_S)),
\]
with the Hilbert space structure
\begin{equation}\label{HE_spinchain}
\mathcal{H}_E={\mathcal H}_{Q^{(N-1)}}\otimes \cdots \otimes{\mathcal H}_{Q^{(1)}},
\end{equation}
where each $\mathcal{H}_{Q^{(n)}}$ is a two-dimensional space corresponding to each spin from $k=2$ to $N$. This can be verified by looking at the 
specific structure of the algebras $J_j$ in Eqs. (\ref{eq:JordanAlg_1})-(\ref{eq:JordanAlg_6}). There is only a single $j$ in this case;
thus it is omitted in Eq. (\ref{HE_spinchain}). Note that in our space structure notation,
the $S$ space interacts with the rightmost one, $\mathcal{H}_{Q_j^{(1)}}$, and because 
dim(dim$\mathcal{H}_A)=1$ in this case, it is omitted in Eq. (\ref{HE_spinchain}). If an extra spin, say, spin 0, 
is attached as $S^\prime$ to spin 1, the algebra on $\mathcal{H}_E$, which is determined by the dynamical Lie algebra $L$, 
changes. Namely, the connected algebra $L_c$ becomes that of Eq. (\ref{eq:str_sub_Lc}), i.e., the full su($\cdot$) 
algebra on $\mathcal{H}_E$.

A simple example in which the split of $\mathcal{H}_A$ can be observed is a chain of three spins-1/2, whose Hamiltonian is
\begin{equation}\label{ham_a_split}
ih_0^\mathrm{XX^\prime} = X_1 X_2 +Y_1 Y_2 + X_2 X_3,
\end{equation}
which may be regarded as a special case of the XX Hamiltonian. Then, with the spin 1 being the $S$ subsystem, this $ih_0^\mathrm{XX^\prime}$ is 
of the type $\mathfrak{S}_4$, and there is a $Z^*$ operator acting on $\mathcal{H}_A$, which is $X_3$ on spin 3 
in the basis used above (See Eq. (\ref{eq:JordanAlg_6})). 
The Hilbert space structure under $ih_0^\mathrm{XX^\prime}$ and su(2) (for spin 1) is the one in Eq. (\ref{eq:str_sub_e}), namely, 
\[
\mathcal{H}_{E}^*=\mathcal{H}_{A}\otimes \mathcal{H}_{Q^{(1)}},
\]
where the subscript $j$ is again omitted since there is only one element in the direct sum. 
Here, $\mathcal{H}_A$ and $\mathcal{H}_{Q^{(1)}}$ are the Hilbert spaces for spins 3 and 2, respectively. 
If we add another controllable spin-1/2 to $S$ so that any su(4) operation becomes available in this subsystem, the space $\mathcal{H}_{A}$ 
splits into two parts as $\mathcal{H}_{A^{(+1)}}\oplus \mathcal{H}_{A^{(-1)}}$. The overall $E$ space then becomes 
\[ 
\mathcal{H}_E=(\mathcal{H}_{A^{(+1)}} \otimes \mathcal{H}_{Q^{(1)}}) \oplus (\mathcal{H}_{A^{(-1)}} \otimes \mathcal{H}_{Q^{(1)}}),
\]
which is in the form of Eq. (\ref{eq:structure_e_ds3}) for the case dim$\mathcal{H}_S\ge 3$.

\subsection{Representations of triple $(J_j,\bar J_j, \hat J_j)$}
Before concluding this section, we show below  explicit representations of candidates for the triple $(J_j,\bar J_j, \hat J_j)$.
Although they look rather complex, they will be of use for understanding how the  controls on $S$ affect $E$ indirectly. 

First, the forms of the operator sets for $J$ are obtained in Lemma \ref{sec:rep_Jor}, as a consequence of the anti-commutation relations required for operators in the algebra, which stems from the limited access to the system (shown in Lemmas 1-4). 
Their specific types are denoted as $\mathfrak{R}$, $\mathfrak{M}_{\gamma}^{(k)}$, and  $\mathfrak{S}_{n}$ and are given as follows: 
\begin{eqnarray}
\mathfrak{R} &:=&  \{Id_A\},
\label{eq:JordanAlg_1}
\\
\mathfrak{M}_{\gamma}^{(1)} &:=&   \mathcal{L} (\{Id_A\otimes X_{k,q}, Id_A\otimes\ket{k}\bra{k}\}_{k\neq q\in\{0,1,\cdots,\gamma-1\}}),
\label{eq:JordanAlg_2}
\\
\mathfrak{M}_{\gamma}^{(2)} &:=&   \mathcal{L} ( \{Id_A\otimes X_{k,q}, Id_A\otimes\ket{k}\bra{k}, Z^{*}\otimes Y_{k,q}\}_{k\neq q\in\{0,1,\cdots,\gamma-1\}}),
\label{eq:JordanAlg_3}
\\
\mathfrak{M}_{\gamma}^{(4)} &:=&   \mathcal{L} ( \{Id_A\otimes Id_{Q^{(1)}}\otimes X_{k,q}, Id_A\otimes Id_{Q^{(1)}}\otimes\ket{k}\bra{k},
 Id_A\otimes W\otimes Y_{k,q}\}_{W\in\{X,Y,Z\},\;k\neq q\in\{0,1,\cdots,\gamma-1\}}),
\label{eq:JordanAlg_4}
\\
\mathfrak{S}_{2n'-1}
&:=&   \mathcal{L} ( \{ \overbrace{Id\otimes\cdots\otimes Id}^{n'-m}\otimes W\otimes\overbrace{Y\otimes\cdots\otimes Y}^{m-1},
 Id\otimes\cdots\otimes Id\}_{W\in\{X,Z\},m\in\{1,2,\cdots,n'-1\}}),
\label{eq:JordanAlg_5}
\\
\mathfrak{S}_{2n'}
&:=&   \mathcal{L} ( \{ \overbrace{Id\otimes\cdots\otimes Id}^{n'-m}\otimes W\otimes\overbrace{Y\otimes\cdots\otimes Y}^{m-1}, Id\otimes\cdots\otimes Id,
 Z^{*}\otimes\overbrace{Y\otimes\cdots\otimes Y}^{n'-1}\}_{W\in\{X,Z\},m\in\{1,2,\cdots,n'-1\}}),
\label{eq:JordanAlg_6}
\end{eqnarray}
where the generalized Pauli operators, $X_{j,k}:=\ket j\bra k+\ket k \bra j$,
 $Y_{j,k}:=-i \ket j\bra k+i\ket k \bra j$,
 $Z_{j,k}:=\ket j\bra j-\ket k \bra k$,
 $X:=X_{0,1}$,
 $Y:=Y_{0,1}$, and 
 $Z:=Z_{0,1}$, are used, 
and $\{\ket j\}_{j\in \{0,1,\cdots\}}$ represents a basis for each space. The operator $Z^*$ is the one mentioned after Eq. (\ref{eq:str_sub_e}), namely, it is a Hermitian operator which satisfies $Z^{*2}=Id_A$ and characterizes subalgebras. 
We have omitted the index $j$, indicating the subspace of $\mathcal H_E$ or $\mathcal{H}_E^*$, for both spaces and operators, for simplicity. We shall do so in the following as well, as long as there is no risk of confusion. 

Second, as for those with a bar, $\bar{\mathfrak{R}}$, $\bar{\mathfrak{M}}_{\gamma}^{(k)}$, and $\bar{\mathfrak{S}}_{n}$, we define them as  Eqs. (\ref{eq:Rbar})-(\ref{eq:Sbar_2n}). They are determined so that they satisfy the relation, $\bar{J}=i\mathcal{L}([J,J])$, which is proved in Lemma \ref{lemm:relation}, for the corresponding $J$ given in Eqs. (\ref{eq:JordanAlg_1})-(\ref{eq:JordanAlg_6}). 
\begin{eqnarray}
\bar{\mathfrak{R}} &:=&  \{0\}, \label{eq:Rbar}\\
\bar{\mathfrak{M}}_{\gamma}^{(1)} &:=&  \mathcal{L}(\{Id_A\otimes Y_{k,q}\}_{k\neq q\in\{0,1,\cdots,\gamma-1\}}), \label{eq:Mbar_gamma1}\\
\bar{\mathfrak{M}}_{\gamma}^{(2)} &:=&  \mathcal{L}(\{Id_A\otimes Y_{k,q}, Z^{*}\otimes X_{k,q}, Z^{*}\otimes Z_{k,q}\}_{k\neq q\in\{0,1,\cdots,\gamma-1\}}), \label{eq:Mbar_gamma2}\\
\bar{\mathfrak{M}}_{\gamma}^{(4)} &:=&  \mathcal{L}(\{Id_A\otimes Id_{Q^{(1)}}\otimes Y_{k,q}, Id_A\otimes W\otimes X_{k,q}, 
\nonumber\\&&\quad
Id_A\otimes W\otimes \ket k\bra k\}_{W\in\{X,Y,Z\},k\neq q\in\{0,1,\cdots,\gamma-1\}}), \label{eq:Mbar_gamma4}\\
\bar{\mathfrak{S}}_{2n'-1} &:=&   \mathcal{L}(\{\overbrace{Id\otimes\cdots\otimes Id}^{n'-m_{2}}\otimes W\otimes\overbrace{Y\otimes\cdots\otimes
Y}^{m_{2}-m_{1}-1}
\nonumber\\&&\quad
\otimes W^{\prime}\otimes\overbrace{Id\otimes\cdots\otimes
Id}^{m_{1}-1}\}_{W,W^{\prime}\in\{X,Z\},m_{1}<m_{2}\in\{1,2,\cdots,n'-1\}} \nonumber \\ 
&&\cup \{\overbrace{Id\otimes\cdots\otimes Id}^{n'-m}\otimes Y\otimes\overbrace{Id\otimes\cdots\otimes Id}^{m-1}\}_{m\in\{1,2,\cdots,n'-1\}}), \label{eq:Sbar_2n-1}\\
\bar{\mathfrak{S}}_{2n'} &:=&  \mathcal{L}( \{ \overbrace{Id\otimes\cdots\otimes Id}^{n'-m_{2}}\otimes W\otimes\overbrace{Y\otimes\cdots\otimes Y}^{m_{2}-m_{1}-1}
\nonumber\\&&\;
\otimes W^{\prime}\otimes\overbrace{Id\otimes\cdots\otimes Id}^{m_{1}-1}
\}_{W,W^{\prime}\in\{X,Z\},m_{1}<m_{2}\in\{1,2,\cdots,n'-1\}} \nonumber\\
&&\cup \{\overbrace{Id\otimes\cdots\otimes Id}^{n'-m}\otimes
Y\otimes\overbrace{Id\otimes\cdots\otimes Id}^{m-1},
\nonumber\\&&
\;
Z^{*}\otimes\overbrace{Y\otimes\cdots\otimes
Y}^{n'-m-1}\otimes W\otimes\overbrace{Id\otimes\cdots\otimes
Id}^{m-1}\}_{W\in\{X,Z\},m\in\{1,2,\cdots,n'-1\}}).
\nonumber\\
\label{eq:Sbar_2n}
\end{eqnarray}

Finally, the operator sets with a hat, $\hat{\mathfrak{R}}$, $\hat{\mathfrak{M}}_{\gamma}^{(k)}$, and $\hat{\mathfrak{S}}_{n}$, are those that commute with the corresponding $J$, i.e., $[\hat{J}, J]=0$ (Lemma \ref{lemm:orthogonal_subspace}). Their forms are:
\begin{eqnarray}
\hat{\mathfrak{R}}
&:=&     \{h\}_{h\in  i\cdot {\rm u}({\rm dim}\mathcal H_A)},
\label{eq:str_hat_1}
\\
\hat{\mathfrak{M}}_{\gamma}^{(1)}
&:=&     \{h\otimes Id_Q\}_{h\in i\cdot {\rm u}({\rm dim}\mathcal H_A)},
\\
\hat{\mathfrak{M}}_{\gamma}^{(2)}
&:=&     \{h\otimes Id_Q\}_{h\in i\cdot{\rm u}({\rm dim}\mathcal H_A)^{*}},
\\
\hat{\mathfrak{M}}_{\gamma}^{(4)}
&:=&     \{h\otimes Id\otimes Id\}_{h\in i\cdot {\rm u}({\rm dim}\mathcal H_A)},
\\
\hat{\mathfrak{S}}_{2n'-1}
&:=&     \{h\otimes Id\otimes\cdots\otimes Id\}_{h\in i\cdot {\rm u}({\rm dim}\mathcal H_A)},
\\
\hat{\mathfrak{S}}_{2n'}
&:=&     \{h\otimes Id\otimes\cdots\otimes Id\}_{h\in i\cdot {\rm u}({\rm dim}\mathcal H_A)^{*}}.
\label{eq:str_hat_6}
\end{eqnarray}
where ${\rm u}({\rm dim}\mathcal H_A)^*$ is the set of all elements in ${\rm u}({\rm dim}\mathcal H_A)$ that commute with $Z^*$.

\section{Properties of the algebra $L$}

Before giving the proofs of the above theorems, let us study the properties of the algebra $L$. We shall use a number of lemmas to prove propositions
in what follows, and the proofs of those lemmas are given in the supplementary material.
Let $g$ be any operator in $L$, then $g$ can be written uniquely, regardless of dim($\mathcal{H}_S$), as 
\begin{equation}
g=g_{Id}\otimes Id_S + \sum_{W\in H_S} g_W\otimes W,
\end{equation}
where $g_{Id}$ and $g_W$ are skew-Hermitian operators acting on the space $\mathcal{H}_E$, and $H_S$ is the basis of $i\cdot$su(dim$\mathcal{H}_S$) consisting of operators $X_{k,q}$, $Y_{k,q}$ and $Z_{k,k+1}$ for $k,q\in \{0,1,\cdots ,{\rm dim}\mathcal H_S-1\}\;(k<q)$. Defining two operator sets by
\begin{eqnarray}
G^{(0)} &:=&  \{g_{Id}\}_{g\in L} \\
G^{(1)} &:=&  \mathcal{L}(\{g_W\}_{g\in L, W\in H_S}),
\end{eqnarray}
we can show 
\begin{equation}\label{eq:basic_str}
L=\mathcal{L}(G^{(0)}\otimes \{Id_S\}\cup iG^{(1)}\otimes  {\rm su}({\rm dim}\mathcal H_S))
\end{equation}
(See Lemma \ref{sec:str_L}). We shall call the pair of sets $G^{(0)}$ and  $G^{(1)}$ the identifiers of the dynamical Lie algebra $L$. 

These identifiers are shown to satisfy the following (anti-)commutation relations in Lemma \ref{sec:prop_Gs}:
\begin{eqnarray}
\label{eq:rel_Gs_1}
\left[G^{(b)},G^{(b)}\right] &\subseteq G^{(0)},\\
\label{eq:rel_Gs_2}
\left[G^{(0)},G^{(1)}\right] &\subseteq G^{(1)},\\
i\left\{ G^{(1)},G^{(1)}\right\} &\subseteq G^{(1)},
\label{eq:rel_Gs_3}
\end{eqnarray}
for $b\in \{0,1\}$. Further, only when ${\rm dim}\mathcal H_S\geq 3$, another commutation relation
\begin{equation}
\left[G^{(1)},G^{(1)}\right]\subseteq G^{(1)}
\label{eq:g1_comm}
\end{equation}
is required (Lemma \ref{sec:prop_Gs}). 

Since $iG^{(1)}$ is closed under the anti-commutator, $iG^{(1)}$ is a Jordan algebra, and is formed by Hermitian operators, including the identity operator, $Id_E$. Then, as shown in Lemma \ref{sec:rep_Jor}, $iG^{(1)}$ can be written as a direct sum of simple Jordan algebras $J_j$ regardless of $\mathrm{dim}\mathcal{H}_S$,
\begin{equation}
G^{(1)} = \bigoplus_ji J_j,
\label{eq:G_1}
\end{equation}
and $J_j$ has to have one of the structures in Eqs. (\ref{eq:JordanAlg_1})-(\ref{eq:JordanAlg_6}). 
Lemma \ref{sec:rep_Jor} also proves that the structure of $\mathcal{H}_{E_j}^*$ in Eq. (\ref{eq:str_sub_e}) is then obtained in accordance with that of $J_j$. The explicit representations of $J_j$ obtained thereby then allow us to have those of $\hat J$ in Eqs. (\ref{eq:str_hat_1})-(\ref{eq:str_hat_6}) and $\bar J$ in Eqs. (\ref{eq:Rbar})-(\ref{eq:Sbar_2n}), with the help of the commutation relations shown in Lemma \ref{lemm:relation}.
Also, from Eqs. (\ref{eq:rel_Gs_1}) and (\ref{eq:G_1}),  and $i[J_j,J_j]\subseteq \bar J_j$ from Lemma \ref{lemm:relation}, we obtain
\begin{equation}
\bigoplus_j i \bar J_j \subseteq \label{eq:G_0_l}G^{(0)}.
\end{equation}
These relations allow us to express ${\mathcal H}_E$ as a direct sum of the spaces ${\mathcal H}_{E_j}^*$, such that any element in $\bar J_j$ and $J_j$ is an operator on ${\mathcal H}_{E_j}^*$. 

Since the identity operator is in all simple Jordan algebras, the projection operator $P_{E_j}$ onto ${\mathcal H}_{E_j}^*$ is in $iG^{(1)}$.
It then follows from Eq. (\ref{eq:rel_Gs_2}) that an operator $\left[g,P_{E_{j'}}\right] \; (\forall g\in G^{(0)})$ must be block diagonalized into the subspaces ${\mathcal H}_{E_j}^*$.
Thus, any element in $G^{(0)}$ is also block diagonalized accordingly, and we let $G^{(0)}_j$ be the set of block elements of $g\in G^{(0)}$ 
whose action is restricted to the subspace ${\mathcal H}_{E_j}^*$.
From Eq. (\ref{eq:rel_Gs_2}), we see $\left[G^{(0)}_j,J_j\right]\subseteq J_j$, and this condition enforces $G^{(0)}_j$ to be a subset of $i{\mathcal L}(\hat J_j\cup \bar J_j)$, where $(\hat J_j, \bar J_j)$ is equal to one of the pairs $(\hat{\mathfrak R}, \bar{\mathfrak R})$,
$(\hat{\mathfrak M}_\gamma^{(k)},\bar {\mathfrak M}_\gamma^{(k)})$, and
$(\hat{\mathfrak S}_n, \bar{\mathfrak S}_n)$, depending on whether $J_j=\mathfrak R$ or ${\mathfrak M}_\gamma^{(k)}$ or ${\mathfrak S}_n$, respectively.
This is because, as shown in Lemma \ref{lemm:pseudo_Ideal}, $i{\mathcal L}(\hat J_j\cup \bar J_j)$ turns out to be the maximum set $J'$ of Hermitian operators that satisfy $i[J',J_j]\subseteq J_j$, namely, ${\mathcal L}(\hat J_j\cup \bar J_j)=\{h|h\in  i\cdot {\rm u}({\rm dim}\mathcal H_{E_j})\land \forall h'\in J_j,\;i[h,h']\in J_j\}$.

Combining these results, we arrive at
\begin{equation}
G^{(0)} \subseteq \bigoplus _j G^{(0)}_j \subseteq \bigoplus_j i{\mathcal L}(\hat J_j\cup \bar J_j),
\label{eq:G_0_u}
\end{equation}
and Eqs. (\ref{eq:g1_comm}), (\ref{eq:G_1}), (\ref{eq:G_0_l}) and (\ref{eq:G_0_u}) imply a relation 
\begin{eqnarray}
&&\bigoplus_j{\mathcal L}(i\bar J_j\otimes \{Id_S\}\cup J_j\otimes {\rm su}({\rm dim}\mathcal H_S))
\nonumber\\
&&\subseteq  L\subseteq \bigoplus_j{\mathcal L}(i\hat J_j\otimes \{Id_S\} \cup i\bar J_j\otimes \{Id_S\} \cup J_j\otimes {\rm su}({\rm dim}\mathcal H_S)).
\label{eq:uppper_bound_DLA}
\end{eqnarray}
All these relations (except for Eq. (\ref{eq:g1_comm})) hold, regardless of ${\rm dim}\mathcal H_S$.

\section{Proofs of theorems}\label{sec:TheoremProofs}
We now show the proofs of Theorems \ref{th:con_discon_Lie_alg}--\ref{th:str_rel} by using the relations we have given in the last section, as well as the specific representations of $(J_j,\bar J_j, \hat J_j)$. Since much of the mathematical argument in the proofs, which is mainly about the structure of the formally real Jordan algebra, is quite involved, we shall only delineate the proofs here to help readers grasp the picture, relying on the lemmas shown in the following section. Those lemmas are devoted to explaining the mathematics behind the proofs of the theorems. 

We shall start with the proof of Theorem \ref{th2:str_ds3_dem}, and then go on to show in the order of Theorems \ref{th3:str_ds2_dem}, \ref{th:con_discon_Lie_alg}, \ref{th:str_32_suf}, and \ref{th:str_rel}.

\noindent\textbf{Proof of Theorem \ref{th2:str_ds3_dem}}. 
When ${\rm dim}\mathcal H_S\geq 3$, Eqs. (\ref{eq:g1_comm}) and (\ref{eq:G_1}) lead to the condition $i[\bigoplus_j  J_j,\bigoplus_j  J_j]\subseteq \bigoplus_j  J_j$, which is equivalent to requiring $i[J_j,J_j]\subseteq J_j$ for all $j$. 
This enforces us to choose ${\mathfrak R}$, ${\mathfrak M}_\gamma^{(2)}$ and $\mathfrak{S}_4$ with $Z^*_j=Id_{A_j}$ or $-Id_{A_j}$ as possible structures of $J_j$ among those in Eqs. (\ref{eq:JordanAlg_1})-(\ref{eq:JordanAlg_6}). It is not hard to verify that others in these equations, such as ${\mathfrak M}_\gamma^{(4)}$, do not fulfill the above condition. Note that the structures of $J_j$ and $\mathcal{H}_{E_j}^*$, Eqs. (\ref{eq:JordanAlg_1})-(\ref{eq:JordanAlg_6}) and (\ref{eq:str_sub_e}), are derived in Lemma \ref{sec:rep_Jor}. In this case of dim($\mathcal{H}_S)\ge 3$, $\mathcal{H}_{E_j}^*$ are relabeled as $\mathcal{H}_{E_j}$ in Theorem \ref{th2:str_ds3_dem}.

When $J_j$ is equal to ${\mathfrak R}=\{Id_A\}$, $\mathcal{H}_{E_j}$ is as simple as a single subspace ${\mathcal H}_{A_j}$ (See Eq. (\ref{eq:str_sub_e})). Thus, by regarding $\mathcal{H}_{B_j}=\mathcal{H}_{A_j}$ and $\mathrm{dim}\mathcal{H}_{R_j}=1$, $\mathcal{H}_{E_j}$ in Eq. (\ref{eq:structure_e_ds3}) has a structure ${\mathcal H}_{B_j}\otimes {\mathcal H}_{R_j}$.
If $J_j$ is equal to ${\mathfrak M}_\gamma^{(2)}$ or $\mathfrak S_4$, $\mathcal{H}_{E_j}$ takes the form of $\mathcal{H}_{A_j}\otimes \mathcal{H}_{Q_j}$ or $\mathcal{H}_{A_j}\otimes \mathcal{H}_{Q_j^{(1)}}$, respectively, according to Eq. (\ref{eq:str_sub_e}). It is then obvious that $\mathcal{H}_{E_j}$ has a structure of Eq. (\ref{eq:structure_e_ds3}), by assigning the first and the second subspaces in the tensor product to be $\mathcal{H}_{B_j}$ and $\mathcal{H}_{R_j}$. Thus, $\mathcal{H}_E$ can be written $\bigoplus_j \mathcal H_{B_j}\otimes\mathcal  H_{R_j}$, irrespective of the form of $J_j$. 

Having identified the subspaces of $\mathcal H_E$, the algebra on $\mathcal{H}_{E}\otimes \mathcal{H}_S$, i.e., $\mathcal L(i\bar J_j\otimes \{Id_S\}\cup J_j\otimes {\rm su}({\rm dim}\mathcal H_S ))$, turns out to be $\{Id_{B_j}\}\otimes {\rm su}({\rm \dim}\mathcal H_{R_j}\cdot{\rm \dim}\mathcal H_{S})$. As a result, the relation (\ref{eq:uppper_bound_DLA}) is reduced to 
\begin{equation}
\bigoplus_j
\{Id_{B_j}\}\otimes {\rm su}({\rm dim}\mathcal H_{R_j}\cdot {\rm dim}\mathcal H_S) 
\subseteq L \subseteq \bigoplus_j{\mathcal L}(
{\rm u}({\rm dim}\mathcal H_{B_j})\otimes \{Id_{R_j}\otimes Id_S\}\cup
\{Id_{B_j}\}\otimes {\rm su}({\rm dim}\mathcal H_{R_j}\cdot {\rm dim}\mathcal H_S)),
\label{eq:uppper_bound_DLA_3}
\end{equation}
which, according to Lemma \ref{sec:abstract_structure_lemm10}, implies that the disconnected and the connected algebras are given by
\begin{eqnarray}
L_d&=&\bigoplus_j 
{\rm u}({\rm dim}\mathcal H_{B_j})\otimes \{Id_{R_j}\otimes Id_S\}, \;\mathrm{and}
\label{eq:str_Lu_3}
\\
L_c&=&\bigoplus_j
\{Id_{B_j}\}\otimes {\rm su}({\rm dim}\mathcal H_{R_j}\cdot {\rm dim}\mathcal H_S).
\label{eq:str_Lc_3}
\end{eqnarray}
Hence, Theorem \ref{th2:str_ds3_dem} is proved. Note that the last statement in Theorem \ref{th2:str_ds3_dem} can be verified rather straightforwardly, since the Lie algebra $L'$ generated by $L\otimes \{Id_{S'}\}$ and $\{Id_E\}\otimes{\rm su}({\rm dim}\mathcal H_S\cdot {\rm dim}\mathcal H_S^\prime) $  satisfies the relation 
\begin{eqnarray}
&& \bigoplus_j
\{Id_{B_j}\}\otimes {\rm su}({\rm dim}\mathcal H_{R_j}\cdot {\rm dim}\mathcal H_S\cdot {\rm dim}\mathcal H_{S'})
\nonumber \\
&&\subseteq L' \subseteq \bigoplus_j{\mathcal L}({\rm u}({\rm dim}\mathcal H_{B_j})\otimes \{Id_{R_j}\otimes Id_S\otimes Id_{S'}\}
\cup
\{Id_{B_j}\}\otimes {\rm su}({\rm dim}\mathcal H_{R_j}\cdot {\rm dim}\mathcal H_S\cdot {\rm dim}\mathcal H_{S'})).
\end{eqnarray}

\noindent\textbf{Proof of Theorem \ref{th3:str_ds2_dem}}.  
The statement of Theorem \ref{th3:str_ds2_dem} is nothing but the consequence of Lemma \ref{sec:abstract_structure}, which states that when  ${\rm dim}\mathcal H_S=2$ Eq. (\ref{eq:uppper_bound_DLA}) implies
\begin{eqnarray}
L_d&=&\bigoplus_j i\hat J_j\otimes \{Id_S\}, \;\mathrm{and}
\label{eq:str_Lu_2}
\\
L_c&=&\bigoplus_j{\mathcal L}(i\bar J_j\otimes \{Id_S\}\cup J_j\otimes {\rm su}({\rm dim}\mathcal H_S)).
\label{eq:str_Lc_2}
\end{eqnarray}

\noindent\textbf{Proof of Theorem \ref{th:con_discon_Lie_alg}}.   
When ${\rm dim}\mathcal H_S\geq 3$, we see from Eqs. (\ref{eq:str_Lu_3}) and (\ref{eq:str_Lc_3}) that $L_d$ and $L_c$ do not have an overlap, thus $L_d\cap L_c=\{0\}$, which is Eq. (\ref{eq:str_abs_L_2}). Also, the inclusion relation of the right side of Eq. (\ref{eq:uppper_bound_DLA_3}) together with Eqs. (\ref{eq:str_Lu_3}) and (\ref{eq:str_Lc_3}) imply $L\subseteq \mathcal{L}(L_d\cup L_c)$ in Eq. (\ref{eq:str_abs_L}).

When  ${\rm dim}\mathcal H_S=2$, Eqs. (\ref{eq:str_Lu_2}) and (\ref{eq:str_Lc_2}) and the explicit expressions of $(J_j,\bar J_j, \hat J_j)$ guarantee the relation in Eq. (\ref{eq:str_abs_L_2}). Also, these expressions and the inclusion relation on the right of Eq. (\ref{eq:uppper_bound_DLA}) leads to Eq.  (\ref{eq:str_abs_L}).

\noindent\textbf{Proof of Theorem \ref{th:str_32_suf}}.  
That $\tilde{L}:=\mathcal{L}(\tilde{L}_d^\prime\cup \tilde{L}_c)$ is closed under the commutator can be seen as
\begin{eqnarray}
[\tilde L,\tilde L]&=&[\mathcal L( \tilde L_d^\prime\cup \tilde L_c),\mathcal L(\tilde L_d^\prime\cup \tilde L_c)]
\nonumber\\
&\subseteq&
\mathcal L([ \tilde L_d^\prime,\tilde L_d^\prime]\cup [ \tilde L_d^\prime,\tilde L_c]\cup [ \tilde L_c, \tilde L_c])
\nonumber\\
&=&
\mathcal L([ \tilde L_d^\prime,\tilde L_d^\prime]\cup [ \tilde L_c, \tilde L_c])
\nonumber\\
&\subseteq &
\mathcal L( \tilde L_d^\prime\cup  \tilde L_c) = \tilde L.
\label{eq:tildeL_closed}
\end{eqnarray}
The second inclusion relation stems from the bilinearity of the commutator.
The equality in the third line is due to the commutation relation $[ \tilde L_d^\prime,\tilde L_c]\subseteq [ \tilde L_d,\tilde L_c]=\{0\}$, which is verified with the definitions of $\tilde L_d$ and $\tilde L_c$, i.e.,  Eqs. (\ref{eq:def_tilde_L_d_3}) and (\ref{eq:def_tilde_L_c_3}). Since $\tilde{L}_d^\prime$ is assumed to be closed under the commutator and so is $\tilde{L}_c$ by Eq. (\ref{eq:def_tilde_L_c_3}), we verify the inclusion relation in the fourth line.

Lemma \ref{sec:abstract_structure_lemm10} tells that if Eq. (\ref{eq:tildeL_closed}) and
\begin{equation}
\tilde L_c \subseteq  \tilde L\subseteq \mathcal L(\tilde L_d\cup \tilde L_c ),
\label{eq:LandLdc_inclusion}
\end{equation}
which is trivially obtained from Eqs. (\ref{eq:def_tilde_L_3}) and  (\ref{eq:def_tilde_L'_3}), hold, then $\tilde L_d$ and $\tilde L_c$ are the disconnected and the connected algebras. Therefore, the first half of Theorem \ref{th:str_32_suf} is justified.

The second half of Theorem \ref{th:str_32_suf}, which is for the case of $\mathrm{dim}\mathcal{H}_S=2$, can be proved in a similar manner. 
Although Eq. (\ref{eq:tildeL_closed}) can be shown to be true,  the relations $[\tilde L_d,\tilde L_c]=\{0\}$ and $[\tilde L_c,\tilde L_c]\subseteq \tilde L_c$ need a bit different reasonings. The former is justified by Eq. (\ref{eq:relation_3}) shown in Lemma \ref{lemm:relation}. We can also check $[\tilde L_c,\tilde L_c]\subseteq \tilde L_c$ by using
$i[ \bar J_j,\bar J_j]\subseteq \bar J_j$,
$i[ J_j, J_j]\subseteq \bar J_j$,
$i[ \bar J_j, J_j]\subseteq  J_j$, and
$\{ J_j, J_j\}\subseteq J_j$, which are from Eqs. (\ref{eq:relation_1}) and (\ref{eq:relation_2}) in Lemma \ref{lemm:relation}. 
(Lemma \ref{sec:prop_Gs_suf}).
With Eq. (\ref{eq:LandLdc_inclusion}), which holds when $\mathrm{dim}\mathcal{H}_S=2$ as well, and Lemma \ref{sec:abstract_structure}, we can show that $\tilde{L}_d$ and $\tilde{L}_c$ in Eqs. (\ref{eq:def_tilde_L_d_2}) and (\ref{eq:def_tilde_L_c_2}) are the disconnected and the connected algebras.

\noindent\textbf{Proof of Theorem \ref{th:str_rel}}.   
Given an \textit{expanded} dynamical Lie algebra $L'$, there must be its identifier $(G^{(0)\prime},G^{(1)\prime})$, such that $L'=\mathcal L(G^{(0)\prime}\otimes \{Id_S\otimes Id_{S^\prime}\}\cup iG^{(1)\prime}\otimes {\rm su}({\rm dim}H_S\cdot {\rm dim}H_{S'}))$ (Lemma \ref{sec:str_L}). These sets satisfy 
\begin{eqnarray}
G^{(1)} = \bigoplus_j iJ_j \subseteq \bigoplus_j iJ_j' \subseteq G^{(1)\prime}, &\;\;\mathrm{and}
\label{eq:G^(1)^prime}
\\
\bigoplus_j  [ J_j',J_j'] \subseteq  G^{(0)\prime},\quad\quad
\label{eq:G^(0)^prime}
\end{eqnarray}
where the first equality in Eq. (\ref{eq:G^(1)^prime}) is from Eq. (\ref{eq:G_1}). Each algebra $J_j'$ is equal to one of the following:
\begin{eqnarray}
\label{eq:def_J'_R}
{\mathfrak{R}'}
&:=&   \{Id_A\},
\\
{\mathfrak{M}}_{\gamma}^{\prime(1)}
&:=& i   \{Id_A\}\otimes    {\rm u}({\rm dim}\mathcal H_Q),
\label{eq:def_J'_M1}
\\
{\mathfrak{M}}_{\gamma}^{\prime(2)}
&:=&
i\mathcal L((\{Id_{A^{(+1)}}\}\oplus
\{Id_{A^{(-1)}}\})\otimes {\rm u}({\rm dim}\mathcal H_Q)),
\label{eq:def_J'_M2}
\\
{\mathfrak{M}}_{\gamma}^{\prime(4)}
&:=& i   \{Id_A\}\otimes  {\rm u}({\rm dim}\mathcal H_{Q^{(1)}}\cdot {\rm dim}\mathcal H_{Q}),
\label{eq:def_J'_M4}
\\
{\mathfrak{S}}_{2n'-1}'
&:=& i   \{Id_A\}\otimes
{\rm u}({\rm dim}\mathcal H_{Q^{(n^\prime-1)}}\cdot {\rm dim}\mathcal H_{Q^{(n^\prime-2)}}\cdots {\rm dim}\mathcal H_{Q^{(1)}}),
\label{eq:def_J'_Sodd}
\\
{\mathfrak{S}}_{2n'}'
&:=&
i\mathcal L((\{Id_{A^{(+1)}}\}\oplus
\{Id_{A^{(-1)}}\})
\otimes {\rm u}({\rm dim}\mathcal H_{Q^{(n^\prime-1)}}\cdot {\rm dim}\mathcal H_{Q^{(n^\prime-2)}}\cdots {\rm dim}\mathcal H_{Q^{(1)}})).
\label{eq:def_J'_Seven}
\end{eqnarray}
There is a one-to-one correspondence between these primed algebras and the non-primed ones in Eqs. (\ref{eq:JordanAlg_1})-(\ref{eq:JordanAlg_6}). For example, if one of the $J_j$ was $\mathfrak{M}_\gamma^{(1)}$ when $\mathrm{dim}\mathcal{H}_S=2$, then appending an ancillary space $\mathcal{H}_{S^\prime}$ makes it change to $\mathfrak{M}^{\prime (1)}_\gamma$. 

The right-most inclusion in Eq. (\ref{eq:G^(1)^prime}) can be justified by the following three facts:
First, $L\otimes \{Id_{S'}\} \subseteq L'$, since $L'$ is a Lie algebra generated by $L\otimes \{Id_{S'} \}$ and $Id_E\otimes {\rm su}({\rm dim}\mathcal H_S\cdot {\rm dim}\mathcal H_{S'})$. Second, as Lemma \ref{sec:prop_Gs} tells, $G^{(1)\prime}$ must be closed under two binary operations $[\cdot ,\cdot ]$ and $i\{\cdot ,\cdot \}$ since ${\rm dim}\mathcal H_S\cdot {\rm dim}\mathcal H_{S'}$ is more than $2$. Third, $iJ'_j$ in Eqs. (\ref{eq:def_J'_R})-(\ref{eq:def_J'_Seven}) are the smallest skew-Hermitian operator sets which contain the corresponding $iJ_j$ and are closed under the binary operations (Lemma \ref{sec:str_mod}).
Equation (\ref{eq:G^(0)^prime}) is simply due to Eq. (\ref{eq:rel_Gs_1}) for $L^\prime$, $[G^{(1)\prime},G^{(1)\prime}]\subseteq G^{(0)\prime}$.

Since $L^\prime=\mathcal{L}(G^{(0)\prime}\otimes \{Id_S\otimes Id_{S^\prime}\} \cup G^{(1)\prime}\otimes \mathrm{su}(\mathrm{dim}\mathcal{H}_S\cdot \mathrm{dim}\mathcal{H}_{S^\prime}))$, together with Eqs. (\ref{eq:G^(1)^prime}) and (\ref{eq:G^(0)^prime}), we obtain
\begin{equation}
\bigoplus_j{\mathcal L}([J_j',J_j']\otimes \{Id_S\otimes  Id_{S'}\} \cup
J_j'\otimes {\rm su}({\rm dim}\mathcal H_S\cdot {\rm dim}\mathcal H_{S'})) \subseteq L^\prime.
\label{eq:str_det_Lup}
\end{equation}
On the other hand, we can verify the relation
\begin{equation}
L^\prime \subseteq 
\bigoplus_j{\mathcal L}(
i\hat J_j\otimes \{Id_S\otimes   Id_{S'}\}
\cup
   [J_j',J_j']\otimes \{Id_S\otimes  Id_{S'}\}
 \cup
  J_j'\otimes {\rm su}({\rm dim}\mathcal H_S\cdot {\rm dim}\mathcal H_{S'})),
\label{eq:str_det_Lp}
\end{equation}
using the explicit expressions of $J_j'$ in Eqs. (\ref{eq:def_J'_R})-(\ref{eq:def_J'_Seven}), together with Eqs. (\ref{eq:str_Lu_2}) and (\ref{eq:str_hat_1})$\sim$(\ref{eq:str_hat_6}); that is, we can readily see that the set on the RHS of Eq. (\ref{eq:str_det_Lp}) is closed under the commutator and contains all generators in $L\otimes \{Id_{S'}\}$ and $\{Id_E\}\otimes {\rm su}({\rm dim}\mathcal H_S\cdot {\rm dim}\mathcal H_{S'})$.
By redefining the structure of $\mathcal H_E$ as $\mathcal H_E=\bigoplus_{j'}\mathcal H_{B_{j'}}\otimes \mathcal H_{R_{j'}}$ as in Theorem \ref{th:str_rel},  Eqs. (\ref{eq:str_det_Lup}) and (\ref{eq:str_det_Lp}) can be rewritten as
\begin{eqnarray}
&&\bigoplus_{j'}
\{Id_{B_{j'}}\}\otimes {\rm su}({\rm dim}\mathcal H_{R_{j'}}\cdot {\rm dim}\mathcal H_S\cdot {\rm dim}\mathcal H_{S'})
\nonumber\\
&& \subseteq L^\prime \subseteq 
\bigoplus_{j^\prime}
{\mathcal L}(L_d\otimes \{Id_{S'}\}\cup \{Id_{B_{j'}}\}\otimes {\rm su}({\rm dim}\mathcal H_{R_{j'}}\cdot {\rm dim}\mathcal H_S\cdot {\rm dim}\mathcal H_{S'})),
\end{eqnarray}
where $L_d=\bigoplus_{j'} {\rm u}({\rm dim}\mathcal H_{B_{j'}})\otimes \{Id_{R_{j'}}\otimes Id_S\}$.
This relation then implies, according to Lemma \ref{sec:abstract_structure_lemm10}, that the disconnected algebra $L'_d$ and the connected algebra $L_c'$ for $L'$
are
\begin{eqnarray}
L_d' &=&
L_d\otimes \{ Id_{S'}\},
\\
L_c' &=&
\bigoplus_{j'} \{Id_{B_{j'}}\}\otimes {\rm su}({\rm dim}\mathcal H_{R_{j'}}\cdot {\rm dim}\mathcal H_S\cdot {\rm dim}\mathcal H_{S'}).
\end{eqnarray}
Hence, Theorem \ref{th:str_rel} is justified.

\section{Conclusion}
We have revealed the structures of the Hilbert space and the Lie algebra from only a few very simple assumptions, in the context of indirect quantum control.
The restrictedness of our artificial operations imposes constraints on what can be controlled in the large entire system. An interesting finding includes that there is a clear distinction depending on the dimension of the directly accessible subsystem $S$  (Theorems 2 and 3). While $E$, which only interacts with $S$ through the drift Hamiltonian $h_0$, is virtually a direct sum of fully controllable subspaces, not all operations are necessarily possible when dim$\mathcal{H}_S=2$.

There have been studies \cite{AAR15,AR12} in a similar direction, which have analyzed the 'controllability' issue depending on dim$\mathcal{H}_S$.  Though there are  differences in meaning of  some terms, e.g. controllability, our analysis can be used to prove their results as well; the details are given in the supplementary material. 

The present analysis can be applied to the study of physical situations where we wish to control a large quantum system with minimal access.
Such scenarios have been discussed under the motivation of suppressing unnecessary interactions between the quantum system and its environment. 
As briefly mentioned after Theorem \ref{th:str_rel},  control problems have been addressed in \cite{BMM10} for a one-dimensional XX spin chain
through direct control of two end spins. Also, closely related is the problem of quantum system identification under  limited access,
which has been discussed intensively in the last decade \cite{BMN09,FPK09,BM09,BMN11,KY14}. 
From the system identification perspective, in which the main task is to identify
the drift Hamiltonian $h_0$, what we have clarified in this paper can be understood as the very fundamental structure of what 
we may be able to identify through $S$, regardless of the physical system. 

The structures of the space and the algebra we have clarified can be used to further investigate the possibility of indirect control of large systems.
In this context, for example, a significant consequence of indirect control is the existence of equivalence classes, 
within which any distinct physical configurations of $E$ and its Hamiltonians cannot be distinguished by any operations on $S$. 
While it has already been studied in the literature, such as \cite{BY12} and \cite{OMTK15}, our results would shed more light on this issue 
in a consistent way. 

There should still be a lot of ground to explore  in front of us. One practically important issue we have not discussed here is the 
time optimality or time dependence of the operation on the system size. This problem has been studied quite actively 
(See, e.g., some recent studies \cite{TBG12,TGLS16,GLST17,LARR18,HC18} and references therein).
In addition, we still have very little insight into how to obtain the specific profile of the control pulses \cite{DD09}. It appears, however, that it is likely that
we have to rely on numerical optimization methods for it.

Despite all this, the framework of indirect control under limited access is promising for realistic large-scale quantum control.
Our attempt would be of use to acquire deeper insights into the physics of quantum control systematically and will hopefully be one of the guiding principles in building the future quantum control methodology.

\section*{Acknowledgments}
We acknowledge financial support by the JSPS Kakenhi (C) No. 26400400 and No. 17K05591. MO also thanks for support by the JSPS Kakenhi (C) No. 16K00014.

\appendix
\section{Supplementary material: Proofs of lemmas}
Here, we prove those lemmas used for proving theorems in the main text. The proofs of some lemmas below are rather involved, so this whole 
section may be browsed quickly or even skipped if readers' interest is in grasping the picture of what our main theorems claim. 
However, it should be interesting to see how the Jordan algebra, which might not be particularly common among quantum physicists 
despite its origin, plays a central role in the study of indirect quantum control. 

The first lemma shows the fundamental structure of the Lie algebra of our principal interest.
\begin{lemma}
\label{sec:str_L}
{\rm (Proved in \cite{AAR15})}
Let $L$ be the Lie algebra of skew-Hermitian operators acting on $\mathcal{H}_E\otimes \mathcal{H}_S$, which contains all elements in $\{Id_E\}\otimes {\rm su}({\rm dim}\mathcal H_S)$. Then $L$ can be written in the form
\begin{equation}
{\mathcal L}(G^{(0)}\otimes \{Id_S\}\cup iG^{(1)}\otimes  {\rm su}({\rm dim}\mathcal H_S))
\end{equation}
with appropriate linear spaces  $G^{(0)}$ and $G^{(1)}$ of skew-Hermitian operators acting on $\mathcal H_E$.
\end{lemma}

\begin{proof}
Using the basis $H_S=\{X_{k,q}, Y_{k,q}, Z_{k,k+1}\}$ of the linear space $i\cdot {\rm su}({\rm dim}\mathcal H_S)$, any operator $g\in L$ can be uniquely written in the form 
\begin{equation}
g = g_{Id}\otimes Id_S+\sum_{W \in H_S}g_W\otimes W,
\end{equation}
where $g_{Id}$ and $g_W$ are skew-Hermitian operators on $\mathcal{H}_E$. Let $G^{(0)}$ and $G^{(1)}$ be sets of these operator components:
\begin{eqnarray}
G^{(0)}&:=&\{g_{Id}\}_{g\in L},\\
G^{(1)}&:=&\mathcal L(\{g_W\}_{{g}\in L,W\in  H_S}).
\end{eqnarray}
This definition indicates $L\subseteq  {\mathcal L}(G^{(0)}\otimes \{Id_S\}\cup iG^{(1)}\otimes  {\rm su}({\rm dim}\mathcal H_S))=:L_0$.
Note that the set $G^{(0)} $ is a linear space since the set $L$ is a linear space.

Now, we show the inclusion of the opposite direction $L\supseteq  L_0$, i.e., for any element $g\in L$,
$g_{Id}\otimes Id$ and  $g_W\otimes h$ are in $L$ for arbitrary $h\in i\cdot {\rm su}({\rm dim}\mathcal H_S)$ and $W\in H_S$. 
To this end, we show the following: 
\begin{equation}\label{lemm1:suff_cond}
\forall g\in L, \forall W\in H_S, \exists h^\prime\in \mathrm{su}(\mathrm{dim}\mathcal{H}_S), \;\mathrm{such\; that} \; g_W\otimes h^\prime \in L.
\end{equation}
If this is fulfilled, $g_W\otimes h$ is in $L$ for any elements $h\in i\cdot {\rm su}({\rm dim}\mathcal H_S)$ since 
${\rm su}({\rm dim}\mathcal H_S)$ is a simple algebra. That is, for any nonzero $h^\prime\in \mathrm{su}(\mathrm{dim}\mathcal{H}_S)$, generators obtained by repeatedly taking commutators with elements $g_m$ in $\mathrm{su}(\mathrm{dim}\mathcal{H}_S)$ will span the whole su(dim$\mathcal{H}_S$):
\begin{eqnarray}
 {\rm su}({\rm dim}\mathcal H_S)&=& i\mathcal L(\{[\cdots[[h',g_1],g_2],\cdots g_n]|
n\in{\mathbb Z}_{\geq 1}\land g_m\in  {\rm su}({\rm dim}\mathcal H_S)\}).
\nonumber\\
\label{eq:extension}
\end{eqnarray}
Since the relation $g_{Id}\otimes Id_S=g-\sum_{W \in H_S}g_W\otimes W$ guarantees that $g_{Id}\otimes Id_S$ is in $L$,
showing Eq. (\ref{lemm1:suff_cond}) is sufficient to prove $L\supseteq  L_0$.

We pick an arbitrary element $g=:g_{Id}\otimes Id_S+\sum_{W \in H_S}g_W\otimes W\in L$.
From the relation
\begin{eqnarray}\label{pickingXY}
&&\frac 1{({\rm dim}\mathcal H_S-1)^2}[[g,\sum_{p\neq k }iId_E\otimes Z_{k,p}],\sum_{p\neq q }iId_E\otimes Z_{p,q}]
\nonumber\\&&\makebox[3cm]{}=
 g_{X_{k,q}}\otimes X_{k,q}+g_{Y_{k,q}}\otimes Y_{k,q}=: g_{k,q}',
\end{eqnarray}
$ g_{k,q}'$ is in $L$ since all operators in the LHS of Eq. (\ref{pickingXY}) are in $L$.
Therefore, a linear combination $g$ and $ g'_{k,q}$,
\begin{equation}
g-\sum_{k=0}^{{\rm dim}\mathcal H_S-2}\sum_{q=k+1}^{{\rm dim}\mathcal H_S-1} g_{k,q}'
= g_{Id}\otimes Id+\sum_{k=0}^{{\rm dim}\mathcal H_S-2}g_{Z_{k,k+1}}\otimes Z_{k,k+1}=: g'
\end{equation}
is also in $L$.
Taking commutators between $g_{k,q}^\prime$, $g'$ and generators in $\{Id_E\}\otimes \mathrm{su}(\mathrm{dim}\mathcal{H}_S)\in L$, we can obtain $g_{X_{k,q}}\otimes Z_{k,q}$, $g_{Y_{k,q}}\otimes Z_{k,q}$ and $g_{Z_{k,k+1}}\otimes X_{k,k+1}$, as follows:\begin{eqnarray}
-\frac 12[ g_{k,q}',iId_E\otimes Y_{k,q}]&=&g_{X_{k,q}}\otimes Z_{k,q},\label{gXtimesZ} \\
\frac 12[ g_{k,q}',iId_E\otimes X_{k,q}]&=&g_{Y_{k,q}}\otimes Z_{k,q}, \label{gYtimesZ} \\
\sum_{q=0}^{{\rm dim}\mathcal H_S-2}\bar\mu_{k,q}[ g',iId_E\otimes Y_{q,q+1}]
&=&\sum_{q,p=0}^{{\rm dim}\mathcal H_S-2}\bar \mu_{k,q}\mu_{q,p}g_{Z_{p,p+1}}\otimes X_{p,p+1} \nonumber \\
&=&g_{Z_{k,k+1}}\otimes X_{k,k+1}, \label{gZtimesX} 
\end{eqnarray}
 where $\bar\mu_{k,q}$ is the $(k,q)$-th element of the inverse of $({\rm dim}\mathcal H_S-1)$-dimensional matrix $M$ whose $(k, q)$-th element is $\mu_{k,q}:=2\delta_{k,q}-\delta_{|k-q|,1}$, where $0\leq k,q<{\rm dim}\mathcal H_S-1$.
The existence of the inverse matrix is guaranteed from $\det M={\rm dim}\mathcal H_S+1$.
Equations (\ref{gXtimesZ})-(\ref{gZtimesX}) mean that the condition (\ref{lemm1:suff_cond}) is satisfied, and hence $L\supseteq L_0$.
\end{proof}

Next, we consider a sufficient condition for a pair of sets $G^{(0)}$ and $G^{(1)}$ to be the identifier of the Lie algebra.
\begin{lemma}
\label{sec:prop_Gs}
If  $L={\mathcal L}(G^{(0)}\otimes \{Id_S\}\cup  iG^{(1)}\otimes  {\rm su}({\rm dim}\mathcal H_S))$ is a  Lie algebra, $G^{(0)}$ and $G^{(1)}$ satisfy
Eqs. (\ref{eq:rel_Gs_1})-(\ref{eq:rel_Gs_3}). If dim$\mathcal{H}_S\geq 3$, then another commutation relation $[G^{(1)},G^{(1)}]\subseteq G^{(1)}$ is also required.
\end{lemma}

\begin{proof}
For any $g_{b},g_{b}^{\prime}\in G^{(b)}$, we can construct equalities
\begin{eqnarray}
\left[g_{0},g_{0}^{\prime}\right]\otimes Id_{S} &=& \left[g_{0}\otimes Id_{S},g_{0}^{\prime}\otimes Id_{S}\right],
\\
\left[g_{1},g_{1}^{\prime}\right]\otimes Id_{S} &=& \frac{1}{d-1}\sum_{k=0}^{{\rm dim}\mathcal H_S-2}\sum_{q=k+1}^{{\rm dim}\mathcal H_S-1}\left[g_{1}\otimes X_{k,q},g_{1}^{\prime}\otimes X_{k,q}\right],
\\
\left[g_{0},g_{1}\right]\otimes Z_{0,1} &=& \left[g_{0}\otimes Id_{S},g_{1}\otimes Z_{0,1}\right], \; \mathrm{and}
\\
i\left\{ g_{1},g_{1}'\right\} \otimes Z_{0,1}   &=& \left[g_{1}\otimes X_{0,1},g_{1}^{\prime}\otimes Y_{0,1}\right].
\end{eqnarray}
From the assumption, any operator in the RHSs, e.g., $g_{0}\otimes Id_{S}$ and $g_{0}^{\prime}\otimes Id_{S}$, is contained in $L$. Therefore, each operator in the LHSs should also be contained in $L$. Looking at the operator on $\mathcal{H}_S$ of these relations, Eqs. (\ref{eq:rel_Gs_1})-(\ref{eq:rel_Gs_3}) can be justified.

When ${\rm dim}\mathcal H_S\geq 3$, we can have equalities such as
\begin{eqnarray}
\left[g_{1},g_{1}^{\prime}\right]\otimes Z_{1,2}    &=& \left[g_{1}\otimes X_{0,1},g_{1}^{\prime}\otimes X_{0,1}\right]-\left[g_{1}\otimes X_{0,2},g_{1}^{\prime}\otimes X_{0,2}\right],
\end{eqnarray}
which means $[G^{(1)},G^{(1)}]\subseteq G^{(1)}$. Note that if dim$\mathcal{H}_S=2$ there is only a single $X$ operator, $X_{0,1}$ (obviously the same for $Y$ and $Z$), thus the commutation relation for $G^{(1)}$ does not necessarily hold.
\end{proof}

The next lemma is for the necessary condition for a pair of sets $G^{(0)}$ and $G^{(1)}$ to be the identifier of the Lie algebra.
\begin{lemma}
\label{sec:prop_Gs_suf}
Suppose that $G^{(0)}$ and $G^{(1)}$ are sets of linear operators, and $L={\mathcal L}(G^{(0)}\otimes \{Id_S\}\cup i G^{(1)}\otimes  {\rm su}({\rm dim}\mathcal H_S))$.
If $G^{(0)}$ and $G^{(1)}$ satisfy Eqs. (\ref{eq:rel_Gs_1})-(\ref{eq:rel_Gs_3}) and if dim$\mathcal H_S=2$,
the operator space $L$ is closed under the commutator, hence $L$ forms a Lie algebra. The same can be said for the case of dim$\mathcal{H}_S\ge 3$, if Eq. (\ref{eq:g1_comm}), $[G^{(1)},G^{(1)}]\subseteq G^{(1)}$, is satisfied in addition to Eqs. (\ref{eq:rel_Gs_1})-(\ref{eq:rel_Gs_3}).
\end{lemma}

\begin{proof}
Let us define a basis of $L$ by a set of operators, each of which has the form $g_{0}\otimes Id_S$ or $i g_{1}\otimes h$ 
with $g_b\in G^{(b)}$ and $h\in H_S$. Therefore, it is sufficient if we check that
commutators between any two elements of such are in $L$.
For any $g_b,g_b^{\prime}\in G^{(b)}$ and $h,h'\in H_S$, we have the commutation relations,
\begin{eqnarray}
\left[g_{0}\otimes Id_{S},g_{0}^{\prime}\otimes Id_{S}\right]  &=& \left[g_{0},g_{0}^{\prime}\right]\otimes Id_{S}, \label{eq:g0g0} \\
\left[g_{0}\otimes Id_{S},g_{1}\otimes h\right] &=& \left[g_{0},g_{1}\right]\otimes h,\label{eq:g0g1} \\
\left[g_{1}\otimes h,g_{1}^{\prime}\otimes h'\right]    &=&
    \frac{1}{2}\left[g_{1},g_{1}^{\prime}\right]\otimes\left\{ h,h^{\prime}\right\} -\frac{1}{2}i\left\{ g_{1},g_{1}^{\prime}\right\} \otimes i\left[h,h^{\prime}\right]. \label{eq:g1g1}
\;\;
\end{eqnarray}
Due to Eqs. (\ref{eq:rel_Gs_1}) and (\ref{eq:rel_Gs_2}), the RHSs of Eqs. (\ref{eq:g0g0}) and (\ref{eq:g0g1}) are in $L$. As for Eq. (\ref{eq:g1g1}), when ${\rm dim}\mathcal H_S=2$, Eqs. (\ref{eq:rel_Gs_1}) and (\ref{eq:rel_Gs_3}) guarantee that its RHS is in $L$, since $\{h,h'\}\propto Id_S$ holds for any basis elements $h, h^\prime \in H_S$ in this case. 
If dim$\mathcal{H}_S\ge 3$, $\{h, h^\prime\}$ can be written as a linear combination of elements in $H_S$ and $Id_S$, and obviously $i[h, h^\prime]$ is again in $H_S$ (if not zero). Thus, the RHS of Eq. (\ref{eq:g1g1}) is also in $L$ because of Eqs. (\ref{eq:rel_Gs_1}), (\ref{eq:rel_Gs_3}), and (\ref{eq:g1_comm}).
\end{proof}

If the pair of operator sets $(G^{(0)},G^{(1)})$ is the identifier of a Lie algebra $L$, $iG^{(1)}$ is a set of  Hermitian operators which is closed under the anti-commutator. That is,  $iG^{(1)}$ is a {\it formally real Jordan algebra}, which is defined as a linear space closed under the commutative bilinear operator such that
\begin{eqnarray}
\{x,y\} &=& \{y,x\}, \nonumber \\
\{\{\{x,x\},y\},x\} &=& \{\{x,x\},\{y,x\}\}, \nonumber \\
\sum_j\{x_j,x_j\}=0&\Rightarrow&x_j=0. \nonumber
\end{eqnarray}
The following lemmas about the structure of the Jordan algebra are useful for classification of Lie algebras that include all elements in $i\{Id_E\}\otimes {\rm su}({\rm dim}\mathcal H_S)$.

\begin{lemma}
\label{sec:Jor_alg_rel}
{\rm (Theorem 14, 16 and 17 in the paper \cite{JNW34})}
For any formally real Jordan algebra $J$, a basis 
$ \{e_{\rho} \}_{\rho\in\{0,1,\cdots, \rho_0-1\}}\cup\{s^{(\rho,\sigma)}_{\mu}\}_{(\rho,\sigma,\mu)\in \Omega}$ can be constructed, where $\rho_0$ is an integer that can be determined when a specific $J$ is given. The indices $\rho, \sigma,$ and $\mu$ are in the range 
$
 \Omega=\{(\rho,\sigma,\mu)|
\rho<\sigma
\;\land\;
 \exists j,\;\rho,\sigma\in\Gamma_j
\;\land\;
\mu\in\{0,1,\cdots,\chi_j-1\}
 \}
$,  where $\{\Gamma_j\}_j$ are a non-overlapping decomposition of 
$\{0,1,\cdots,\rho_0-1 \}$, i.e., $\bigoplus_j \Gamma_j= \{0,1,\cdots,\rho_0-1 \}$, and $\chi_j$ are positive integers indexed by $j$. The basis elements $ \{e_{\rho} \}$ and $\{s^{(\rho,\sigma)}_{\mu}\}$ satisfy the following three anti-commutation relations:
\begin{eqnarray}
\{e_\rho,e_\sigma\} &=&  2\delta_{\rho,\sigma}e_\rho,
\label{eq:str_jor_2_1}
\\
\{s_\mu^{(\rho,\sigma)},s_\nu^{(\rho,\sigma)}\} &=& 2\delta_{\mu,\nu}(e_\rho+e_\sigma),
\label{eq:str_jor_2_2}
\\
\{e_\rho,s_\mu^{(\sigma,\tau)}\}    &=& (\delta_{\rho,\sigma}+\delta_{\rho,\tau})s_\mu^{(\sigma,\tau)}.
\label{eq:str_jor_2_3}
\end{eqnarray}
\end{lemma}

As a quick consequence of Eqs. (\ref{eq:str_jor_2_1})-(\ref{eq:str_jor_2_3}) in this lemma, let us show three useful relations. The first one is
\begin{eqnarray}
e_{\rho}e_{\sigma}
&=& \frac 12 \{e_{\rho},e_{\rho}\}e_{\sigma}
\nonumber\\
    &=&  e_{\rho}\{e_{\rho},e_{\sigma}\}-\frac{1}{2}\{\{e_{\rho},e_{\sigma}\},e_{\rho}\}+\frac{1}{4}\{\{e_{\rho},e_{\rho}\},e_{\sigma}\}
\nonumber\\
    &=&  2\delta_{\rho,\sigma}e_{\rho}^{2}-2\delta_{\rho,\sigma}e_{\rho}+\delta_{\rho,\sigma}{e_{\rho}}
\nonumber\\
    &=&  \delta_{\rho,\sigma}(2\{e_{\rho},e_{\rho}\}-e_{\rho})
\nonumber\\
    &=&  \delta_{\rho,\sigma}e_{\rho},
\label{eq:str_jor_2_1+}
\end{eqnarray}
where Eq. (\ref{eq:str_jor_2_1}) is used in the first, third and the last equalities, while the second and the fourth equalities can be verified just by the definition of the anti-commutator.
The second is, for $\rho<\sigma$,
\begin{eqnarray}
s_{\mu}^{(\rho,\sigma)}
&=& \{e_\rho,\{e_\sigma,s_\mu^{(\rho,\sigma)}\}\}
\nonumber\\
&=& e_{\rho}s_{\mu}^{(\rho,\sigma)}e_{\sigma}+e_{\sigma}s_{\mu}^{(\rho,\sigma)}e_{\rho},
\label{eq:str_jor_2_2+}
\end{eqnarray}
where Eq. (\ref{eq:str_jor_2_3}) has been used recursively in the first equality, and Eq. (\ref{eq:str_jor_2_1+}) in the second line. The last one we show here is
\begin{eqnarray}
 &   &  e_{\rho}s_{\mu}^{(\rho,\sigma)}e_{\sigma}s_{\nu}^{(\rho,\sigma)}e_{\rho}+e_{\rho}s_{\nu}^{(\rho,\sigma)}e_{\sigma}s_{\mu}^{(\rho,\sigma)}e_{\rho}
\nonumber\\
    &=& e_{\rho}\{e_{\rho}s_{\mu}^{(\rho,\sigma)}e_{\sigma}+e_{\sigma}s_{\mu}^{(\rho,\sigma)}e_{\rho},e_{\rho}s_{\nu}^{(\rho,\sigma)}e_{\sigma}+e_{\sigma}s_{\nu}^{(\rho,\sigma)}e_{\rho}\}e_{\rho}
\nonumber\\
    &=& e_{\rho}\{s_{\mu}^{(\rho,\sigma)},s_{\nu}^{(\rho,\sigma)}\}e_{\rho}
\nonumber\\
    &=& 2\delta_{\mu,\nu}e_{\rho},
\label{eq:str_jor_2_3+}
\end{eqnarray}
where Eqs. (\ref{eq:str_jor_2_1+}) and (\ref{eq:str_jor_2_2+}) have been applied in the first and the second equalities. In the last step, Eqs. (\ref{eq:str_jor_2_2}) and (\ref{eq:str_jor_2_1+}) are used.

Although Lemma \ref{sec:Jor_alg_rel} has been known since \cite{JNW34}, we shall give a proof of these relations in order to make this paper self-contained. As we are interested in the (Jordan) algebra, which is expressed on a Hermitian operator space, our discussion is automatically restricted to the formally real Jordan algebra.

\begin{proof}
Let us start with a simple proposition about the Jordan algebra with Hermitian operators. That is, for any element $h$ of the algebra, projection operators onto the eigenspace of $h$ for any non-zero eigenvalue are in the algebra.
To prove this, pick an arbitrary element $h$ in the Jordan algebra and let $v_{k}$ and $h_{k}$ for $k\in \{1,2,\cdots,n_{0}\}$ be its non-zero eigenvalues and projection operators onto the corresponding eigenspaces.
Then, define a matrix $M$ such that its $(q, k)$-element $\mu_{q,k}$ is equal to $(v_k)^{q-1}$ for $k,q\in\{1,2,\cdots,n_{0}\}$. Similarly, we define $\bar \mu_{q,k}$ as the $(q, k)$-element of the inverse $M^{-1}$. The existence of the inverse matrix is guaranteed by the fact that ${\rm det}M=\prod_{k>k^{\prime}}(v_{k}-v_{k^{\prime}})\neq 0$.
Using $\mu_{q,k}$ and $\bar{\mu}_{q,k}$, the projection operator $h_k$ can be written as
\begin{equation}
h_{k} = \sum_{q,k'}\bar \mu_{k,q}\mu_{q,k'}h_{k'}
=\sum_{q}\bar \mu_{k,q}h^{q-1}.
\end{equation}
Because $\{h^{n},h\}=2h^{ n+1}$ for all $n\in\mathbb{N}_{>0}$, the above equation implies that the projection operator onto an eigenspace $h_k$ is also an element in the algebra. 

Next, we shall define a set $\{e_\rho\}_\rho$ in the following way and look into its properties.
First, let a set $J^{(0)}$ be $J$. For $\rho\geq 0$, $e_\rho$  is defined from a subset $J^{(\rho)}$ of $J$ such that $e_\rho$ is a non-zero operator which has the smallest rank in the set $J^{(\rho)}$ whose largest eigenvalue is $1$.
Then, $J^{(\rho+1)}$ is defined as a set of elements of $J^{(\rho)}$ that anti-commute with $e_\rho$, i.e., $J^{(\rho+1)}=\{h|\{h,e_{\rho}\}=0,h\in J^{(\rho)}\}$.
As we have seen in the above argument, $e_0$ is a projection operator, and, for any element $h \in J^{(1)}$, $he_0=\frac12(\{e_0,h\}+\{e_0,h\}e_0- e_0\{e_0,h\})=0$ holds. Thus, for any element $h$, we see $h'\in J^{(1)}$, $\{\{h,h'\},e_0\}=\{h,\{h',e_0\}\}+\{h',\{h,e_0\}\}+2(he_0)h'+2(h'e_0)h=0$, and this means that $J^{(1)}$ is also a Jordan subalgebra of $J$.
Iterating this process for larger $\rho$, we can state that $e_\rho$ is a projection operator, any element in $J^{(\rho')}$
anti-commutes with $e_\rho$ if $\rho^\prime>\rho$, thus $\{e_\rho,e_{\rho'}\}=0$ as $e_{\rho^\prime} \in J^{(\rho')}$.
Therefore, Eq. (\ref{eq:str_jor_2_1}), as well as Eq. (\ref{eq:str_jor_2_1+}), can be justified. The sets $\{J^{(\rho)}\}$ clearly have the inclusion relations $J^{(0)}\supsetneq J^{(1)} \supsetneq \cdots$. Since the entire space is finite dimensional, there exists a number $\rho_0$ such that $\cdots J^{(\rho_0-1)}\supsetneq J^{(\rho_0)}=\{0\}$.

In order to show other properties of $\{e_\rho\}$ and $\{s_\mu^{(\rho,\sigma)}\}$, let us now define linear spaces $E_{\rho}:=\{e_{\rho}he_{\rho}| h\in J\}$ and $S^{(\rho,\sigma)}:=\{e_{\sigma} h e_{\rho}+e_{\rho} h e_{\sigma}| h\in J\}$ for $\rho\neq \sigma$.
Any elements in  $E_{\rho}$ and  $S^{(\rho,\sigma)}$ are in $J$ since, using Eq. (\ref{eq:str_jor_2_1+}), their elements can be written as
\begin{eqnarray}\label{EandS}
e_{\rho}he_{\rho} &=& \frac{1}{2}\{\{h,e_{\rho}\},e_{\rho}\}-\frac{1}{2}\{h,e_{\rho}\} \\
e_{\sigma}he_{\rho}+e_{\rho}he_{\sigma} &=& \{\{h,e_{\rho}\},e_{\sigma}\}.
\end{eqnarray}

An immediate consequence of this definition of $E_\rho$ is that any element in $E_{\rho}$ is proportional to $e_\rho$.
Let us prove it by contradiction. Suppose that there exists an operator $h$ in $E_{\rho}$ which is not proportional to $e_\rho$.
We pick a projection operator $e_{h}$ onto an eigenspace of $h$ that corresponds to a nonzero eigenvalue.
Since  $h\in E_{\rho}$, $e_\rho h e_\rho=h$ holds due to Eq. (\ref{eq:str_jor_2_1+}), thus the range of $h$ is not larger than $e_\rho$.
Because the range of $e_h$ is smaller than $h$ by assumption, the range of $e_h$ is smaller than $e_\rho$.
This implies $e_\rho e_h e_\rho=e_h$ and $\mathrm{rank}\, e_h < \mathrm{rank}\,e_\rho$.
Meanwhile, $e_{h}$ is not only in $J$ (since $h\in J$), but also in $J^{(\rho)}$ because $\{e_{\rho'},e_h\}=e_{\rho'}e_\rho e_h e_\rho+e_\rho e_h e_\rho e_{\rho'}=0$ for $\rho^\prime\neq \rho$. The existence of a projection $e_{h}\in J^{(\rho)}$ such that $\mathrm{rank}\, e_h < \mathrm{rank}\,e_\rho$ contradicts with the definition of $e_\rho$, that is, $e_\rho$ must have the smallest rank in $J^{(\rho)}$.

Let us now prove Eqs. (\ref{eq:str_jor_2_2}) and (\ref{eq:str_jor_2_3}), the equalities concerning $s_\mu^{(\rho,\sigma)}$ in the set $S^{(\rho,\sigma)}$. We shall consider only the case where all elements of $S^{(\rho,\sigma)}$ are nonzero. By defining an inner product $f_{\rho,\sigma}$ on the linear space $S^{(\rho,\sigma)}$ as 
\begin{equation}
f_{\rho,\sigma}(h,h') = \frac{{\rm Tr}hh'}{{\rm Tr}(e_\rho+e_\sigma)},
\end{equation}
we can construct a normalized orthogonal basis with respect to $f_{\rho,\sigma}$, and we shall let $s_\mu^{(\rho,\sigma)}$ be such a basis.

Equation (\ref{eq:str_jor_2_3}), thus Eq. (\ref{eq:str_jor_2_2+}) as well, can be verified rather straightforwardly with the definition of $S^{(\rho,\sigma)}$ and Eq. (\ref{eq:str_jor_2_1+}). In order to prove Eq. (\ref{eq:str_jor_2_2}), let us note that $e_\tau\{s^{(\rho,\sigma)}_\mu,s^{(\rho,\sigma)}_\nu \}e_\tau=a^{(\rho,\sigma,\tau)}_{\mu,\nu}e_\tau$, where $a_{\mu,\nu}^{(\rho,\sigma,\tau)}$ is a real number. This is because $e_\tau\{s^{(\rho,\sigma)}_\mu,s^{(\rho,\sigma)}_\nu \}e_\tau$ is in the set $E_\tau$, thus it is proportional to $e_\tau$ and we let $a^{(\rho,\sigma,\tau)}_{\mu,\nu}$ denote the proportionality constant.
Then, together with Eqs. (\ref{eq:str_jor_2_1+}) and (\ref{eq:str_jor_2_2+}), we can see the following relations hold:
\begin{eqnarray}
a^{(\rho,\sigma,\rho)}_{\mu,\nu}e_\rho&=&e_\rho \{s^{(\rho,\sigma)}_\mu,s^{(\rho,\sigma)}_\nu \}
 e_\rho
\nonumber\\
&=&
(e_\rho s^{(\rho,\sigma)}_\mu e_\sigma)
(e_\rho s^{(\rho,\sigma)}_\nu e_\sigma)^\dagger
+
(e_\rho s^{(\rho,\sigma)}_\nu e_\sigma)
(e_\rho s^{(\rho,\sigma)}_\mu e_\sigma)^\dagger ,
\nonumber\\
\label{eq:tmp_1_1}
a^{(\rho,\sigma,\sigma)}_{\mu,\nu}e_\sigma&=&e_\sigma \{s^{(\rho,\sigma)}_\mu,s^{(\rho,\sigma)}_\nu \}
 e_\sigma
\nonumber\\&=&
(e_\rho s^{(\rho,\sigma)}_\nu e_\sigma)^\dagger
(e_\rho s^{(\rho,\sigma)}_\mu e_\sigma)
+
(e_\rho s^{(\rho,\sigma)}_\mu e_\sigma)^\dagger
(e_\rho s^{(\rho,\sigma)}_\nu e_\sigma),
\nonumber\\
\label{eq:tmp_1_2}
a^{(\rho,\sigma,\rho)}_{\mu,\nu}e_\rho
+
a^{(\rho,\sigma,\sigma)}_{\mu,\nu}e_\sigma
&=&
\{s^{(\rho,\sigma)}_\mu,s^{(\rho,\sigma)}_\nu \}.
\label{eq:tmp_1_3}
\end{eqnarray}
Complex conjugates are included in the first two equations simply for the convenience for the next step, recalling that the algebra $J$ consists of only Hermitian operators. By setting $\nu=\mu$ in Eqs. (\ref{eq:tmp_1_1}) and (\ref{eq:tmp_1_2}), we have 
\begin{eqnarray}
a^{(\rho,\sigma,\rho)}_{\mu,\mu}e_\rho
&=&
2(e_\rho s^{(\rho,\sigma)}_\mu e_\sigma)
(e_\rho s^{(\rho,\sigma)}_\mu e_\sigma)^\dagger,
\\
a^{(\rho,\sigma,\sigma)}_{\mu,\mu}e_\sigma
&=&
2(e_\rho s^{(\rho,\sigma)}_\mu e_\sigma)^\dagger
(e_\rho s^{(\rho,\sigma)}_\mu e_\sigma).
\end{eqnarray}
Since $s^{(\rho,\sigma)}_\mu$ is a Hermitian operator and has the form (\ref{eq:str_jor_2_2+}),
the operator $e_\rho s^{(\rho,\sigma)}_\mu e_\sigma$ is not equal to zero.  This fact and the above two relations  guarantee that the rank of $e_\rho$ is equal to that of $e_\sigma$, thus
\begin{equation}\label{tre=trs}
{\rm Tr}e_\rho = {\rm Tr}e_\sigma.
\end{equation}
It is clear from the RHSs of Eqs. (\ref{eq:tmp_1_1}) and (\ref{eq:tmp_1_2}) that their traces are equal, i.e.,
\begin{equation}\label{a=a}
a^{(\rho,\sigma,\rho)}_{\mu,\nu}{\rm Tr}e_\rho = a^{(\rho,\sigma,\sigma)}_{\mu,\nu}{\rm Tr}e_\sigma.
\end{equation}
Equation (\ref{eq:tmp_1_3}) and the orthogonality of $\{s^{(\rho,\sigma)}_\mu\}$ imply
\begin{eqnarray}\label{ae_orthogonality}
a^{(\rho,\sigma,\rho)}_{\mu,\nu}{\rm Tr}e_\rho+a^{(\rho,\sigma,\sigma)}_{\mu,\nu}{\rm Tr}e_\sigma&=&
{\rm Tr} \{s^{(\rho,\sigma)}_\mu, s^{(\rho,\sigma)}_\nu\}
\nonumber \\
&=& 2{\rm Tr} s^{(\rho,\sigma)}_\mu s^{(\rho,\sigma)}_\nu \nonumber \\
&=& 2\delta_{\mu,\nu}{\rm Tr}(e_\rho+e_\sigma).
\end{eqnarray}
It then follows from Eqs. (\ref{tre=trs}), (\ref{a=a}), and (\ref{ae_orthogonality}) that $a^{(\rho,\sigma,\rho)}_{\mu,\nu}=a^{(\rho,\sigma,\sigma)}_{\mu,\nu}=2\delta_{\mu,\nu}$, hence we obtain Eq. (\ref{eq:str_jor_2_2}) from Eq. (\ref{eq:tmp_1_3}). As described above (before the proof of this lemma), Eq. (\ref{eq:str_jor_2_2}) also leads to Eq. (\ref{eq:str_jor_2_3+}).

It still has to be shown that $\{e_\rho\}_\rho\cup \{s_\mu^{(\rho,\sigma)}\}_{\rho<\sigma,\mu}$ forms a basis of $J$.
From Eqs.  (\ref{eq:str_jor_2_1+}),  (\ref{eq:str_jor_2_2+}), and that $\{s_\mu^{(\rho,\sigma)}\}_{\mu}$ are orthogonal to each other, we can check that these elements are linearly independent.
As shown above, $\{e_\rho\}$ and $\{s_\mu^{(\rho,\sigma)}\}_{\mu}$ are the bases of $E_\rho$ and $S^{(\rho,\sigma)}$, respectively, and any elements in $E_\rho$ and $S^{(\rho,\sigma)}$ are in $J$. Therefore, all we have to check is that every element in $J$ can be expressed as a linear combination of $\{e_\rho\}_\rho \cup\{s_\mu^{(\rho,\sigma)}\}_{\rho<\sigma ,\mu}$.
Let $h$ be an element in $J$ and define $I=\sum_{\rho}e_{\rho}\in J$. 
Noting Eq. (\ref{eq:str_jor_2_1+}), which implies $I^2=I$, and that $I$ and $h$ are Hermitian, we can see the following relations:
\begin{eqnarray}
\{I,h\}-IhI-h &=& \frac 32\{h,I\}-h-\frac12\{\{h,I\},I\}, \label{eq:tmp_ihi01}
\\
(hI-IhI)(hI-IhI)^\dagger &=&
-\{\{h,h\},I\}
+\frac32\{\{h,I\},h\}
+\frac14\{\{\{h,I\},\{h,I\}\},I\}
\nonumber\\
&& {}+\frac14\{\{\{h,h\},I\},I\}
-\{\{h,I\},\{h,I\}\}.\label{eq:tmp_ihi02}
\end{eqnarray}
The RHSs contain only anti-commutators of elements in $J$, thus they are in $J$.
Since $e_\rho I=I e_\rho=e_\rho$, the LHSs of Eqs. (\ref{eq:tmp_ihi01}) and (\ref{eq:tmp_ihi02}) anti-commute with $e_\rho$ for any $\rho$.
By definition of $e_\rho$, such operators should be equal to 0, i.e., $\{I,h\}-IhI-h=0$ and $hI-IhI=0$. Then, obviously
\begin{equation}
\{I,h\}-IhI-h-(hI-IhI)-(hI-IhI)^\dagger = IhI-h=0
\end{equation}
holds, that is, $IhI=h$.
Resubstituting $I=\sum_\rho e_\rho$, we see
\begin{equation}
h=IhI=\sum_{\rho}e_{\rho}he_{\rho}+\sum_{\sigma>\rho}(e_{\rho}he_{\sigma}+e_{\sigma}he_{\rho}),
\end{equation}
which means that $h$ is in the space spanned by $E_\rho$ and $S^{(\rho,\sigma)}$.

Next, we focus on the dimension $\chi_{\rho,\sigma}$ of the space $S^{(\rho,\sigma)}$, and will show that
if $\chi_{\rho,\sigma}\neq 0$, $\chi_{\rho,\tau}$ is not larger than  $\chi_{\sigma,\tau}$ for mutually distinct $\sigma$, $\rho$ and $\tau$.
Due to the symmetry with respect to the permutation of $\sigma,\rho,\tau$, this means that
if  $\chi_{\rho,\sigma}\neq 0$ and $\chi_{\rho,\tau}\neq 0$, these two and $\chi_{\sigma,\tau}$ are equal to each other. To this end, let us pick a basis  $\{s^{(\rho,\tau)}_\mu\}_\mu$ of $S^{(\rho,\tau)}$ and a normalized element  $s^{(\rho,\sigma)}_0$ in  $S^{(\rho,\sigma)}$.
Defining $s_{\mu}^{\prime(\sigma,\tau)}:=\{s_{0}^{(\rho,\sigma)},s_{\mu}^{(\rho,\tau)}\}\in J$ for $\mu\in \{0,\cdots ,\chi_{\rho,\tau}-1\}$, we see
\begin{eqnarray}\label{sprime_offdiag}
 &&     e_{\sigma}s_{\mu}^{\prime(\sigma,\tau)}e_{\tau}+e_{\tau}s_{\mu}^{\prime(\sigma,\tau)}e_{\sigma}
\nonumber\\
    &=& e_{\sigma}(s_{0}^{(\rho,\sigma)}s_{\mu}^{(\rho,\tau)}+s_{\mu}^{(\rho,\tau)}s_{0}^{(\rho,\sigma)})e_{\tau}+e_{\tau}(s_{0}^{(\rho,\sigma)}s_{\mu}^{(\rho,\tau)}+s_{\mu}^{(\rho,\tau)}s_{0}^{(\rho,\sigma)})e_{\sigma}
\nonumber\\
    &=& e_{\sigma}s_{0}^{(\rho,\sigma)}e_{\rho}s_{\mu}^{(\rho,\tau)}e_{\tau}+e_{\tau}s_{\mu}^{(\rho,\tau)}e_{\rho}s_{0}^{(\rho,\sigma)}e_{\sigma}
\nonumber\\
    &=& (e_{\sigma}s_{0}^{(\rho,\sigma)}e_{\rho}+e_{\rho}s_{0}^{(\rho,\sigma)}e_{\sigma})(e_{\rho}s_{\mu}^{(\rho,\tau)}e_{\tau}+e_{\tau}s_{\mu}^{(\rho,\tau)}e_{\rho})
+(e_{\rho}s_{\mu}^{(\rho,\tau)}e_{\tau}+e_{\tau}s_{\mu}^{(\rho,\tau)}e_{\rho})(e_{\sigma}s_{0}^{(\rho,\sigma)}e_{\rho}+e_{\rho}s_{0}^{(\rho,\sigma)}e_{\sigma})
\nonumber\\
    &=& \{s_{0}^{(\rho,\sigma)},s_{\mu}^{(\rho,\tau)}\}=s_{\mu}^{\prime(\sigma,\tau)}.
\end{eqnarray}
In the second and the fourth equalities, Eqs. (\ref{eq:str_jor_2_1+}) and (\ref{eq:str_jor_2_2+}) are used, and the third equality can be verified by expanding the RHS, using Eq. (\ref{eq:str_jor_2_1+}). Equation (\ref{sprime_offdiag}) implies that $s_{\mu}^{\prime(\sigma,\tau)}$ is an element in $S^{(\sigma,\tau)}$.
In addition, the following relation can also be verified in a similar manner:
\begin{eqnarray}
&&\{s_{\mu}^{\prime(\sigma,\tau)},s_{\nu}^{\prime(\sigma,\tau)}\}
\nonumber\\
 &=& \{\{s_{0}^{(\rho,\sigma)},s_{\mu}^{(\rho,\tau)}\},\{s_{0}^{(\rho,\sigma)},s_{\nu}^{(\rho,\tau)}\}\}
\nonumber\\
    &=&
 s_{0}^{(\rho,\sigma)}s_{\mu}^{(\rho,\tau)}s_{0}^{(\rho,\sigma)}s_{\nu}^{(\rho,\tau)}
+s_{0}^{(\rho,\sigma)}s_{\nu}^{(\rho,\tau)}s_{0}^{(\rho,\sigma)}s_{\mu}^{(\rho,\tau)}
+s_{\mu}^{(\rho,\tau)}s_{0}^{(\rho,\sigma)}s_{\nu}^{(\rho,\tau)}s_{0}^{(\rho,\sigma)}
\nonumber\\&&{}
+s_{\nu}^{(\rho,\tau)}s_{0}^{(\rho,\sigma)}s_{\mu}^{(\rho,\tau)}s_{0}^{(\rho,\sigma)}
+s_{\mu}^{(\rho,\tau)}s_{0}^{(\rho,\sigma)}s_{0}^{(\rho,\sigma)}s_{\nu}^{(\rho,\tau)}
+s_{\nu}^{(\rho,\tau)}s_{0}^{(\rho,\sigma)}s_{0}^{(\rho,\sigma)}s_{\mu}^{(\rho,\tau)}
+s_{0}^{(\rho,\sigma)}\{s_{\nu}^{(\rho,\tau)},s_{\mu}^{(\rho,\tau)}\}s_{0}^{(\rho,\sigma)}
\nonumber\\
    &=& s_{\mu}^{(\rho,\tau)}(e_{\rho}+e_{\sigma})s_{\nu}^{(\rho,\tau)}+s_{\nu}^{(\rho,\tau)}(e_{\rho}+e_{\sigma})s_{\mu}^{(\rho,\tau)}+2\delta_{\mu,\nu}s_{0}^{(\rho,\sigma)}(e_{\rho}+e_{\tau})s_{0}^{(\rho,\sigma)}
\nonumber\\
    &=& e_{\tau}s_{\mu}^{(\rho,\tau)}e_{\rho}s_{\nu}^{(\rho,\tau)}e_{\tau}+e_{\tau}s_{\nu}^{(\rho,\tau)}e_{\rho}s_{\mu}^{(\rho,\tau)}e_{\tau}+2\delta_{\mu,\nu}e_{\sigma}s_{0}^{(\rho,\sigma)}e_{\rho}s_{0}^{(\rho,\sigma)}e_{\sigma}
\nonumber\\
    &=& 2\delta_{\mu,\nu}(e_{\sigma}+e_{\tau}).
\end{eqnarray}
This means that the operators $\{s_{\mu}^{\prime(\sigma,\tau)}\}_{\mu\in\{0,1,\cdots,\chi_{\rho,\tau}-1\}}$ are all non-zero and linearly independent. Thus, the number of linearly independent $s_\mu^{\prime(\sigma,\tau)}$ is larger than or equal to $\chi_{\rho,\tau}$, i.e., $\chi_{\rho,\tau}\leq\chi_{\sigma,\tau}$. Hence, $\chi_{\rho,\sigma}=\chi_{\rho,\tau}=\chi_{\sigma,\tau}$ as mentioned above. 

The above argument, $\chi_{\rho,\sigma}=\chi_{\rho,\tau}=\chi_{\sigma,\tau}$ if  $\chi_{\rho,\sigma}\neq 0$ and $\chi_{\rho,\tau}\neq 0$,
shows that the set $\{0,1,\cdots, \rho_0-1\}$ can be decomposed into non-overlapping subsets $\Gamma_j$, i.e.,  $\{0,1,\cdots, \rho_0-1\}=\bigoplus_j\Gamma_j$. Grouping for each $\Gamma_j$ is done so that $\chi_{\rho,\sigma}\neq 0$ if and only if both $\rho$ and $\sigma$ are in a single set $\Gamma_j$. Within the same $\Gamma_j$, all $\chi_{\rho,\sigma}$ are the same, namely, $\chi_{\rho,\sigma}=\chi_{\rho^\prime,\sigma^\prime}$ for any $\rho, \sigma, \rho^\prime, \sigma^\prime \in \Gamma_j$.
To prove this statement, we define an equivalence relation $\sim$ such that the relation $\rho\sim\sigma$ holds if and only if $\chi_{\rho,\sigma}\neq0$ or $\rho=\sigma$.
The reflexivity and the symmetry relations hold trivially, and the transitivity relation is guaranteed by the above argument, where we have seen $(\rho\sim\sigma)\land (\rho\sim\tau) \Rightarrow \sigma\sim\tau$.
Noting the fact $(\rho\sim\sigma) \land (\rho\sim\tau) \Rightarrow \chi_{\sigma,\tau}=\chi_{\rho,\sigma}=\chi_{\rho,\tau}$,
we can group the indices that are connected with the equivalence relation "$\sim$" as $\{\Gamma_j\}_j$. Then, by rewriting $\chi_j:=\chi_{\rho,\sigma}$ for $\rho,\sigma\in\Gamma_j$, all properties of the basis stated in Lemma \ref{sec:Jor_alg_rel} about the formally real Jordan algebra have been derived.
\end{proof}

The relations between the basis vectors shown in Lemma \ref{sec:Jor_alg_rel} imply a very unique structure of the formally real Jordan algebra. They then allow us to obtain explicit expressions of $J$ on the space of Hermitian operators with an appropriate basis.
\begin{lemma}
\label{sec:rep_Jor}
Suppose that  $J$ is a representation of the Jordan algebra on a Hermitian-operator space that includes the identity operator, i.e. $J$ is a linear space of Hermitian operators such that $\{J,J\}\subseteq J$ and $Id\in J$.
Then, $J$ is a direct sum of simple Jordan algebras, each of which has one of the forms {\rm (\ref{eq:JordanAlg_1})-(\ref{eq:JordanAlg_6})} with an appropriate basis.
\end{lemma}
We assign the characters $\mathfrak R$, $\mathfrak M_\gamma^{(k)}$ and $\mathfrak S_n$ to the possible simple Jordan algebras, following the notations in \cite{JNW34}. 

\begin{proof}
From Lemma \ref{sec:Jor_alg_rel}, we can choose a basis  $ \{e_{\rho} \}_{\rho\in\{0,1,\cdots, \rho_0-1\}}\cup\{s^{(\rho,\sigma)}_{\mu}\}_{(\rho,\sigma,\mu)\in \Omega}$ which satisfies  Eqs. (\ref{eq:str_jor_2_1})-(\ref{eq:str_jor_2_3}).

First,  we show that $\sum_\sigma {e_\sigma}$ is equal to the identity.
Since $Id\in J$, $Id$ can be expressed as a linear combination of the basis vectors.
This fact and the relations (\ref{eq:str_jor_2_1+}) and (\ref{eq:str_jor_2_2+}) indicate that  $(\sum_\sigma {e_\sigma})Id=Id$. On the other hand, obviously $(\sum_\sigma {e_\sigma})Id=(\sum_\sigma {e_\sigma})$ also holds, thus these lead to
\begin{equation}
\sum_{\rho}e_{\rho} = Id.
\label{eq:str_jor_2_4}
\end{equation}

 From the properties  (\ref{eq:str_jor_2_1+}) and  (\ref{eq:str_jor_2_4}), we can define a basis
$\{\ket{k,\rho}\}$ of the complex linear space $\mathcal{H}_{E}$
  such that $\ket{k,\rho}$ is the $k$-th basis vector in the space projected by $e_{\rho}$, where the range of the parameter $k$ is $\{0,1, \cdots, \mathrm{rank}e_\rho-1\}$. Next, we define a $j$-dependent subspace $\mathcal{H}_{E_{j}}$
  as a space spanned by $\{\ket{k,\rho}\}_{\rho\in\Gamma_{j},k}$. Since $\oplus_{j}\Gamma_{j}=\{0,1,\cdots,\rho_0-1\}$,
the space $\mathcal{H}_{E}$ can be expressed as a direct sum of $\mathcal{H}_{E_{j}}$, i.e.
$\oplus_{j}\mathcal{H}_{E_{j}}=\mathcal{H}_{E}$.
The basis of $J$ can also be divided into subsets, each of which is characterized by $j$, that is,  
$\{e_{\rho}\}_{\rho\in \Gamma_j}\cup\{s_{\mu}^{(\rho,\sigma)}\}_{\rho<\sigma\in \Gamma_j,\mu\in\{0,1,\cdots ,\chi_j-1\}}$. 
From the relations (\ref{eq:str_jor_2_1+}) and (\ref{eq:str_jor_2_2+}),
we can check that
any range of elements in 
the subset $\{e_{\rho}\}_{\rho\in \Gamma_j}\cup\{s_{\mu}^{(\rho,\sigma)}\}_{\rho<\sigma\in \Gamma_j,\mu\in\{0,1,\cdots ,\chi_j-1\}}$ is in the space $\mathcal H_{E_j}$.
Therefore, $J$ has a direct sum structure, and
all we have to check is that a subalgebra generated by $\{e_{\rho}\}_{\rho\in \Gamma_j}\cup\{s_{\mu}^{(\rho,\sigma)}\}_{\rho<\sigma\in \Gamma_j,\mu\in\{0,1,\cdots ,\chi_j-1\}}$  has one of the structures (\ref{eq:JordanAlg_1})$\sim$(\ref{eq:JordanAlg_6}) on the space $\mathcal H_{E_j}$. In the following, we consider a certain $j$, so that we can omit the index $j$, and relabel the indices $\rho$ and $\sigma$ for simplicity such that $\Gamma_{j}=\{0,1,\cdots,\gamma_{j}-1\}$.

If $\gamma=1$, the subalgebra consists of only the projection operator $e_{0}$. Therefore, this situation corresponds to Eq. (\ref{eq:JordanAlg_1}), i.e. the corresponding simple Jordan algebra has the structure $\mathfrak R$. We now assume $\gamma\geq 2$.

Due to Eq. (\ref{tre=trs}), the range of the parameter $k$ is independent of $\rho$ in a subalgebra for a fixed $j$. This implies that the space spanned by $\{\ket{k,\rho}\}_{k,\rho}$ can be regarded to have a direct product structure, i.e., $\{\ket{k}\otimes\ket{\rho}\}_{k,\rho}$. We can show that there exists a unitary transformation $U$ that connects these two structures, such that
\begin{eqnarray}
Us_{0}^{(\rho,\rho+1)}U^{\dagger}&  =&  Id\otimes X_{\rho,\rho+1},
\label{eq:str_jor_m_01}
\\
Ue_{\rho}U^{\dagger}    &=& Id\otimes\ket{\rho}\bra{\rho}.
\label{eq:str_jor_m_02}
\end{eqnarray}
Equation (\ref{eq:str_jor_2_2+}) indicates that we can write $s_{0}^{(\rho,\rho+1)}$ as $A_{\rho}+A_{\rho}^{\dagger}$ where $A_{\rho}$ is an element in the space spanned by $\{\ket{k,\rho+1}\bra{k^{\prime},\rho}\}_{k,k^{\prime}}$.
By setting $\mu=\nu=0$ and $\sigma=\rho+1$ in Eq. (\ref{eq:str_jor_2_3}), i.e., $(s_{0}^{(\rho,\rho+1)})^2=e_{\rho}+e_{\rho+1}$, we have
\begin{equation}
A_{\rho}^{\dagger}A_{\rho}=e_{\rho}, \;\;\mbox{and} \;\;
A_{\rho}A_{\rho}^{\dagger}=e_{\rho+1}.
\end{equation}
Therefore, we can define $U$ as $U :=  \sum_{\rho=0}^{\gamma-1}Id\otimes\ket{\rho}\bra{0}A_{0}^{\dagger}\cdots A_{\rho-2}^{\dagger}A_{\rho-1}^{\dagger}$, with which we can derive Eqs. (\ref{eq:str_jor_m_01}) and (\ref{eq:str_jor_m_02}).

Now we focus on the structure of $\{s_{\mu}^{(0,1)}\}_{\mu\in\{0,1,\cdots,\chi-1\}}$. We will show that an isometry $U$ can be constructed such that, in addition to Eqs. (\ref{eq:str_jor_m_01}) and (\ref{eq:str_jor_m_02}), the following are satisfied:
\begin{eqnarray}
 Us_{0}^{(\rho,\rho+1)}U^{\dagger}  &=& Id\otimes\overbrace{Id\otimes\cdots\otimes Id}^{\lfloor\frac{\chi-1}{2}\rfloor}\otimes X_{\rho,\rho+1},
\label{eq:str_jor_m_03}\\
Us_{2n'-1}^{(0,1)}U^{\dagger}   &=& Id\otimes\overbrace{Id\otimes\cdots\otimes Id}^{\lfloor\frac{\chi-1}{2}\rfloor-n'}\otimes Z\otimes\overbrace{Y\otimes\cdots\otimes Y}^{n'-1}\otimes Y_{0,1} \;\;\;
{\rm for}\ \lfloor\frac{\chi-1}{2}\rfloor\geq n'\geq1,
\label{eq:str_jor_m_04}\\
Us_{2n'}^{(0,1)}U^{\dagger} &=& Id\otimes\overbrace{Id\otimes\cdots\otimes Id}^{\lfloor\frac{\chi-1}{2}\rfloor-n'}\otimes X\otimes\overbrace{Y\otimes\cdots\otimes Y}^{n'-1}\otimes Y_{0,1} \;\;\;
{\rm for}\ \lfloor\frac{\chi-1}{2}\rfloor\geq n'\geq1,
\label{eq:str_jor_m_05}\\
Us_{\chi-1}^{(0,1)}U^{\dagger}  &=& Z^{*}\otimes\overbrace{Y\otimes\cdots\otimes Y}^{\lfloor\frac{\chi-1}{2}\rfloor}\otimes Y_{0,1}\;\;\mbox{when $\chi$ is even},
\label{eq:str_jor_m_06}\\
Ue_{\rho}U^{\dagger}    &=& Id\otimes\overbrace{Id\otimes\cdots\otimes Id}^{\lfloor\frac{\chi-1}{2}\rfloor}\otimes\ket{\rho}\bra{\rho},
\label{eq:str_jor_m_07}
\end{eqnarray}
 where $Z^{*}$ is a Hermitian matrix whose eigenvalues are $1$ or $-1$ only.
As seen in these equations, the space of the image of $U$ has a direct product structure consisting of a single arbitrary dimensional space, $\lfloor\frac{\chi-1}{2}\rfloor$ of $2$-dimensional spaces, and a single $\gamma$-dimensional space. The basis is now denoted as $\{\ket{a}\otimes\ket{b_{\lfloor\frac{\chi-1}{2}\rfloor-1}}\otimes\cdots\otimes\ket{b_{0}}\otimes\ket{\rho}\}$, where the ranges of indices are $a \in\{0,1,\cdots a_0-1\},b_{m}\in\{0,1\},\rho\in\{0,1,\cdots,\gamma-1\}$ with a certain integer $a_0$.

We shall give a proof of the existence of $U$ by induction in terms of $\chi$. When $\chi=1$, Eqs. (\ref{eq:str_jor_m_03}) and (\ref{eq:str_jor_m_07}) are simply a paraphrase of Eqs. (\ref{eq:str_jor_m_01}) and (\ref{eq:str_jor_m_02}).

Assume that the proposition holds when $\chi$ is an odd number $2n-1$, and consider the case of $\chi=2n$.
By this assumption, even if $\chi=2n$, there exists an isometry $U$ for the first $2n-1$ $s_\mu^{(0,1)}$'s, i.e., for $\mu\in\{0,1,...,2n-2\}$,
such that Eqs. (\ref{eq:str_jor_m_03})-(\ref{eq:str_jor_m_05}) and (\ref{eq:str_jor_m_07}) hold. 
Then, we attempt to show that Eq. (\ref{eq:str_jor_m_06}) also holds for the remaining basis, $s^{(0,1)}_{2n-1}$. 

Because $s_{\chi-1}^{(0,1)}$ has nonzero entries for the $(0,1)$-th and the $(1,0)$-th off-diagonal blocks due to Eq. (\ref{eq:str_jor_2_2+}), $Us_{\chi-1}^{(0,1)}U^{\dagger}$ should have $X_{0,1}$ and/or $Y_{0,1}$ components for the rightmost space spanned by $\ket{\rho}$. Thus, it can be written as a linear combination of terms, each of which has the form
\begin{equation}
V\otimes W_{n-1}\otimes\cdots\otimes W_{1}\otimes Y_{0,1}
\makebox{ or }
V\otimes W_{n-1}\otimes\cdots\otimes W_{1}\otimes X_{0,1},
\label{eq:tmp_2_1}
\end{equation}
where  $W_m\in\{Id,X,Y,Z\}$, and $V$ is an arbitrary Hermitian operator.
Equation (\ref{eq:str_jor_2_2}) for $\rho=0$, $\sigma=1$, $\mu=0$ and $\nu=\chi-1$, i.e., $\{s_0^{(0,1)},s_{\chi-1}^{(0,1)}\}=0$, guarantees that the second type in Eq. (\ref{eq:tmp_2_1}) must be $0$.
Similarly, $\{s_{\mu_0}^{(0,1)},s_{\chi-1}^{(0,1)}\}=0$ for $0<\mu_0<\chi-1$ implies that the terms of the form
\begin{eqnarray}
&&V\otimes W_{n-1}\otimes\cdots \otimes W_{\frac {\mu_0+3}2}\otimes Z\otimes \overbrace{Y\otimes\cdots\otimes Y}^{\frac{\mu_0-1}{2}}\otimes Y_{0,1}
\makebox{ or}\\
&&V\otimes W_{n-1}\otimes\cdots \otimes W_{\frac {\mu_0}2+2}\otimes X\otimes\overbrace{Y\otimes\cdots\otimes Y}^{\frac{\mu_0}{2}-1}\otimes Y_{0,1}
\end{eqnarray}
have no contributions to $Us_{\chi-1}^{(0,1)}U^\dagger$ when $\mu_0$ is odd or even, respectively. 
Therefore, $Us_{\chi-1}^{(0,1)} U^\dagger$ must have the form
\begin{eqnarray}
&&V\otimes\overbrace{Y\otimes\cdots\otimes Y}^{n-1}\otimes Y_{0,1}.
\label{eq:tmp_2_2}
\end{eqnarray}
Further, another relation $\left(s_{\chi-1}^{(0,1)}\right)^2=e_0+e_1$, which is also obtained from Eq. (\ref{eq:str_jor_2_2}), requires that the square of Eq. (\ref{eq:tmp_2_2}) be equal to $U(e_0+e_1)U^\dagger$, which means $V^{2}=Id$, thus $V$ can be taken to be $Z^{*}$.
Therefore, Eq. (\ref{eq:str_jor_m_06}) holds for $\chi=2n$. 

Let us now prove the remaining step for induction. Assume that the proposition holds when $\chi$ is an even number $2n$, and show that it also does when $\chi=2n+1$.
Let us rewrite Eqs. (\ref{eq:str_jor_m_03})-(\ref{eq:str_jor_m_07}) for clarity for the case $\chi=2n$. Since $\lfloor \frac{\chi-1}{2}\rfloor=n-1$, the assumption is that an isometry $U$ exists, such that the following hold for the subset of $s_\mu$'s and $e_\rho$'s,
\begin{eqnarray}
\label{eq:str_indctn01}
Us_{0}^{(\rho,\rho+1)}U^{\dagger}   &=& Id\otimes\overbrace{Id\otimes\cdots\otimes Id}^{n-1}\otimes X_{\rho,\rho+1}, 
\\
Us_{2n'-1}^{(0,1)}U^{\dagger}    &=& Id\otimes\overbrace{Id\otimes\cdots\otimes Id}^{n-n'-1}\otimes Z\otimes\overbrace{Y\otimes\cdots\otimes Y}^{n'-1}\otimes Y_{0,1}
\;\;\; {\rm for}\ n-1\geq n'\geq1,
\label{eq:str_indctn02}
\\
Us_{2n'}^{(0,1)}U^{\dagger}  &=& Id\otimes\overbrace{Id\otimes\cdots\otimes Id}^{n-n'-1}\otimes X\otimes\overbrace{Y\otimes\cdots\otimes Y}^{n'-1}\otimes Y_{0,1}
\;\;\; {\rm for}\ n-1\geq n'\geq1,
\label{eq:str_indctn03}
\\
\label{eq:str_indctn04}
Us_{2n-1}^{(0,1)}U^{\dagger}   &=& Z^{*}\otimes\overbrace{Y\otimes\cdots\otimes Y}^{n-1}\otimes Y_{0,1},
\\
\label{eq:str_indctn05}
Ue_{\rho}U^{\dagger}    &=& Id\otimes\overbrace{Id\otimes\cdots\otimes Id}^{n-1}\otimes\ket{\rho}\bra{\rho}.
\end{eqnarray}
Then we will show the existence of an isometry $U^\prime$ that transforms $s_{2n}^{(0,1)}$ to Eq. (\ref{eq:str_jor_m_05}) with $n^\prime=n$, i.e., $Id\otimes X\otimes Y\otimes \cdots \otimes Y\otimes Y_{0,1}$, while other relations (\ref{eq:str_jor_m_03})-(\ref{eq:str_jor_m_07}) for $n^\prime<n$ are also satisfied with $U^\prime$ instead of $U$. 

Similarly as above, due to Eq. (\ref{eq:str_jor_2_2+}), $Us_{\chi-1}^{(0,1)}U^{\dagger}$ can be written as a linear combination of terms in Eq. (\ref{eq:tmp_2_1}).
Also, it follows from the $\chi-2$ relations, Eq. (\ref{eq:str_jor_2_2}) with $\rho=0$, $\sigma=1$, $0\leq\mu<\chi-2$ and $\nu=\chi-1$,
that $Us_{\chi-1}^{(0,1)}U^{\dagger}$ should have the form of (\ref{eq:tmp_2_2}).
The relation (\ref{eq:str_jor_2_2}) with $\rho=0$, $\sigma=1$ and $\mu+1=\nu=\chi-1$ leads to $Z^{*}V+VZ^{*}=0$. This means that $V$ can be written as $A+A^{\dagger}$, where the kernel of $A$ is not smaller than the eigenspace of $Z^*$ corresponding to the eigenvalue $-1$, while the range of $A$ is not larger than the same eigenspace.
Equation (\ref{eq:str_jor_2_2}) for $\rho=0$, $\sigma=1$ and $\mu=\nu=\chi-1$, i.e., $ \left(s_{\chi-1}^{(0,1)}\right)^2=e_{0}+e_{1}$, indicates $V^2=Id$, thus $A^{\dagger}A=\frac{1}{2}(Id+Z^{*})$ and $AA^{\dagger}=\frac{1}{2}(Id-Z^{*})$. Therefore, the dimensions of the eigenspaces of $Z^*$ for eigenvalues $\pm 1$ are equal, and the space on which $Z^*$ acts has an even dimension $k_0$.

Now we consider an isometry $U_{Z^*}$ from the $k_0$-dimensional space to a product of two spaces, which are $k_{0}/2$- and $2$-dimensional spaces. It transforms an eigenvector of $Z^*$, corresponding to the eigenvalue $b=\pm 1$, to the $k_0/2$-dimensional subspace with the eigenvalue being encoded in the second (2-dim) subspace as $ \ket{(1-b)/2}$. 
With $U_{Z^*}$, we can consider a unitary operator $U^\sharp:=(Id\otimes X)U_{Z^*}A+(Id\otimes \ket1\bra1)U_{Z^*}$, so that a new isometry $U^\prime=(U^\sharp \otimes\overbrace{Id\otimes\cdots\otimes Id}^{n}) U$ transforms the space spanned by $\{\ket{k,\rho}\}_{k,\rho}$ to the one spanned by $\{\ket{k}\otimes\ket{b_{n}}\otimes\cdots\otimes\ket{b_{0}}\otimes\ket{\rho}\}_{k\in\{0,1,\cdots\frac {k_0}2-1\},b_{m}\in\{0,1\},\rho\in\{0,1,\cdots,\gamma-1\}}$.
We can directly check that Eqs. (\ref{eq:str_jor_m_03})-(\ref{eq:str_jor_m_05}) and (\ref{eq:str_jor_m_07}) hold after replacing $U$ with $U^\prime$, noting the effect of $U^\sharp$, e.g., $U^\sharp U^{\sharp\dagger}=Id\otimes Id$ and $U^\sharp V U^{\sharp\dagger}=U^\sharp(A+A^\dagger)U^{\sharp\dagger}=I\otimes X$, where the second space on the right is two-dimensional.
Now that the induction is complete, an isometry $U$ exists such that Eqs. (\ref{eq:str_jor_m_03})-(\ref{eq:str_jor_m_07}) as well as Eqs. (\ref{eq:str_jor_m_01}) and (\ref{eq:str_jor_m_02}) for any positive integer $\chi$.

Equations (\ref{eq:str_jor_m_03})-(\ref{eq:str_jor_m_07}) can be generalized to arbitrary combinations of $\rho$ and $\sigma$, leading to the justification of  Eqs. (\ref{eq:JordanAlg_2})-(\ref{eq:JordanAlg_6}). Let us see how this can be done.

If $\gamma=2$, the algebra will look like either $\mathfrak{S}_{2n^\prime-1}$ in Eq. (\ref{eq:JordanAlg_5}) or $\mathfrak{S}_{2n^\prime}$ in (\ref{eq:JordanAlg_6}), depending on whether $\chi$ is odd or even, respectively. That is, the corresponding simple Jordan algebra has the structure of $\mathfrak S_{\chi+2}$.

When $\gamma\geq3$, recall that the linear space  $ J$ spanned by
$\{Ue_{\rho}U^\dagger\}_{\rho\in \Gamma}\cup\{Us_{\mu }^{(\rho,\sigma)}U^\dagger\}_{\rho<\sigma\in \Gamma,\mu\in\{0,1,\cdots ,\chi-1\}}$ is closed under anti-commutation, and all the generators are linearly independent, for $\{e_{\rho}\}_{\rho\in \Gamma}$ and $\{s_{\mu }^{(\rho,\sigma)}\}_{\rho<\sigma\in \Gamma,\mu\in\{0,1,\cdots ,\chi-1\}}$ are a basis of the space.
Since any other $s_\mu^{(\rho,\sigma)}$-type basis generators can be obtained by taking anti-commutators as 
\begin{eqnarray}
s_{\mu}^{\prime(0,1)}&:=& s_{\mu}^{(0,1)},
\\
s_{0}^{\prime(\rho,\sigma)}&:=& \{\cdots\{s_{0}^{(\rho,\rho+1)},s_{0}^{(\rho+1,\rho+2)}\},\cdots ,s_{0}^{(\sigma-1,\sigma)}\},\\
s_{\mu}^{\prime(1-b,\sigma)}&:=& \{s_{\mu}^{(0,1)},s_{0}^{\prime(b,\sigma)}\},
\\
s_{\mu}^{\prime(\sigma,\tau)}&:=& \{\{s_{\mu}^{(0,1)},s_{0}^{\prime(0,\sigma)}\},s_{0}^{\prime (1,\tau)}\},
\end{eqnarray}
with $b\in\{0,1\}$, $\rho,\sigma,\tau\in\Gamma$, $\mu \in\{0,\cdots,\chi-1\}$ and $\rho+1<\sigma<\tau$, we can see their structures, following Eqs. (\ref{eq:str_jor_m_03}) - (\ref{eq:str_jor_m_06}),
\begin{eqnarray}
U s_{0}^{\prime(\rho,\sigma)}U^{\dagger}    &=&  Id\otimes\overbrace{Id\otimes\cdots\otimes Id}^{\lfloor\frac{\chi-1}{2}\rfloor}\otimes X_{\rho,\sigma}, \label{sprime_str_0}
\\
Us_{2n'-1}^{\prime(\rho,\sigma)} U^{\dagger}&=&  Id\otimes\overbrace{Id\otimes\cdots\otimes Id}^{\lfloor\frac{\chi-1}{2}\rfloor-n'}\otimes Z\otimes\overbrace{Y\otimes\cdots\otimes Y}^{n'-1}\otimes Y_{\rho,\sigma} \;\;\;{\rm for}\ \lfloor\frac{\chi-1}{2}\rfloor\geq n'\geq1, 
\label{sprime_str_odd}
\\
Us_{2n'}^{\prime(\rho,\sigma)}  U^{\dagger} &=& Id\otimes\overbrace{Id\otimes\cdots\otimes Id}^{\lfloor\frac{\chi-1}{2}\rfloor-n'}\otimes X\otimes\overbrace{Y\otimes\cdots\otimes Y}^{n'-1}\otimes Y_{\rho,\sigma} \;\;\;
{\rm for}\ \lfloor\frac{\chi-1}{2}\rfloor\geq n'\geq1, 
\label{sprime_str_even}
\\
Us_{\chi-1}^{\prime(\rho,\sigma)} U^{\dagger}   &=&  Z^{*}\otimes\overbrace{Y\otimes\cdots\otimes Y}^{\lfloor\frac{\chi-1}{2}\rfloor}\otimes Y_{\rho,\sigma} \;\;\mbox{when $\chi$ is even.} \label{sprime_str_chim1}
\end{eqnarray}

By construction, $\{e_{\rho}\}_{\rho\in \Gamma}\cup\{s_{\mu }^{\prime(\rho,\sigma)}\}_{\rho<\sigma\in \Gamma,\mu\in\{0,1,\cdots ,\chi-1\}}$ is the set of linearly independent operators, the number of which is equal to the dimension of $J$, thus this set is a basis of $J$.

When $\chi=1$ or $2$, it is straightforward to see that $J$ has a structure of $\mathfrak M_{\gamma}^{(\chi)}$ in Eqs. (\ref{eq:JordanAlg_2}) or (\ref{eq:JordanAlg_3}), respectively, due to Eqs. (\ref{eq:str_jor_m_03}), (\ref{eq:str_jor_m_06}), and (\ref{eq:str_jor_m_07}).

Let us consider the remaining cases of $\gamma \geq3$ and $\chi\geq3$.
From Eqs. (\ref{sprime_str_odd}) and (\ref{sprime_str_even}), we see that
\begin{eqnarray}
\{s_{1}^{\prime(0,1)},s_{2}^{\prime(0,2)}\}
&=&
U^\dagger(    Id\otimes\overbrace{Id\otimes\cdots\otimes Id}^{\lfloor\frac{\chi-1}{2}\rfloor-1}\otimes\{Z \otimes Y_{0,1},X\otimes  Y_{0,2}\})U
\nonumber \\
&=&
U^\dagger  (  Id\otimes\overbrace{Id\otimes\cdots\otimes Id}^{\lfloor\frac{\chi-1}{2}\rfloor-1}\otimes Y \otimes Y_{1,2})U
\label{sprime_str_12}
\end{eqnarray}
should be in $ J$, thus this must be written as a linear combination of  $\{e_{\rho}\}_{\rho\in \Gamma}\cup\{s_{\mu }^{\prime(\rho,\sigma)}\}_{\rho<\sigma\in \Gamma,\mu\in\{0,1,\cdots ,\chi-1\}}$. Nevertheless, since Eq. (\ref{sprime_str_12}) can be obtained from $\{s_\mu^{\prime (\rho,\sigma)}\}$ in Eqs. (\ref{sprime_str_0})-(\ref{sprime_str_chim1}) only by setting $\chi=4$ and $Z^*=Id$ in Eq. (\ref{sprime_str_chim1}), the explicit forms of the algebra $J$ indicate that such a requirement holds only when 
$\chi=4$ and $Z^*\propto Id$. Therefore, the corresponding simple Jordan algebra has the structure $\mathfrak M_{\gamma}^{(4)}$ in Eq. (\ref{eq:JordanAlg_4}).

In conclusion, when it is spanned by the basis $\{e_{\rho}\}_{\rho\in \Gamma}$ and $\{s_{\mu }^{(\rho,\sigma)}\}_{\rho<\sigma\in \Gamma,\mu\in\{0,1,\cdots ,\chi-1\}}$ that satisfy (\ref{eq:str_jor_2_1})$\sim$(\ref{eq:str_jor_2_3}) and (\ref{eq:str_jor_2_4}), the space must have one of the structures in Eq. (\ref{eq:JordanAlg_1})$\sim$(\ref{eq:JordanAlg_6}).
\end{proof}

Other linear spaces we have seen in Sec. \ref{sec:main}, namely $\bar{J}$ and $\hat{J}$, satisfy simple algebraic relations as follows.

\begin{lemma}
\label{lemm:relation}
Let a triple $(J,\bar{J}, \hat {J})$ be equal to either one of the three combinations; 
$({\mathfrak R}, \bar{\mathfrak R}, \hat {\mathfrak R})$,
$({\mathfrak M}_\gamma^{(k)}, \bar{\mathfrak M}_\gamma^{(k)}, \hat {\mathfrak M}_\gamma^{(k)})$, and
$({\mathfrak S}_n, \bar{\mathfrak S}_n, \hat{\mathfrak S}_n)$, where $\gamma\geq3$, $k\in\{1,2,4\}$ and $n\geq3$.
Then the following relations
\begin{eqnarray}
i\mathcal L([ \bar J,\bar J]) &\subseteq& \bar J=i\mathcal L([ J, J]),
\label{eq:relation_1}
\\
i\mathcal L([ \bar J, J]) &\subseteq& J= \{J,J\},
\label{eq:relation_2}
\\
i[\hat  J, J] &=& i[ \hat  J, \bar J]=\{0\}
\label{eq:relation_3}
\end{eqnarray}
hold.
\end{lemma}
\begin{proof}
Equation (\ref{eq:relation_3}) can be verified straightforwardly by using the definition of $\bar{J}$ and $\hat{J}$.
Lemma \ref{sec:rep_Jor} leads to $\{J,J\}\subseteq J$, and $\{J,J\}\supseteq J$ also holds because $\frac12Id\in J$.
Thus, $J=\{J,J\}$ as in Eq. (\ref{eq:relation_2}).

The equality of $i\mathcal L([ J, J])=\bar J$, which is in Eq. (\ref{eq:relation_1}), turns out to be a sufficient condition for the remaining two inclusion relations in Eqs. (\ref{eq:relation_1}) and (\ref{eq:relation_2}). An inclusion
\begin{equation}\label{eq:inclusion_JJJ}
\mathcal L(i[i[J,J],J]) \subseteq  
\mathcal L(\{J,\{J,J\}\}),
\end{equation}
can be obtained because of the identity
\begin{equation}
i[i[j_1,j_2],j_3] = 
-\{j_1,\{j_2,j_3\}\}
+\{j_2,\{j_1,j_3\}\}.
\end{equation}
The LHS of Eq. (\ref{eq:inclusion_JJJ}) is equal to $i\mathcal L([\bar J,J])$ if $i\mathcal L([ J, J])=\bar J$.
Also its RHS must be a subset of $J$ since $\{J,J\}=J$, hence, $i\mathcal L([ \bar J, J])\subseteq J$.
Similarly, we have 
\begin{equation}\label{eq:inclusion_JJbarJ}
\mathcal L(i[i[J,J],\bar J]) \subseteq  
\mathcal L(i[i[J,\bar J],J]),
\end{equation}
due to
\begin{equation}
i[i[j_1,j_2],j_3] = 
i[i[j_1,j_3],j_2]
-i[i[j_2,j_3],j_1].
\end{equation}
Thanks to the condition $i\mathcal L([ J, J])=\bar J$ and its consequence $i\mathcal L([ \bar J, J])\subseteq J$, Eq. (\ref{eq:inclusion_JJbarJ}) implies $i\mathcal L([ \bar J, \bar J])\subseteq \bar J$, which is the first inclusion relation in Eq. (\ref{eq:relation_1}).

The proof of $i\mathcal L([ J, J])=\bar J$ is straightforward from the explicit forms of $J$ and $\bar J$, albeit rather tedious.
In the following, we use trivial symmetries $X_{k,q}=X_{q,k}$ and $Y_{k,q}=-Y_{q,k}$ without mentioning.

Let us consider the six cases of $J$ being $\mathfrak R, \mathfrak M_\gamma ^{(1)}, \mathfrak M_\gamma ^{(2)}, \mathfrak M_\gamma ^{(4)}, \mathfrak S_{2n'-1}$, and $\mathfrak S_{2n'}$.

\begin{enumerate}
\item[(i)] $J=\mathfrak R$.  Trivially, $[J,J]=\{0\}=\bar {\mathfrak R}$.

\item[(ii)] $J=\mathfrak M_\gamma ^{(1)}$.
For $k\neq q$ and $k'\neq q'$, the commutators
\begin{eqnarray}
i[Id\otimes X_{k,q},Id\otimes X_{k',q'}]
&=&\;
\left\{
\begin{array}{ll}
-Id\otimes Y_{q,q'} & 
\makebox[.7cm][l]{}
\makebox {if $k=k'$ and $q\neq q'$,}\\
\makebox[.7cm][l]{$0$} &
\makebox {if $k=k'$ and $q= q'$ or both $k$ and $q$ are neither $k'$ nor $q'$,}
\end{array}
\right.
\nonumber \\
i[Id\otimes X_{k,q},Id\otimes \ket{k'}\bra{k'}]
&=& 
(\delta(k,k')-\delta (q,k'))Id\otimes Y_{k,q},
\nonumber\\
i[Id\otimes \ket{k}\bra{k},Id\otimes \ket{k'}\bra{k'}]
&=& 0,
\label{eq:relation_2_1}
\end{eqnarray}
imply $i\mathcal L([J,J])\subseteq \mathcal L( \bar J)$. The inclusion in the opposite direction is guaranteed by 
\begin{equation}
Id\otimes Y_{k,q} = \frac 12 i[Id\otimes X_{k,q}, Id\otimes \ket{k}\bra{k}-Id\otimes \ket{q}\bra{q}].
\label{eq:relation_2_2}
\end{equation}

\item[(iii)] $J=\mathfrak M_\gamma ^{(2)}$.
For $k\neq q$ and $k'\neq q'$,
\begin{eqnarray}
i[Id\otimes X_{k,q},Z^*\otimes Y_{k',q'}]
&=&\;
\left\{
\begin{array}{ll}
Z^*\otimes X_{q,q'}&\makebox {if $k=k'$ and $q\neq q'$,}\\
-2\;Z^*\otimes Z_{k,q}&\makebox {if $k=k'$ and $q= q'$,}\\
0&\makebox {if both $k$ and $q$ are neither $k'$ nor $q'$,}
\end{array}
\right.
\nonumber \\
i[Z^*\otimes Y_{k,q},Id\otimes \ket{k'}\bra{k'}]
&=&
(-\delta(k,k')+\delta (q,k'))Z^*\otimes X_{k,q},
\nonumber \\
i[Z^*\otimes Y_{k,q},Z^*\otimes Y_{k',q'}]
&=&\;
\left\{
\begin{array}{ll}
-Id\otimes Y_{q,q'} &
\makebox[.5cm]{}
\makebox {if $k=k'$ and $q\neq q'$,}\\
\makebox[.5cm]{$0$} &
\makebox {if $k=k'$ and $q= q'$ or 
both $k$ and $q$ are neither $k'$ nor $q'$,}
\end{array}
\right.
\end{eqnarray}
and Eqs. (\ref{eq:relation_2_1}) imply $i\mathcal L([J,J])\subseteq \mathcal L( \bar J)$. The inclusion in the opposite direction can be shown by 
\begin{eqnarray}
Z^*\otimes X_{k,q}
&=&
-\frac12 i[Z^*\otimes Y_{k,q}
,Id\otimes \ket{k}\bra{k}-Id\otimes \ket{q}\bra{q}],
\\
Z^*\otimes Z_{k,q}
&=&
-\frac 12 i[Id\otimes X_{k,q},Z^*\otimes Y_{k,q}],
\end{eqnarray}
and Eq. (\ref{eq:relation_2_2}).

\item[(iv)] $J=\mathfrak M_\gamma ^{(4)}$, where $\gamma \geq 3$.
Let $k\neq q$, $k'\neq q'$ and $(W,W')$ be equal to either of $(X,Y)$, $(Y,Z)$ or $(Z,X)$, then
\begin{eqnarray}
i[ Id\otimes X_{k,q}, W\otimes Y_{k',q'}]
&=&\;
\left\{
\begin{array}{ll}
\makebox[4cm][l]{$
W\otimes X_{q,q'}
$} &
\makebox {if $k=k'$ and $q\neq q'$,}\\
-2( W\otimes\ket{k}\bra{k}-W\otimes \ket{q}\bra{q}) &
\makebox[5mm]{}
\makebox {if $k=k'$ and $q= q'$,}\\
\makebox[4cm][l]{$0$} &
\makebox {if both $k$ and $q$ are neither $k'$ nor $q'$,}
\end{array}
\right.
\nonumber \\
i[W\otimes Y_{k,q},Id\otimes \ket{k'}\bra{k'}]
&=&
(-\delta(k,k')+\delta (q,k'))W\otimes X_{k,q},
\nonumber \\
i[W\otimes Y_{k,q},W\otimes Y_{k',q'}]
&=&
\left\{
\begin{array}{ll}
-Id\otimes Y_{q,q'} &
\makebox[.5cm][l]{}
\makebox {if $k=k'$, $q\neq q'$,}
\\
\makebox[.5cm][l]{$0$} &
\makebox {if $k=k'$ and $q= q'$ or both $k$ and $q$ are neither $k'$ nor $q'$,}\\
\end{array}\right.
\nonumber \\
i[W\otimes Y_{k,q},W'\otimes Y_{k',q'}] &=&
\left\{
\begin{array}{ll}
\makebox[4cm][l]{$
W''\otimes X_{q,q'}
$} &
\makebox {if $k=k'$, $q\neq q'$ and $W=W'$,}\\
-( W''\otimes\ket{k}\bra{k}+W''\otimes \ket{q}\bra{q}) &
\makebox[5mm][l]{}
\makebox {if $k=k'$ and $q= q'$,}
\\
\makebox[4cm][l]{$0$} &
\makebox {if both $k$ and $q$ are neither $k'$ nor $q'$,}
\end{array}\right.
\end{eqnarray}
can be shown for $W''$ being equal to either of the Pauli operators, $X$, $Y$ and $Z$, when the pair $(W,W')$ is equal to either of $(Y,Z)$, $(Z,X)$ or $(X,Y)$, respectively. 
These relations and Eqs. (\ref{eq:relation_2_1}) lead to $i\mathcal L([J,J])\subseteq \mathcal L( \bar J)$. 
The one in the opposite direction can be verified by 
\begin{eqnarray}
W\otimes X_{k,q}
&=&
-\frac12 i[W\otimes Y_{k,q}
,Id\otimes \ket{k}\bra{k}-Id\otimes \ket{q}\bra{q}],
\\
W\otimes \ket k\bra k
&=&
\frac 12 (i[W'\otimes Y_{q,r},W''\otimes Y_{q,r}]
- i[W'\otimes Y_{k,q},W''\otimes Y_{k,q}]
\nonumber\\&&{}
- i[W'\otimes Y_{k,r},W''\otimes Y_{k,r}]),
\end{eqnarray}
where $r\in \{0,1\cdots ,\gamma -1\}$ is a number different from $k$ and $q$, as well as Eq. (\ref{eq:relation_2_2}).

\item[(v)] $J=\mathfrak S_{2n'-1}$.
For $m\geq m'\in\{1,2,\cdots n'-1\}$ and $W,W'\in \{X,Z\}$,
\begin{eqnarray}
&& i[
	\overbrace{Id\otimes\cdots\otimes Id}^{n'-m}\otimes W\otimes\overbrace{Y\otimes\cdots\otimes Y}^{m-1},
	Id\otimes\cdots\otimes Id
]=  0,
\nonumber \\
&& i[
	\overbrace{Id\otimes\cdots\otimes Id}^{n'-m}\otimes W\otimes\overbrace{Y\otimes\cdots\otimes Y}^{m-1},
	\overbrace{Id\otimes\cdots\otimes Id}^{n'-m'}\otimes W'\otimes\overbrace{Y\otimes\cdots\otimes Y}^{m'-1}
] \nonumber\\
&& =
\left\{
\begin{array}{cc}
2s \overbrace{Id\otimes\cdots\otimes Id}^{n'-m}\otimes W\otimes\overbrace{Y\otimes\cdots\otimes
Y}^{m-m'-1}\otimes \overline W^{\prime}\otimes\overbrace{Id\otimes\cdots\otimes
Id}^{m'-1} & 
\mbox{when $m>m'$,}
\\
-2s\overbrace{Id\otimes\cdots\otimes Id}^{n'-m}\otimes Y\otimes\overbrace{Id\otimes\cdots\otimes Id}^{m-1}
& 
\mbox{when $m=m'$, $W\neq W'$,}
\\
0
& 
\mbox{when $m=m'$, $W= W'$,}
\end{array}
\right.
\label{eq:relation_5_1}
\end{eqnarray}
can be shown for $(\overline W', W', s)$ being equal to either $(X, Z, 1)$ or $(Z, X, -1)$.
These relations imply the inclusion $i\mathcal L([J,J])\subseteq \mathcal L( \bar J)$. 
That in the opposite direction can be derived from 
\begin{eqnarray}
&&\overbrace{Id\otimes\cdots\otimes Id}^{n'-m}\otimes W\otimes\overbrace{Y\otimes\cdots\otimes
Y}^{m-m'-1}\otimes W^{\prime}\otimes\overbrace{Id\otimes\cdots\otimes
Id}^{m'-1}
\nonumber\\
&&=\; -\frac s2 i[
	\overbrace{Id\otimes\cdots\otimes Id}^{n'-m}\otimes W\otimes\overbrace{Y\otimes\cdots\otimes Y}^{m-1},
	\overbrace{Id\otimes\cdots\otimes Id}^{n'-m'}\otimes \overline W'\otimes\overbrace{Y\otimes\cdots\otimes Y}^{m'-1}
],
\nonumber\\
&&\overbrace{Id\otimes\cdots\otimes Id}^{n'-m}\otimes
Y\otimes\overbrace{Id\otimes\cdots\otimes Id}^{m-1}
\nonumber\\
&&=\; \frac12 i[
	\overbrace{Id\otimes\cdots\otimes Id}^{n'-m}\otimes X\otimes\overbrace{Y\otimes\cdots\otimes Y}^{m-1},
	\overbrace{Id\otimes\cdots\otimes Id}^{n-m}\otimes Z\otimes\overbrace{Y\otimes\cdots\otimes Y}^{m-1}
],
\nonumber\\
\label{eq:relation_5_2}
\end{eqnarray}
where $m\neq m'$.

\item[(vi)]  $J=\mathfrak S_{2n'}$.
For $m\in\{1,2,\cdots n'-1\}$ and $W\in \{X,Z\}$,
\begin{eqnarray}
i[ Z^{*}\otimes\overbrace{Y\otimes\cdots\otimes Y}^{n'-1}, Id\otimes\cdots\otimes Id]
&=&
i[
Z^{*}\otimes\overbrace{Y\otimes\cdots\otimes Y}^{n'-1},
Z^{*}\otimes\overbrace{Y\otimes\cdots\otimes Y}^{n'-1}
] = 0,
\nonumber\\
i[
\overbrace{Id\otimes\cdots\otimes Id}^{n'-m}\otimes W\otimes\overbrace{Y\otimes\cdots\otimes Y}^{m-1},
Z^{*}\otimes\overbrace{Y\otimes\cdots\otimes Y}^{n'-1}
]
&=&
2s \overbrace{Id\otimes\cdots\otimes Id}^{n'-m}\otimes\overline W\otimes\overbrace{Id\otimes\cdots\otimes
Id}^{m-1}
\end{eqnarray}
can be shown for $(\overline W, W, s)$ being equal to either $(X, Z, 1)$ or $(Z, X, -1)$.  
These relations and Eqs. (\ref{eq:relation_5_1}) 
indicate $i\mathcal L([J,J])\subseteq \mathcal L( \bar J)$. The inclusion in the opposite direction is verified by using the following commutator
\begin{equation*}
Z^*\otimes\overbrace{Y\otimes\cdots\otimes Y}^{n'-m-1}\otimes W\otimes\overbrace{Id\otimes\cdots\otimes Id}^{m-1}
=
-\frac s2 i[
	\overbrace{Id\otimes\cdots\otimes Id}^{n'-m}\otimes \overline W\otimes\overbrace{Y\otimes\cdots\otimes Y}^{m-1},
Z^{*}\otimes\overbrace{Y\otimes\cdots\otimes Y}^{n'-1}
]
\end{equation*}
as well as Eq. (\ref{eq:relation_5_2}).
\end{enumerate}

\end{proof}

We now give two algebraic relations below, whose proofs are simple, but rather lengthy.
The first one states that, for a simple Jordan algebra $J$, the maximum set $G$ of operators that satisfies $[G,J]\subseteq J$ can be written in a compact form. In the following Lemma \ref{lemm:pseudo_Ideal}, the triple $(J, \bar{J}, \hat {J})$ is assumed to be equal to either one of the three,
$({\mathfrak R},\bar{\mathfrak R}, \hat{\mathfrak R})$,
$({\mathfrak M}_\gamma^{(k)}, \bar{\mathfrak M}_\gamma^{(k)}, \hat{\mathfrak M}_\gamma^{(k)})$, and
$({\mathfrak S}_n, \bar{\mathfrak S}_n, \hat{\mathfrak S}_n)$, where $\gamma\geq3$, $k\in\{1,2,4\}$ and $n\geq3$, as in Lemma \ref{lemm:relation}. Also, $\mathcal{H}$ denotes the range of the largest projection operator in $J$.

\begin{lemma}
\label{lemm:pseudo_Ideal}
Let $i\tilde{G}^{(0)}$ be the maximal set of Hermitian operators that satisfies $[i\tilde{G}^{(0)}, J]\subseteq J$. Then, it is equal to $i{\mathcal L}(\hat J \cup \bar J)$, namely,
\begin{equation}
i\tilde{G}^{(0)}:= \{\tilde{h}|\tilde{h}\in i \cdot {\rm u}({\rm dim}\mathcal H)\land \forall h\in J,\;i[\tilde{h},h]\in J\}={\mathcal L}(\hat J\cup \bar J).
\end{equation}
\end{lemma}

\begin{proof}
From Lemma \ref{lemm:relation},  $i[\hat{J},J]=\{0\}\subseteq  J$ 
and $i[\bar{J},J]\subseteq J$ hold, thus
${\mathcal L}(\hat J \cup \bar J)\subseteq i\tilde{G}^{(0)}$.

So, let us focus on the proof of the opposite inclusion, $  {\mathcal L}(\hat J \cup \bar J)\supseteq i\tilde{G}^{(0)}$. 
We will prove this relation for each form of $J$, one by one.

\begin{enumerate}
\item[(i)]  $J=\mathfrak{R}$.
We can easily identify $i\tilde{G}^{(0)}$ to be $ i \cdot{\rm u}({\rm dim}\mathcal H)$, because it is equal to $\hat{\mathfrak{R}}$,
thus ${\mathcal L}(\hat J \cup \bar J)\supseteq i\tilde{G}^{(0)}$ holds trivially.

\item[(ii)] $J=\mathfrak{M}_{\gamma}^{(1)}$. Any element in the set $i\tilde{G}^{(0)}$ should be expanded with respect to the basis on the second space as
\begin{equation}
\tilde{h}  :=    \sum_{k=0}^{\gamma-1}\left(\tilde{h}_{\ket{k}\bra{k}}\otimes\ket{k}\bra{k}
+\sum_{q=k+1}^{\gamma-1}
\left(\tilde{h}_{X_{k,q}}\otimes X_{k,q}+\tilde{h}_{Y_{k,q}}\otimes Y_{k,q}\right)\right).
\label{eq:str_c-dyn2_m_2}
\end{equation}
From the requirement $i[\tilde{h},Id\otimes\ket{k}\bra{k}]\in\mathfrak{M}_{\gamma}^{(1)}$ for $k\in\{0,1,\cdots,\gamma-1\}$, we have $\tilde{h}_{X_{k,q}}=0$ and $\tilde{h}_{Y_{k,q}}\propto Id_A$.
Using these and $i[\tilde{h},Id\otimes X_{k,k+1}]\in\mathfrak{M}_{\gamma}^{(1)}$ for $k\in\{0,1,\cdots,\gamma-2\}$, we see $\tilde{h}_{\ket{k}\bra{k}}=\tilde{h}_{\ket{k+1}\bra{k+1}}$.
Therefore, $\tilde{h}$ must have the form
\begin{equation}
\tilde{h}_{\ket{0}\bra{0}}\otimes Id_Q
+\sum_{k=0}^{\gamma-2}\sum_{q=k+1}^{\gamma-1}y_{k,q}Id_A\otimes Y_{k,q}
\end{equation}
with an appropriate $y_{k,q}\in\mathbb R$. Since  $\tilde{h}_{\ket{0}\bra{0}}\otimes Id_Q\in \hat {\mathfrak{M}}_{\gamma}^{(1)}$  and
$Id_A\otimes Y_{k,q}\in \bar {\mathfrak{M}}_{\gamma}^{(1)}$, $\tilde{h}\in  {\mathcal L}(\hat J\cup \bar J)$ is guaranteed.

\item[(iii)] $J=\mathfrak{M}_{\gamma}^{(2)}$. Again, any element in the set $i\tilde{G}^{(0)}$ can be written as Eq. (\ref{eq:str_c-dyn2_m_2}).
Because $i[\tilde{h},Id\otimes\ket{k}\bra{k}]\in\mathfrak{M}_{\gamma}^{(2)}$ for $k\in\{0,1,\cdots,\gamma-1\}$, which is the condition for $i\tilde{G}^{(0)}$, 
$\tilde{h}_{X_{k,q}}\propto Z^{*}$ and $\tilde{h}_{Y_{k,q}}\propto Id$ should hold.
From these and $i[\tilde{h},Id\otimes X_{k,k+1}]\in\mathfrak{M}_{\gamma}^{(2)}$ for $k\in\{0,1,\cdots,\gamma-2\}$, we have $\tilde{h}_{\ket{k}\bra{k}}-\tilde{h}_{\ket{k+1}\bra{k+1}}\propto Z^{*}$. Then, using these relations and $i[\tilde{h},Z^{*}\otimes Y_{0,1}]\in\mathfrak{M}_{\gamma}^{(2)}$, we also obtain 
$[\tilde{h}_{\ket{0}\bra{0}},Z^{*}]=0$, i.e., $\tilde{h}_{\ket{0}\bra{0}}\in \mathrm{u(dim}\mathcal{H}_A)^*$.
Thus, the form of $\tilde{h}$ is reduced to
\begin{equation}
\label{jprime_mgamma2}
\sum_{k=0}^{\gamma-2}\left(z_kZ^*\otimes Z_{k,k+1}
+
\sum_{q=k+1}^{\gamma-1}\left(x_{k,q}Z^*\otimes X_{k,q}+y_{k,q}Id_A\otimes Y_{k,q}\right)
\right)
+ (\tilde{h}_{\ket{0}\bra{0}}-z_0Z^*)\otimes Id_Q
\end{equation}
for appropriate $x_{k,q},y_{k,q},z_{k}\in\mathbb R$.  Since $(\tilde{h}_{\ket{0}\bra{0}}-z_0Z^*)\otimes Id_Q\in \hat {\mathfrak{M}}_{\gamma}^{(2)}$ as $\tilde{h}_{\ket{0}\bra{0}}\in \mathrm{u(dim}\mathcal{H}_A)^*$ and
$Z^*\otimes X_{k,q},Z^*\otimes Z_{k,q},Id_A\otimes Y_{k,q}$ are all in $\bar {\mathfrak{M}}_{\gamma}^{(2)}$, Eq. (\ref{jprime_mgamma2}) means
$h\in  {\mathcal L}(\hat J\cup \bar J)$.

\item[(iv)] $J=\mathfrak{M}_{\gamma}^{(4)}$. In accordance with the structure of $\mathfrak{M}_{\gamma}^{(4)}$ in Eq. (\ref{eq:JordanAlg_4}), any element in the set $i\tilde{G}^{(0)}$ should be written as
\begin{equation}\label{mgamma4bar_gen}
\tilde{h} := \sum_{W\in\{X,Y,Z,Id\}}\sum_{k=0}^{\gamma-1}\Bigl(
\tilde{h}_{W,\ket{k}\bra{k}}\otimes W\otimes\ket{k}\bra{k} +
\sum_{q=k+1}^{\gamma-1}(\tilde{h}_{W,X_{k,q}}\otimes W\otimes X_{k,q}+\tilde{h}_{W,Y_{k,q}}\otimes W\otimes Y_{k,q})\Bigr).
\end{equation}
Similarly as the previous cases, since $i[\tilde{h},Id\otimes Id\otimes\ket{k}\bra{k}]\in\mathfrak{M}_{\gamma}^{(4)}$ for $k\in\{0,1,\cdots,\gamma-1\}$, we obtain $\tilde{h}_{Id,X_{k,q}}=\tilde{h}_{W,Y_{k,q}}=0$, $\tilde{h}_{W,X_{k,q}}\propto Id$, and $\tilde{h}_{Id,Y_{k,q}}\propto Id$, where $W\in\{X,Y,Z\}$. In addition to these, because $i[\tilde{h},Id\otimes Id\otimes X_{k,k+1}]\in\mathfrak{M}_{\gamma}^{(4)}$ should hold for $k\in\{0,1,\cdots,\gamma-2\}$, we can have $\tilde{h}_{Id,\ket{k}\bra{k}}=\tilde{h}_{Id,\ket{k+1}\bra{k+1}}$ and $\tilde{h}_{W,\ket{k}\bra{k}}-\tilde{h}_{W,\ket{k+1}\bra{k+1}}\propto Id$, where $W\in\{X,Y,Z\}$. These relations, together with another condition, $i[\tilde{h},Id\otimes W\otimes Y_{k,k+1}]\in\mathfrak{M}_{\gamma}^{(4)}$ for $k\in\{0,1,\cdots,\gamma-2\}$ with $W\in\{X,Y,Z\}$, imply $\tilde{h}_{W,\ket{k}\bra{k}}\propto Id$. 
Since $\sum_k \tilde{h}_{Id,\ket{k}\bra{k}}\otimes Id \otimes \ket{k}\bra{k} = \tilde{h}_{Id,\ket{0}\bra{0}}\otimes Id \otimes Id \in \hat {\mathfrak{M}}_{\gamma}^{(4)}$ 
and the remaining terms in Eq. (\ref{mgamma4bar_gen}) are in $\bar {\mathfrak{M}}_{\gamma}^{(4)}$,
we can conclude $\tilde{h}\in  {\mathcal L}(\hat J\cup \bar J)$.

\item[(v)] 
$J=\mathfrak{S}_{n}$. Since it is rather hard to consider a general form of Hermitian operators $\tilde{h}$ that fulfill $i[\tilde{h},\mathfrak{S}_n]\subseteq \mathfrak{S}_n$, we shall define a larger set $\mathfrak{S}_{n^\prime,n^{\prime\prime}}$, and attempt to show $\tilde{h}$ of the form of Eq. (\ref{eq:str_c-dyn2_m_3})  will satisfy $i[\tilde{h},\mathfrak{S}_n]\subseteq \mathfrak{S}_{n^\prime n^{\prime\prime}}$. Then, we will tighten the condition for $\tilde{h}$ later to make it satisfy $i[\tilde{h},\mathfrak{S}_n]\subseteq \mathfrak{S}_n$.

Let us define the set $\mathfrak{S}_{n^\prime,n^{\prime\prime}}$ by
\begin{eqnarray}
{\mathfrak{S}}_{n',n''} &:=& \mathcal{L}(\{\overbrace{Id\otimes\cdots\otimes Id}^{n'-m}\otimes W\otimes\overbrace{Y\otimes\cdots\otimes Y}^{m-1} \cup Id\otimes\cdots\otimes Id
\nonumber\\
&& \cup U\otimes\overbrace{Y\otimes\cdots\otimes Y}^{n''}\}_{W\in\{X,Z\},U\in \mathrm{u(dim}\mathcal{H}^\prime),m\in\{0,1,\cdots,n^{\prime}-1\}}),
\nonumber\\
\end{eqnarray}
where $n^\prime=\lceil\frac{n}{2}\rceil$ and $\mathcal{H}'$ is a direct product of the first $n'-n''$ spaces. We now show by induction that any Hermitian operator $\tilde{h}$, which satisfies $i[\tilde{h},\mathfrak{S}_{n}]\subseteq\mathfrak{S}_{n',n''}$ for $0\leq n''<n'=\lceil\frac{n}{2}\rceil$, has the form:
\begin{eqnarray}
\tilde{h} &:=&   \tilde{h}_{Id}\otimes\overbrace{Id\otimes\cdots\otimes Id}^{n''}
\nonumber\\
&&{}    +\sum_{W,W^{\prime}\in\{X,Z\}}\sum_{m_{1}<m_{2}\in\{1,\cdots,n''\}} \tilde{h}_{W,W^{\prime},m_{1},m_{2}}\otimes \overbrace{Id\otimes\cdots\otimes Id}^{n''-m_{2}}
\otimes W\otimes\overbrace{Y\otimes\cdots\otimes Y}^{m_{2}-m_{1}-1}\otimes W^{\prime}\otimes\overbrace{Id\otimes\cdots\otimes Id}^{m_{1}-1}
\nonumber\\
&&{}    +\sum_{W\in\{X,Z\}}\sum_{m\in\{1,\cdots,n''\}} \tilde{h}_{W,m}\otimes\overbrace{Y\otimes\cdots\otimes Y}^{n''-m}
\otimes W\otimes\overbrace{Id\otimes\cdots\otimes Id}^{m-1}
\nonumber\\
&&{}+\sum_{m\in\{1,\cdots,n''\}} \tilde{h}_{m}\otimes \overbrace{Id\otimes\cdots\otimes Id}^{n''-m}\otimes Y\otimes\overbrace{Id\otimes\cdots\otimes Id}^{m-1},
\label{eq:str_c-dyn2_m_3}
\end{eqnarray}
where $\tilde{h}_{Id}, \tilde{h}_{W,W^{\prime},m_{1},m_{2}}, \tilde{h}_{W,m},$ and $\tilde{h}_{m}$ are the operators acting on $\mathcal{H}^\prime$.

When $n''=0$, because $\mathfrak{S}_{n',0}$ contains all unitaries $U\in \mathrm{u(dim}\mathcal{H}^\prime)$, $i[\tilde{h},\mathfrak{S}_{n}]\subseteq\mathfrak{S}_{n',n''}$ does not impose any condition on $\tilde{h}$. Thus, it can also be arbitrary unitary and it is of the form of Eq. (\ref{eq:str_c-dyn2_m_3}). Assume that the proposition holds for $n''=n_{0}''-1$, then the general form of $\tilde{h}$ that satisfies $i[\tilde{h},\mathfrak{S}_{n}]\subseteq{\mathfrak{S}}_{n',n_{0}''-1}$ should have the form of
\begin{eqnarray}\label{eq:g_Sn'n''}
\tilde{h}  &=&   \tilde{h}_{Id}'\otimes\overbrace{Id\otimes\cdots\otimes Id}^{n_{0}''}
\nonumber\\
    && {}    +\sum_{W,W^{\prime}\in\{X,Z\}}\sum_{m_{1}<m_{2}\in\{1,\cdots,n_{0}''\}} \tilde{h}'_{W,W^{\prime},m_{1},m_{2}}\otimes\overbrace{Id\otimes\cdots\otimes Id}^{n_{0}^{\prime\prime}-m_{2}}
\otimes W\otimes\overbrace{Y\otimes\cdots\otimes Y}^{m_{2}-m_{1}-1}\otimes W^{\prime}\otimes\overbrace{Id\otimes\cdots\otimes Id}^{m_{1}-1}
\nonumber\\
    && {}    +\sum_{W\in\{X,Z\}}\sum_{m\in\{1,\cdots,n_{0}''\}}
                \tilde{h}_{W,m}'\otimes\overbrace{Y\otimes\cdots\otimes Y}^{n_{0}''-m}\otimes W\otimes\overbrace{Id\otimes\cdots\otimes Id}^{m-1}
\nonumber\\
    && {}+\sum_{m\in\{1,\cdots,n_{0}''\}} \tilde{h}_{m}'\otimes \overbrace{Id\otimes\cdots\otimes Id}^{n''_0-m}\otimes Y\otimes\overbrace{Id\otimes\cdots\otimes Id}^{m-1}
\nonumber\\
    && {}    +\sum_{W\in\{X,Y,Z\}}\sum_{W',W''\in\{X,Z\}}\sum_{m_{1}<m_{2}\in\{1,\cdots,n_{0}''-1\}}\Delta \tilde{h}_{W,W',W'',m_{1},m_{2}}
\nonumber\\&& {}\quad
\otimes W \otimes\overbrace{Id\otimes\cdots\otimes Id}^{n_{0}''-m_{2}-1}
\otimes W'\otimes\overbrace{Y\otimes\cdots\otimes Y}^{m_{2}-m_{1}-1}\otimes W''\otimes\overbrace{Id\otimes\cdots\otimes Id}^{m_{1}-1}
\nonumber\\
    && {}    +\sum_{W\in\{X,Z\}}\sum_{m\in\{1,\cdots,n_{0}''-1\}}
                \Delta \tilde{h}_{Id,W,m}\otimes Id
\otimes\overbrace{Y\otimes\cdots\otimes Y}^{n_{0}''-m-1}\otimes W\otimes\overbrace{Id\otimes\cdots\otimes Id}^{m-1}
\nonumber\\
    && {}    +\sum_{W\in\{X,Y,Z\}}\sum_{m\in\{1,\cdots,n_{0}''-1\}}\Delta \tilde{h}_{W,m}\otimes W
    \otimes\overbrace{Id\otimes\cdots\otimes Id}^{n_{0}''-m-1}\otimes Y\otimes\overbrace{Id\otimes\cdots\otimes Id}^{m-1},
\end{eqnarray}
where we have split the operator for the left-most space in Eq. (\ref{eq:str_c-dyn2_m_3}) into two parts according to the tensor product structure of $\mathfrak{S}_{n}$. Since one of the spaces thereby split is two-dimensional, it can be spanned by the basis $\{Id, X, Y, Z\}$. Thus, the $\tilde{h}$ operators in Eq. (\ref{eq:str_c-dyn2_m_3}) can be written as tensor products as follows.
\begin{eqnarray}
\tilde{h}_{Id}&=& \tilde{h}_{Id}'\otimes Id+\tilde{h}_{n''_0}'\otimes Y+\sum_{W\in\{X,Z\}} \tilde{h}_{W,n''_0}'\otimes W,
\\
\tilde{h}_{W,W',m_1,m_2}&=&  \tilde{h}_{W,W',m_1,m_2}'\otimes Id+\!\!\sum_{W''\in\{X,Y,Z\}}\!\!\Delta \tilde{h}_{W'',W,W',m_1,m_2}\otimes W,
\nonumber\\
\\
\tilde{h}_{W,m}&=& \tilde{h}_{W,m}'\otimes Y+\Delta \tilde{h}_{Id,W,m}\otimes Id+ 
\!\!\!\!\!\sum_{W'\in\{X,Z\}}\!\!\!\!\! \tilde{h}_{W',W,n''_0,m}'\otimes W',
\nonumber\\
\\
\tilde{h}_{m}&=& \tilde{h}_{m}'\otimes Id+\sum_{W\in\{X,Y,Z\}}\Delta \tilde{h}_{W,m}\otimes W.
\end{eqnarray}

Due to the inclusion ${\mathfrak{S}}_{n',n''+1}\subseteq{\mathfrak{S}}_{n',n''}$, any Hermitian operator $\tilde{h}$ which satisfies $i[\tilde{h},\mathfrak{S}_{n}]\subseteq{\mathfrak{S}}_{n',n_{0}''}$ should also have the form of Eq. (\ref{eq:g_Sn'n''}). Further, $i[\tilde{h},\mathfrak{S}_{n}]\subseteq{\mathfrak{S}}_{n',n_{0}''}$ imposes additional conditions that are of help to get rid of some terms in Eq. (\ref{eq:g_Sn'n''}).
Since $i[\tilde{h},\overbrace{Id\otimes\cdots\otimes Id}^{n'-n_0''}\otimes X\otimes\overbrace{Y\otimes\cdots\otimes Y}^{n_{0}''-1}]\in{\mathfrak{S}}_{n',n_0''}$, we can have $\Delta \tilde{h}_{W,W',W'',m_1,m_2}=\Delta \tilde{h}_{Id,W',m}=\Delta \tilde{h}_{W,m}=0$ for $W\in \{Y,Z\}$ and $W',W''\in\{X,Z\}$. Similarly, $i[\tilde{h},\overbrace{Id\otimes\cdots\otimes Id}^{n'-n_0''}\otimes Z\otimes\overbrace{Y\otimes\cdots\otimes Y}^{n_{0}''-1}]\in{\mathfrak{S}}_{n',n_0''}$ implies $\Delta \tilde{h}_{X,W,W',m_1,m_2}=\Delta \tilde{h}_{X,m}=0$ for $W,W'\in \{Y,Z\}$.

Let us now consider the case where $n$ is an odd number. Since $\mathfrak{S}_{2n'-1}\subseteq{\mathfrak{S}}_{n',n'-1}$, any Hermitian operator $\tilde{h}$ which satisfies $i[\tilde{h},\mathfrak{S}_{2n'-1}]\subseteq\mathfrak{S}_{2n'-1}$ can be written as Eq. (\ref{eq:str_c-dyn2_m_3}) for  $n''=n'-1$.
Because $i[\tilde{h},\overbrace{Id\otimes\cdots\otimes Id}^{n'-m-1}\otimes W\otimes\overbrace{Y\otimes\cdots\otimes Y}^{m-1}]\in\mathfrak{S}_{2n'-1}$ for $m\in\{1,2,\cdots,n'-1\}$ and $W\in\{X,Z\}$, both $\tilde{h}_{W_1,W_2,m_1,m_2}$ and $\tilde{h}_{m}$ are proportional to $Id$ and $\tilde{h}_{W,m}=0$. Therefore, such Hermitian operators $\tilde{h}$ are in $\mathcal L(\hat J, \bar J)$.

When $n$ is even,  any Hermitian operator $\tilde{h}$ satisfying $i[\tilde{h},\mathfrak{S}_{2n'}]\subseteq\mathfrak{S}_{2n'}$ has the form of Eq. (\ref{eq:str_c-dyn2_m_3}) for $n''=n'-1$, because $\mathfrak{S}_{2n'}\subseteq{\mathfrak{S}}_{n',n'-1}$.
The condition $i[\tilde{h},\overbrace{Id\otimes\cdots\otimes Id}^{n'-m-1}\otimes W\otimes\overbrace{Y\otimes\cdots\otimes Y}^{m-1}]\in\mathfrak{S}_{2n'}$ for $m\in \{1,2,\cdots,n'-1\}$ then implies $\tilde{h}_{W_1,W_2,m_1,m_2}\propto Id, \tilde{h}_{m}\propto Id$ and $\tilde{h}_{W,m}\propto Z^*$, where  $W,W_1,W_2\in\{X,Z\}$.
Another one, $i[\tilde{h}, Z^*\otimes\overbrace{Y\otimes\cdots\otimes Y}^{n'-1}]=0$, leads to $[\tilde{h}_{Id},Z^{*}]\subseteq\mathfrak{S}_{2n'}$. Hence, such Hermitian operators $\tilde{h}$ are in $\mathcal L(\hat J\cup \bar J)$, and the proposition of the lemma has been proved.
\end{enumerate}
\end{proof}

The next lemma demonstrates that,  we can easily express a set of all the Hermitian operators that commute with a simple Jordan algebra.


\begin{lemma}
\label{lemm:orthogonal_subspace}
Suppose that a pair $(J, \hat J)$ of sets of  Hermitian operators on $\mathcal H$ is equal to
$({\mathfrak R}, \hat{\mathfrak R})$,
$({\mathfrak M}_\gamma^{(k)}, \hat{\mathfrak M}_\gamma^{(k)})$ or
$({\mathfrak S}_n, \hat{\mathfrak S}_n)$ for $\gamma\geq3$, $k\in\{1,2,4\}$ and $n\geq3$.
If we define $J^\prime$ to be the set of all the Hermitian operators that commute with the simple Jordan algebra $J$, i.e., 
\begin{equation}
J^\prime:=\{h^\prime|h^\prime\in i\cdot {\rm u}({\rm dim}\mathcal H)\land \forall h\in J,\;[h^\prime,h]= 0\},
\end{equation}
then $J^\prime$ is equal to $\hat J$.
\end{lemma}
Here, $\mathcal{H}$ is again the support of the largest projection operator in $J$.

\begin{proof}
It is straightforward to verify $[\hat{J},J]=\{0\}$ by using the explicit forms of these algebras, Eqs. (\ref{eq:str_hat_1})-(\ref{eq:str_hat_6}) and Eqs. (\ref{eq:JordanAlg_1})-(\ref{eq:JordanAlg_6}), thus $\hat J\subseteq J^\prime$.
So, in the following, we show the inclusion of the opposite direction, that is, if a Hermitian operator $h^\prime$ commute with any operator in $J$, then 
$h^\prime\in \hat{J}$. Let us examine each case, depending on the type of $J$. Below, the range of $k$ is $\{0, 1, ..., \gamma-1\}$, unless it is for the Pauli operators for which its range is $\{0, 1, ..., \gamma-2\}$.

\begin{enumerate}
\item[(i)] $J=\mathfrak{R}$. Any operator commutes with $Id_A$, so $J'$ is ${\rm u}({\rm dim}\mathcal H)$, which is the same as $\hat{\mathfrak{R}}$. Thus, $ \hat J\supseteq J'$.

\item[(ii)] $J=\mathfrak{M}_{\gamma}^{(1)}$ or $\mathfrak{M}_{\gamma}^{(2)}$. any element in the set $J'$
can be written as
\begin{equation}
h^\prime := h^\prime_{Id}\otimes Id+
\sum_{k=0}^{\gamma-2}\!
\left(\!h^\prime_{Z_{k,k+1}}\otimes Z_{k,k+1}
+
\!\!\!\sum_{q=k+1}^{\gamma-1}\!\!\!
\left(h^\prime_{X_{k,q}}\otimes X_{k,q}+h^\prime_{Y_{k,q}}\otimes Y_{k,q}\right)\!\right)\!,
\label{eq:str_c-dyn2_m_2_sub}
\end{equation}
where $h^\prime_W$ is a Hermitian operator on the first space that makes a pair with the operator $W$ on the second space.
Because of the condition, $[h^\prime,Id\otimes\ket{k}\bra{k}]=0$, we have $h^\prime_{X_{k,q}}=h^\prime_{Y_{k,q}}=0$.
With this and another condition, $[h^\prime,Id\otimes X_{k,k+1}]=0$, we obtain $2h^\prime_{Z_{k,k+1}}=h^\prime_{Z_{k-1,k}}+h^\prime_{Z_{k+1,k+2}}$, where we set $h^\prime_{Z_{-1,0}}=h^\prime_{Z_{\gamma-1,\gamma}}=0$. This is enough to conclude $h^\prime_{Z_{k,k+1}}=0$, and thus we have $h^\prime=h^\prime_{Id}\otimes Id$. Therefore,  in the case of $J=\mathfrak{M}_{\gamma}^{(1)}$, $h^\prime\in \hat J=\hat{\mathfrak{M}}_\gamma^{(1)}$ holds.
When $J=\mathfrak{M}_{\gamma}^{(2)}$, an additional condition, $[h^\prime,Z^{*}\otimes Y_{0,1}]=0$, leads to $[h^\prime_{Id},Z^{*}]=0$, and thus $h^\prime\in \hat {\mathfrak{M}}_\gamma^{(2)}$.

\item[(iii)] $J=\mathfrak{M}_{\gamma}^{(4)}$. Any element in the set $J'$ can be written as
\begin{eqnarray}
h^\prime &:=&   \sum_{W\in\{X,Y,Z,Id\}}  \left(h^\prime_{W,Id}\otimes W\otimes Id+\sum_{k=0}^{\gamma-2}\Big(h^\prime_{W,Z_{k,k+1}}\otimes W\otimes Z_{k,k+1}\right.
\nonumber\\
&&{}    +\left.\sum_{q=k+1}^{\gamma-1}\left(h^\prime_{W,X_{k,q}}\otimes W\otimes X_{k,q}+h^\prime_{W,Y_{k,q}}\otimes W\otimes Y_{k,q}\right)\Big)\right),
\nonumber\\
\end{eqnarray}
where $h^\prime_{W,W'}$ is a Hermitian operator on the first space that makes a tensor product with $W$ and $W'$. The commutation condition $[h^\prime,Id\otimes Id\otimes\ket{k}\bra{k}]=0$ implies $h^\prime_{W,X_{k,q}}=h^\prime_{W,Y_{k,q}}=0$ for $W\in\{Id,X,Y,Z\}$. Together with this and $[h^\prime,Id\otimes Id\otimes X_{k,k+1}]=0$, we get $2h^\prime_{W,Z_{k,k+1}}=h^\prime_{W,Z_{k-1,k}}+h^\prime_{W,Z_{k+1,k+2}}$, where we set $h^\prime_{W,Z_{-1,0}}=h^\prime_{W,Z_{\gamma-1,\gamma}}=0$. These allow us to obtain $h^\prime_{W,Z_{k,k+1}}=0$.
Further, another commutation, $[h^\prime,Id\otimes W\otimes Y_{k,k+1}]\in\mathfrak{M}_{\gamma}^{(4)}$ for $W\in\{X,Y,Z\}$, as well as the relations obtained above, lead to $h^\prime_{W,Id}=0$ for $W\in\{X,Y,Z\}$. Therefore, the remaining term is only $h^\prime_{Id,Id}\otimes Id\otimes Id$ and this is obviously in $\hat{\mathfrak{M}}_{\gamma}^{(4)}$.

\item[(iv)] $J=\mathfrak{S}_{n}$. As in the proof of Lemma \ref{lemm:pseudo_Ideal}, instead of considering the general form of operators in $J^\prime$, we first define a larger set $\mathfrak{S}_{n^\prime,n^{\prime\prime}}$, and find the form of $h^\prime$ that commutes with any operator in $\mathfrak{S}_{n^\prime,n^{\prime\prime}}$. Then, we will tighten the condition to have the set of $h^\prime$ that meets the condition $[h^\prime,h]=0$ for all $h\in \mathfrak{S}_{n^\prime,n^{\prime\prime}}$. 

Let us define 
\begin{equation}
{\mathfrak{S}}_{n',n''} := \mathcal{L}(\{\overbrace{Id\otimes\cdots\otimes Id}^{n'-m}\otimes W\otimes\overbrace{Y\otimes\cdots\otimes Y}^{m-1}\}_{W\in\{X,Z\},m\in\{0,1,...,n''\}}),
\end{equation}
where $n^\prime=\lceil\frac{n}{2}\rceil$. We prove by induction that any Hermitian operator $h^\prime$ that satisfies $[h^\prime,h]=0$ for any $h\in \mathfrak{S}_{n',n''}$ has the following form:
\begin{equation}
h^\prime  := h^\prime_{Id}\otimes\overbrace{Id\otimes\cdots\otimes Id}^{n''},
\label{eq:str_c-dyn2_m_3_sub}
\end{equation}
where $h^\prime_{Id}$ is a Hermitian operator acting on the direct product of the first $n'-n''$ spaces and $0\leq n''<n'$.

When $n''=0$, $\mathfrak{S}_{n^\prime,n^{\prime\prime}}=\overbrace{Id\otimes\cdots\otimes Id}^{n'}$, thus the proposition trivially holds. Assume that $h^\prime$ has the form of Eq. (\ref{eq:str_c-dyn2_m_3_sub}) when $n''=n_{0}''-1$. 
When $n''=n_{0}''$, any Hermitian operator $h^\prime$ that commutes with any $h\in \mathfrak{S}_{n',n_{0}''}$ should have the following form:
\begin{equation}
h^\prime := h^{\prime\prime}_{Id}\otimes\overbrace{Id\otimes\cdots\otimes Id}^{n_{0}''}
    +\sum_{W\in\{X,Y,Z\}}\Delta h^\prime_{W,Id}\otimes W\otimes\overbrace{Id\otimes\cdots\otimes Id}^{n_{0}''-1}.
\end{equation}
Comparing with Eq. (\ref{eq:str_c-dyn2_m_3_sub}), we see 
\begin{equation}
h^\prime_{Id} = h^{\prime\prime}_{Id}\otimes Id+\sum_{W\in\{X,Y,Z\}}\Delta h^\prime_{W,Id}\otimes W.
\end{equation}
While $[h^{\prime\prime}_{Id}\otimes \overbrace{Id\otimes\cdots\otimes Id}^{n_0''},\overbrace{Id\otimes\cdots\otimes Id}^{n'-n_0''}\otimes W\otimes\overbrace{Y\otimes\cdots\otimes Y}^{n_{0}''-1}]=0$ for $W\in\{X,Z\}$, in order for $[\sum_{V}\Delta h^\prime_{V,Id}\otimes V\otimes\overbrace{Id\otimes\cdots\otimes Id}^{n_{0}''-1},\overbrace{Id\otimes\cdots\otimes Id}^{n'-n_0''}\otimes W\otimes\overbrace{Y\otimes\cdots\otimes Y}^{n_{0}''-1}]$ to be zero, $\Delta h^\prime_{V,Id}=0$ for any $V\in\{X, Y, Z\}$.

We now consider the case of odd $n=2n^\prime-1$. Any Hermitian operator $h^\prime$ that commutes with any $h\in \mathfrak{S}_{2n'-1}$
can be written as Eq. (\ref{eq:str_c-dyn2_m_3_sub}) for  $n''=n'-1$ since $\mathfrak{S}_{2n'-1}\supseteq{\mathfrak{S}}_{n',n'-1}$, which means that such an operator $h^\prime$ is in $\hat{\mathfrak{S}}_{2n^\prime-1}$.

Similarly, for even $n=2n^\prime$, an operator $h^\prime$ that commutes with any $h\in \mathfrak{S}_{2n'}$ should have the form of Eq. (\ref{eq:str_c-dyn2_m_3_sub}) with $n''=n'-1$, since $\mathfrak{S}_{2n'}\supseteq{\mathfrak{S}}_{n',n'-1}$.
Due to the commutation condition, $[h^\prime, Z^*\otimes\overbrace{Y\otimes\cdots\otimes Y}^{n'-1}]=0$, we have a constraint $[h^\prime_{Id},Z^{*}]=0$. Therefore, $h^\prime\in \hat{\mathfrak{S}}_{2n^\prime}$.
\end{enumerate}
\end{proof}

When ${\rm dim}\mathcal H_S=2$, it is not hard to specify the largest and the smallest possible Lie algebras for a given Jordan algebra $iG^{(1)}$. It can be done thanks to Lemmas \ref{lemm:pseudo_Ideal} and \ref{lemm:orthogonal_subspace}, as well as the inclusion relations that identifiers $iG^{(0)}$ and $iG^{(1)}$ fulfill. This fact is of help for identifying the disconnected and connected algebras, as stated in the following lemma.
 
\begin{lemma}
\label{sec:abstract_structure}
Suppose that a triple $(J_j, \bar{J}_j, \hat{J}_j)$ of algebras on $\mathcal H_{E_j}$ is equal to either one of the following three;
$({\mathfrak R}, \bar{\mathfrak R}, \hat{\mathfrak R})$,
$({\mathfrak M}_\gamma^{(k)}, \bar{\mathfrak M}_\gamma^{(k)}, \hat{\mathfrak M}_\gamma^{(k)})$ or
$({\mathfrak S}_n, \bar{\mathfrak S}_n, \hat{\mathfrak S}_n)$, where $\gamma\geq3$, $k\in\{1,2,4\}$ and $n\geq3$ (for each $j$,). 
When $\mathrm{dim}\mathcal H_S=2$, if the relation
\begin{equation}
\bigoplus_j{\mathcal L}(i\bar J_j\otimes \{Id_S\}\cup  J_j\otimes {\rm su}({\rm dim}\mathcal H_S))
\subseteq L \subseteq
\bigoplus_j{\mathcal L}(i   \hat J_j\otimes \{Id_S\}\cup i\bar J_j\otimes \{Id_S\}\cup J_j\otimes {\rm su}({\rm dim}\mathcal H_S))
\label{eq:abstract_structure_1}
\end{equation}
holds and $L$ is a Lie algebra, the disconnected and the connected algebras can be written as
\begin{eqnarray}
L_d &=& \bigoplus_j i\hat J_j\otimes \{Id_S\}
\label{eq:abstract_structure_2}
\\
L_c &=& \bigoplus_j{\mathcal L}(i  \bar J_j\otimes \{Id_S\}\cup J_j\otimes {\rm su}({\rm dim}\mathcal H_S))
\label{eq:abstract_structure_3}
\end{eqnarray}
\end{lemma}

\begin{proof}
For the sake of convenience, we let $\tilde{L}_d$ and $\tilde{L}_c$ denote the RHSs of Eqs. (\ref{eq:abstract_structure_2}) and (\ref{eq:abstract_structure_3}), respectively. From the definition of the connected Lie algebra, we see
\begin{eqnarray}
L_c &\supseteq& \mathcal L([\{Id_E\}\otimes {\rm su}({\rm dim}\mathcal H_S),L])
\nonumber\\
&\supseteq& 
\mathcal L([\{Id_E\}\otimes {\rm su}({\rm dim}\mathcal H_S),
\bigoplus_j J_j\otimes {\rm su}({\rm dim}\mathcal H_S)
])
\nonumber\\
&\supseteq &
\bigoplus_j J_j\otimes {\rm su}({\rm dim}\mathcal H_S).
\label{eq:abstract_structure_6}
\end{eqnarray}

In the second inclusion, we have used the first relation of the assumption Eq. (\ref{eq:abstract_structure_1}).
Using this, and since $L_c$ is a linear space that is closed under the commutator, we have
\begin{eqnarray}
L_c &\supseteq& \mathcal L([L_c,L_c])
\nonumber\\
&\supseteq &\mathcal L([\bigoplus_j J_j\otimes \{iX\},\bigoplus_j J_j\otimes \{iX\}]),
\nonumber\\
&=& 
\bigoplus_j i\bar J_j\otimes \{Id\}
\label{eq:abstract_structure_7}
\end{eqnarray}
where we have used $i\mathcal L([J_j,J_j])=\bar J_j$ (Lemma \ref{lemm:relation}).
Equations (\ref{eq:abstract_structure_6}) and (\ref{eq:abstract_structure_7}) imply $L_c\supseteq \tilde L_c$. 

The inclusion of the opposite direction $L_c\subseteq \tilde L_c$ can be shown as follows.
\begin{eqnarray}
&&[\tilde L_c,L]
\nonumber\\
&&\subseteq  [\tilde L_c,\bigoplus_j{\mathcal L}(i   \hat J_j\otimes \{Id_S\}\cup i\bar J_j\otimes \{Id_S\}\cup J_j\otimes {\rm su}({\rm dim}\mathcal H_S))]
\nonumber\\
&&\subseteq \bigoplus_j{\mathcal L}(
 [\bar J_j,  \hat J_j]\otimes \{Id_S\}\cup
i[       J_j,  \hat J_j]\otimes {\rm su}({\rm dim}\mathcal H_S)\cup 
 [\bar J_j,  \bar J_j]\otimes \{Id_S\}
\nonumber\\
&& \;\;\; \cup 
i[       J_j,  \bar J_j]\otimes {\rm su}({\rm dim}\mathcal H_S)\cup 
[       J_j,         J_j]\otimes \{Id_S\}\cup 
\{      J_j,        J_j\}\otimes {\rm su}({\rm dim}\mathcal H_S))
\nonumber\\
&&= \bigoplus_j{\mathcal L}(
i \bar J_j\otimes \{Id_S\}\cup
       J_j\otimes {\rm su}({\rm dim}\mathcal H_S)) 
\nonumber \\
&&= \tilde L_c
\label{eq:abstract_structure_4}
\end{eqnarray}
The first inclusion is simply from Eq. (\ref{eq:abstract_structure_1}).
The second inclusion is due to a generic inclusion $[A\otimes {\rm su}({\rm dim}\mathcal H_S),B\otimes {\rm su}({\rm dim}\mathcal H_S)] \subseteq \mathcal L([A,B]\otimes \{Id\}\cup \{A,B\}\otimes {\rm su}({\rm dim}\mathcal H_S))$, which is valid for arbitrary operator sets $A$ and $B$ when dim$\mathcal{H_S}=2$.
The third relation comes from the results of Lemma \ref{lemm:relation}, namely,
 $[\bar J_j,\hat J_j]=[ J_j,\hat J_j]=\{0\}$,
 $i[\bar J_j,\bar J_j]\subseteq i\mathcal L([J_j,J_j])= \bar J_j$, and
 $i[ J_j,\bar J_j]\subseteq\{J_j,J_j\}= J_j$. 
Equation (\ref{eq:abstract_structure_4}) and $\tilde L_c\subseteq L$ imply that $\tilde L_c$ is an ideal of $L$, which means $L_c\subseteq \tilde L_c$. Hence, Eq. (\ref{eq:abstract_structure_3}) is proved.

Next, we show Eq. (\ref{eq:abstract_structure_2}), which is about the disconnected algebra $L_d$.
It is straightforward to verify $[\tilde L_d,L_c]=\{0\}$ by noting $[\bar J_j,\hat J_j]=[ J_j,\hat J_j]=\{0\}$ (Lemma \ref{lemm:relation}), thus we have $\tilde L_d\subseteq L_d$. The opposite is shown as follows:
\begin{eqnarray}
L_{d}&\subseteq&
\{g|g\in  {\rm u}({\rm dim}\mathcal H_{E}\; {\rm dim}\mathcal H_S )\land \forall g'\in (\bigoplus_j J_j)\otimes {\rm su}({\rm dim}\mathcal H_S), [g,g']=0 \}.
\nonumber\\
&=&
i\{h|h\in i\cdot{\rm u}({\rm dim}\mathcal H_{E} )\land \forall h^\prime\in  \bigoplus_j J_j, [h,h^\prime]=0 \}\otimes \{Id_S\}.
\nonumber\\
&=& \bigoplus_j i\{h|h\in  i\cdot{\rm u}({\rm dim}\mathcal H_{E_j} )\land \forall h^\prime\in  J_j, [h,h^\prime]=0 \}\otimes \{Id_S\}.
\nonumber \\
&=& 
\bigoplus_j i\hat J_j \otimes \{Id_S\}
\nonumber \\
&=& \tilde L_d
\end{eqnarray}
The first relation is a paraphrase of the definition of $L_d$, and the second and the third equalities are justified since $Id_{E}$ is in $\bigoplus_j J_j$ and $Id_{E_j}\in J_j$, respectively. The fourth is due to Lemma  \ref{lemm:orthogonal_subspace}.
\end{proof}

In the case of ${\rm dim}\mathcal H_S\ge 3$, the structure of $G^{(1)}$ is simple.
Therefore, as in the case of dim$\mathcal{H}_S=2$, the largest and smallest Lie algebras for a given $G^{(1)}$ can be obtained, and this constraint enables us to identify the disconnected and connected algebras as follows.

\begin{lemma}
\label{sec:abstract_structure_lemm10}
When ${\rm dim}\mathcal H_S\ge 3$, if the relation
\begin{equation}
\bigoplus_j \{Id_{B_j}\}\otimes {\rm su}({\rm dim}\mathcal H_{R_j}\cdot {\rm dim}\mathcal H_S\}
\subseteq L  \subseteq \bigoplus_j {\mathcal L}({\rm u}(B_j)\otimes \{Id_{R_j}\otimes Id_S\}\cup \{Id_{B_j}\}\otimes {\rm su}({\rm dim}\mathcal H_{R_j}\cdot {\rm dim}\mathcal H_S\})
\label{eq:abstract_structure_lemm10_0}
\end{equation}
holds and $L$ is a Lie algebra,
the disconnected and the connected algebras can be written as
\begin{eqnarray}
L_d &=& \bigoplus_j {\rm u}(B_j)\otimes \{Id_{R_j}\otimes Id_S\},
\label{eq:abstract_structure_lemm10_1}\\
L_c &=& \bigoplus_j \{Id_{B_j}\}\otimes {\rm su}({\rm dim}\mathcal H_{R_j}\cdot {\rm dim}\mathcal H_S\}.
\label{eq:abstract_structure_lemm10_2}
\end{eqnarray}
\end{lemma}

\begin{proof}
We shall let $\tilde{L}$, $\tilde{L}_d$, and $\tilde{L}_c$ denote the RHSs of Eqs. (\ref{eq:abstract_structure_lemm10_0}), (\ref{eq:abstract_structure_lemm10_1}), and (\ref{eq:abstract_structure_lemm10_2}), respectively, to simplify the notations for the proof.
From the definition of the connected Lie algebra, we have
\begin{eqnarray}
L_c&\supseteq&  \mathcal L([\{Id_E\}\otimes {\rm su}({\rm dim}\mathcal H_S),L])
\nonumber\\
&\supseteq&  
\mathcal L([\{Id_E\}\otimes {\rm su}({\rm dim}\mathcal H_S),
\bigoplus_j i\{Id_{B_j}\}\otimes {\rm u}({\rm dim}\mathcal H_{R_j})\otimes {\rm su}({\rm dim}\mathcal H_S)
])
\nonumber\\
&\supseteq & 
\bigoplus_j i\{Id_{B_j}\}\otimes {\rm u}({\rm dim}\mathcal H_{R_j})\otimes {\rm su}({\rm dim}\mathcal H_S).
\label{eq:abstract_structure_lemm10_4}
\end{eqnarray}
In the second line, we have used the first relation of Eq. (\ref{eq:abstract_structure_lemm10_0}).
Since, $L_c$ is closed under the commutator and is a linear space,
\begin{eqnarray}
L_c&\supseteq& \mathcal L([L_c,L_c])
\nonumber\\
&\supseteq&  \mathcal L([
\bigoplus_j i\{Id_{B_j}\}\otimes {\rm u}({\rm dim}\mathcal H_{R_j})\otimes  \{iX_{0,1}\},
\bigoplus_j i\{Id_{B_j}\}\otimes {\rm u}({\rm dim}\mathcal H_{R_j})\otimes  \{iX_{0,1}\}])
\nonumber\\
&=&  
\bigoplus_j i\{Id_{B_j}\}\otimes {\rm su}({\rm dim}\mathcal H_{R_j})\otimes  \{i(\ket0\bra0+\ket1\bra1)\}.
\label{eq:abstract_structure_lemm10_5}
\end{eqnarray}
We have used Eq. (\ref{eq:abstract_structure_lemm10_4}) in the second inclusion relation, and the fact $i\mathcal L([{\rm u}(N),{\rm u}(N)])={\rm su}(N)$ in the last step.

From Eqs. (\ref{eq:abstract_structure_lemm10_4}) and (\ref{eq:abstract_structure_lemm10_5}) and the fact that $L_c$ is a linear space, we know $L_c\supseteq \tilde L_c$. 
The opposite inclusion, $L_c\subseteq \tilde L_c$, is also guaranteed by the fact that $\tilde L_c$ includes the set 
$\{Id_E\}\otimes{\rm su}({\rm dim}\mathcal H_S)$, and that $\tilde L_c$ is an ideal of $L$, i.e. $[\tilde L_c,L]\subseteq [\tilde L_c,\tilde L]\subseteq \tilde L_c$.
As a result, Eq. (\ref{eq:abstract_structure_lemm10_2}) is derived.

Next, let us turn to $L_d$. Recall the definition of $L_d$ in Eq. (\ref{Ld_def}), viz., $L_d := \{g|g\in {\rm u}({\rm dim}\mathcal H_E\cdot {\rm dim}\mathcal H_S )\land \forall g^\prime\in L_c, [g,g^\prime]=0 \}$. Then, from the definition of $\tilde L_d$, which is the RHS of Eq. (\ref{eq:abstract_structure_lemm10_1}), we have $[\tilde L_d,L_c]=\{0\}$, thus $\tilde L_d\subseteq L_d$. The opposite inclusion is shown as follows.
\begin{eqnarray}
L_{d} &\subseteq& 
\{g|g\in  {\rm u}({\rm dim}\mathcal H_{E}\cdot {\rm dim}\mathcal H_S )
\land \forall g^\prime\in (\bigoplus_j 
\{Id_{B_j}\}\otimes{\rm u}({\rm dim}\mathcal H_{R_j}) )\otimes {\rm su}({\rm dim}\mathcal H_S), [g,g^\prime]=0 \}.
\nonumber\\
&=& 
\{g|g\in  {\rm u}({\rm dim}\mathcal H_{E} )
\land \forall g^\prime \in \bigoplus_j \{Id_{B_j}\}\otimes{\rm u}({\rm dim}\mathcal H_{R_j}), [g,g^\prime]=0 \}\otimes \{Id_S\}.
\nonumber\\
&=& 
\bigoplus_j \{g|g\in  {\rm u}({\rm dim}\mathcal H_{B_j}\cdot {\rm dim}\mathcal H_{R_j})
\land \forall g^\prime\in  \{Id_{B_j}\}\otimes{\rm u}({\rm dim}\mathcal H_{R_j}), [g,g^\prime]=0 \}\otimes \{Id_S\}.
\nonumber\\
&=& 
\bigoplus_j  {\rm u}({\rm dim}\mathcal H_{B_j})\otimes\{Id_{R_j}\otimes Id_S\} \nonumber \\
&=& \tilde L_d.
\end{eqnarray}
The second and the third inclusions are results of $Id_S\in \mathrm{u(dim}\mathcal{H}_S)$ and $Id_{R_j}\in {\rm u}({\rm dim}\mathcal H_{R_j})$, respectively.
\end{proof}

As we have seen in Theorems \ref{th2:str_ds3_dem} and \ref{th3:str_ds2_dem}, the space $\mathcal{H}_E$ can have a structure of either Eq. (\ref{eq:structure_e_ds3}) or Eq. (\ref{eq:structure_e_ds2}), when dim$\mathcal{H}_S\ge 3$ or $=2$, respectively. Let us now consider the situation in which an additional space $S^\prime$ is attached to 2-dimensional $S$. While $\mathcal{H}_E$ has a structure of Eq. (\ref{eq:structure_e_ds2}) because dim$\mathcal{H}_S=2$, it can also have a structure of Eq. (\ref{eq:structure_e_ds3}) if we regard $SS^\prime$ as a single space whose dimensionality is higher than 4 (because dim$\mathcal{H}_{S^\prime}\ge 2$). This means that these two structures coexist in this case we can give two structures for $\mathcal H_E$ depending on an operator $H$ acting on $\mathcal H_E\otimes \mathcal H_S$ when ${\rm dim}\mathcal H_S=2$.
The following lemma is useful for understanding the relation between the two types of structures of $\mathcal{H}_E$, and thus for proving Theorem \ref{th:str_rel}.


\begin{lemma}
\label{sec:str_mod}
Consider pairs of algebras, $(J, J^\prime)$, each of which is taken from Eqs. (\ref{eq:JordanAlg_1})-(\ref{eq:JordanAlg_6}) and Eqs. (\ref{eq:def_J'_R})-(\ref{eq:def_J'_Seven}) so that there is a correspondence in the type of algebras, such as $(\mathfrak{M}_\gamma^{(4)}, \mathfrak{M}_\gamma^{\prime(4)})$. 
Then, $J^\prime$ is the smallest set of Hermitian operators that is closed under commutator and anticommutator, while containing the set $J$.
\end{lemma}

For the sake of clarity, let us re-list Eqs. (\ref{eq:JordanAlg_1})-(\ref{eq:JordanAlg_6}) and (\ref{eq:def_J'_R})-(\ref{eq:def_J'_Seven}) here again:
\setcounter{equation}{24}
\begin{eqnarray}
\mathfrak{R} &:=&  \{Id_A\}, 
\\
\mathfrak{M}_{\gamma}^{(1)} &:=&   \mathcal{L} (\{Id_A\otimes X_{k,q}, Id_A\otimes\ket{k}\bra{k}\}_{k\neq q\in\{0,1,\cdots,\gamma-1\}}), 
\\
\mathfrak{M}_{\gamma}^{(2)} &:=&   \mathcal{L} ( \{Id_A\otimes X_{k,q}, Id_A\otimes\ket{k}\bra{k}, Z^{*}\otimes Y_{k,q}\}_{k\neq q\in\{0,1,\cdots,\gamma-1\}}), 
\\
\mathfrak{M}_{\gamma}^{(4)} &:=&   \mathcal{L} ( \{Id_A\otimes Id_{Q^{(1)}}\otimes X_{k,q},
 Id_A\otimes Id_{Q^{(1)}}\otimes\ket{k}\bra{k},
 Id_A\otimes W\otimes Y_{k,q}\}_{W\in\{X,Y,Z\},\;k\neq q\in\{0,1,\cdots,\gamma-1\}}), 
\\
\mathfrak{S}_{2n'-1}
&:=&   \mathcal{L} ( \{ \overbrace{Id\otimes\cdots\otimes Id}^{n'-m}\otimes W\otimes\overbrace{Y\otimes\cdots\otimes Y}^{m-1},
 Id\otimes\cdots\otimes Id\}_{W\in\{X,Z\},m\in\{1,2,\cdots,n'-1\}}), 
\\
\mathfrak{S}_{2n'}
&:=&   \mathcal{L} ( \{ \overbrace{Id\otimes\cdots\otimes Id}^{n'-m}\otimes W\otimes\overbrace{Y\otimes\cdots\otimes Y}^{m-1}, Id\otimes\cdots\otimes Id,
 Z^{*}\otimes\overbrace{Y\otimes\cdots\otimes Y}^{n'-1}\}_{W\in\{X,Z\},m\in\{1,2,\cdots,n'-1\}}), 
\end{eqnarray}
and 
\setcounter{equation}{64}
\begin{eqnarray}
{\mathfrak{R}^\prime} &=&   \{Id_A\},
\\
{\mathfrak{M}}_{\gamma}^{\prime(1)}
&=& i   \{Id_A\}\otimes    {\rm u}({\rm dim}\mathcal H_Q), 
\\
{\mathfrak{M}}_{\gamma}^{\prime(2)}
&=&
i\mathcal L((\{Id_{A^{(+1)}}\}\oplus
\{Id_{A^{(-1)}}\})\otimes {\rm u}({\rm dim}\mathcal H_Q)), 
\\
{\mathfrak{M}}_{\gamma}^{\prime(4)}
&=& i   \{Id_A\}\otimes  {\rm u}({\rm dim}\mathcal H_{Q^{(1)}}\cdot {\rm dim}\mathcal H_{Q}), 
\\
{\mathfrak{S}}_{2n'-1}'
&=& i   \{Id_A\}\otimes
{\rm u}({\rm dim}\mathcal H_{Q^{(n'-1)}}\cdot {\rm dim}\mathcal H_{Q^{(n^\prime-2)}}\cdots {\rm dim}\mathcal H_{Q^{(1)}}), 
\\
{\mathfrak{S}}_{2n'}'
&=&
i\mathcal L((\{Id_{A^{(+1)}}\}\oplus
\{Id_{A^{(-1)}}\})
\otimes {\rm u}({\rm dim}\mathcal H_{Q^{(n^\prime-1)}}\cdot {\rm dim}\mathcal H_{Q^{(n^\prime-2)}}\cdots {\rm dim}\mathcal H_{Q^{(1)}})). 
\end{eqnarray}
\setcounter{equation}{197}

\begin{proof}
It is almost trivial to see $J\subseteq J^\prime$, $i[J^\prime,J^\prime]\subseteq J'$, and $\{J^\prime,J^\prime\}\subseteq J^\prime$, from the definitions above. 
Therefore, to prove the lemma, it is enough if we check that any element of the basis of $J^\prime$ can be generated by $J$.
We will verify below that this proposition holds for all instances of $(J, J^\prime)$. 

\begin{enumerate}
\item[(i)] $J=\mathfrak{R}$. This case is trivial because $\mathfrak{R}=\mathfrak{R}^\prime$.

\item[(ii)] $J=\mathfrak{M}_{\gamma}^{(1)}$. Clearly, $X_{k,q}$ and $\ket{k}\bra{k}$ in the $Q$ space are sufficient to span u(dim$\mathcal{H}_Q$) in Eq. (\ref{eq:def_J'_M1}). 

\item[(iii)] $J=\mathfrak{M}_{\gamma}^{(2)}$. 
Note that $\{Z^*\otimes X_{k,q}, Z^*\otimes\ket{k}\bra{k},Id_A\otimes Y_{k,q}\}$ are in $\mathfrak{M}_\gamma^{(2)}$, since 
\begin{eqnarray}\label{eq:Mgamma2AndMgamma2prime}
Z^*\otimes X_{k,q} &=&  i[Id_A\otimes\ket{k}\bra{k},Z^*\otimes Y_{k,q}],
\nonumber\\
Z^*\otimes\ket{k}\bra{k}  &=& -\frac 14 \{i[Id_A\otimes X_{k,q},Z^{*}\otimes Y_{k,q}],Id_A\otimes \ket k\bra k\},
\nonumber\\
Id_A\otimes Y_{k,q}   &=& i[Id_A\otimes X_{k,q},Id_A\otimes\ket{k}\bra{k}].
\end{eqnarray}
It follows from Eq. (\ref{eq:Mgamma2AndMgamma2prime}) and Eq. (\ref{eq:JordanAlg_3}) that it is possible to span $\mathfrak{M}_\gamma^{\prime (2)}$ by elements in $\mathfrak{M}_\gamma^{(2)}$.

\item[(iv)] $J=\mathfrak{M}_{\gamma}^{(4)}$. 
Similarly, we see that 
\begin{eqnarray}
Id_A\otimes W\otimes X_{k,q}  
&=& i[Id_A\otimes Id_{Q^{(1)}}\otimes\ket{k}\bra{k},Id_A\otimes W\otimes Y_{k,q}],
\nonumber\\
Id_A\otimes W\otimes\ket{k}\bra{k}
&=& -\frac14\{i[Id_A\otimes Id_{Q^{(1)}}\otimes X_{k,q},Id_A\otimes W\otimes Y_{k,q}],Id_A\otimes Id_{Q^{(1)}}\otimes\ket{k}\bra{k}\},
\nonumber\\
Id_A\otimes Id_{Q^{(1)}}\otimes Y_{k,q}
&=& i[Id_A\otimes Id_{Q^{(1)}}\otimes X_{k,q},Id_A\otimes Id_{Q^{(1)}}\otimes\ket{k}\bra{k}]
\end{eqnarray}
with $W\in \{X,Y,Z\}$ are in $\mathfrak{M}_{\gamma}^{(4)}$, thus, together with Eq. (\ref{eq:JordanAlg_4}), $\mathfrak{M}_{\gamma}^{\prime(4)}$ can be spanned by elements in $\mathfrak{M}_{\gamma}^{(4)}$.

\item[(v)] $J=\mathfrak{S}_{2n-1}$.
To prove that $\mathfrak{S}_{2n-1}^\prime$ can be generated by $\mathfrak{S}_{2n-1}$, let us define sets of operators on $\mathcal H$ labeled by integer $0\leq n'<n$ as
\begin{eqnarray}
\label{eq:def_Mn}
M_{n'} &:=& \{Id_A\otimes\overbrace{Id_{Q^{(n-1)}}\otimes\cdots\otimes Id_{Q^{(n^\prime + 1)}}}^{n-n'-1}
\nonumber\\&&{}\quad
\otimes W_{n'}\otimes W_{n'-1}\otimes\cdots\otimes W_{1}\}_{W_{m}\in\{X,Y,Z,Id\}},
\end{eqnarray}
and prove that, when $m < n$, $M_{m} \subset M_n$ and any element in $M_{m}$ can be generated by $M_{m-1}$ and $\mathfrak{S}_{2n-1}$.

As the basis, we see that $M_{0}$ contains only one element  $Id_E$.
For the inductive step, suppose that any $h\in M_{m-1}$ can be generated from $\mathfrak{S}_{2n-1}$ for $m<n$.
Then, by taking anti-commutators of elements in $\mathfrak{S}_{2n-1}$, we see
\begin{eqnarray}
\label{198}
Id_1\otimes W\otimes Id_2 &=& \frac12\{ Id_1\otimes W\otimes Y_2, Id_1\otimes Id_{Q^{(m)}}\otimes Y_2\},
\nonumber\\
\\
\label{199}
h\cdot (Id_1\otimes W'\otimes Id_2) &=& \frac12\{h,Id_1\otimes W'\otimes Id_2\},
\end{eqnarray}
where
\begin{eqnarray}
Id_1 &:=& Id_A\otimes Id_{Q^{(n-1)}}\otimes Id_{Q^{(n-2)}}\otimes\cdots\otimes Id_{Q^{(m+1)}}, \nonumber
\\
Id_2 &:=& Id_{Q^{(m-1)}}\otimes Id_{Q^{(m-2)}}\otimes\cdots\otimes Id_{Q^{(1)}}, \nonumber 
\\
Y_2 &:=& 
\overbrace{Y\otimes\cdots\otimes Y}^{m-1}, \nonumber 
\\
W &\in& \{X,Z\}, \; W^\prime \in \{X,Y,Z\}, \; h \in M_{m-1}. \nonumber
\end{eqnarray}
Note that $Id_1\otimes W\otimes Y_2$ is in $\mathfrak{S}_{2n-1}$, and $Id_1\otimes Id_{Q^{(m)}}\otimes Y_2$ is in $M_{m-1}$. In addition, $Id_1\otimes Y\otimes Id_2$ is obtained by taking a commutator of elements in $\mathfrak{S}_{2n-1}$.
Therefore, Eqs. (\ref{198}) and (\ref{199}) mean that, when $m<n$, we can generate any element  in $M_{m}=\{h\cdot (Id_1\otimes W''\otimes Id_2)\}$, where ${h\in M_{m-1},W''\in \{Id,X,Y,Z\}}$, by $M_{m-1}$ and $\mathfrak{S}_{2n-1}$. Combining this fact and the assumption as well as $Id_E\in \mathfrak{S}_{2n-1}$, we can conclude that any element in $M_{m}$ can be generated from $\mathfrak{S}_{2n-1}$ where $0\leq m\leq n-1$.

Since $M_{n-1}$ is a basis of $\mathfrak{S}_{2n-1}^\prime$, we have proved that $\mathfrak{S}_{2n-1}^\prime$ can be generated by $\mathfrak{S}_{2n-1}$.

\item[(vi)] $J=\mathfrak{S}_{2n}$. Since $\mathfrak{S}_{2n-1}$ is a subspace of $\mathfrak{S}_{2n}$, any element in $M_{n-1}$ can be generated from $\mathfrak{S}_{2n}$, where $M_m$ is the set defined in Eq. (\ref{eq:def_Mn}). The following equalities
\begin{eqnarray}
Z^*\otimes Id_3
&=&
\frac 12\{Z^{*}\otimes Y_3 ,Id_A\otimes Y_3\},
\nonumber\\
h\cdot Z^*\otimes Id_3
&=&
\frac 12 \{Z^*\otimes Id_3,h\},
\end{eqnarray}
where
\begin{eqnarray}
h &\in& M_{n-1} \nonumber \\
Id_3 &:=& Id_{Q^{(n-1)}}\otimes Id_{Q^{(n-2)}}\otimes\cdots\otimes Id_{Q^{(1)}} 
\nonumber \\
Y_3 &:=& \overbrace{Y\otimes\cdots\otimes Y}^{n-1}, \nonumber
\end{eqnarray}
indicate that  $h\cdot Z^*\otimes Id_3$ can be generated  from $M_{n-1}$ and $\mathfrak{S}_{2n}$.
Since $M_{n-1}\cup \{h\cdot Z^*\otimes Id_3\}_{h\in M_{n-1}}$ is a basis of $\mathfrak S_{2n}^\prime$,
$\mathfrak{S}_{2n}^\prime$ can be generated by $\mathfrak{S}_{2n}$.

\end{enumerate}
\end{proof}

\section{Relation with other investigations about indirect control}
There has been a paper~\cite{AAR15}, whose results appear to be similar to ours at the first sight. Although nothing is conflictive and their paper is very significant in its own right, we find it instructive to describe their main results in our language and elucidate the generality of our results. 

Let us prove the central results in~\cite{AAR15}, namely, its Theorems 2 and 3, with our theorems and lemmas. 
We shall keep using our notations for the sake of consistency, although \cite{AAR15} uses a set of different notations\footnote{For instance, the systems $S$ and $A$ in \cite{AAR15} correspond to $E$ and $S$ in the present paper.}. 

They make several assumptions for the Lie algebra $L$ on $\mathcal H_E\otimes \mathcal H_S$:
\begin{itemize}
\item[(i)] The set $ L$ contains at least one  element which is nonzero in $\mathcal L({\rm su}({\rm dim}\mathcal H_E)\otimes {\rm su}({\rm dim}\mathcal H_S))$. (the condition (A-a) therein)
\item[(ii)] Generators of any control on the space $\mathcal H_S$ are in the algebra $L$, that is, $\{Id_E\} \otimes \mathrm{su(dim}\mathcal{H}_S)\in L$. (the condition (A-b))
\item[(iii)] All elements in $L$ are traceless.
\end{itemize}
Let us reexpress those theorems in our notation before proving them by using our results under these assumptions.

\noindent
\textbf{Theorem 2 in \cite{AAR15}}. When ${\rm dim}\mathcal H_S\geq 3$, for any density matrix $\rho_S$ on $\mathcal H_S$, i.e., any positive semi-definite operator with unit trace,
\begin{equation}
L={\rm su}({\rm dim}\mathcal H_E\cdot {\rm dim}\mathcal H_S)
\Longleftrightarrow
\forall U,\:\exists g\in L,\:\forall \rho_E,\:
{\rm Tr}_S e^g \rho_E\otimes \rho_S e^{-g}=U\rho_EU^\dagger
\label{th:alessandro_1}
\end{equation}
where $\rho_E$ and $U$ are a density operator and a unitary operator on $\mathcal H_E$, respectively.

\noindent
\textbf{Theorem 3 in \cite{AAR15}}. When ${\rm dim}\mathcal H_S= 2$, different structures occur for the dynamical Lie algebra $L$, depending on the rank of the density matrix $\rho_S^\prime$. If $\rho_S'$ is of rank-2 on $\mathcal H_S$, the same proposition as (\ref{th:alessandro_1}) holds:
\begin{equation}
L={\rm su}({\rm dim}\mathcal H_E\cdot {\rm dim}\mathcal H_S)
\Longleftrightarrow
\forall U,\:\exists g\in L,\:\forall \rho_E,\:
{\rm Tr}_S e^g \rho_E\otimes \rho'_S e^{-g}=U\rho_EU^\dagger.
\label{th:alessandro_2}
\end{equation}

If rank$\rho_S=1$, namely, $\rho_S=|\phi_S\rangle\langle \phi_S|$,
\begin{eqnarray}
&& \exists \bar J,J\subseteq i\cdot {\rm u}({\rm dim}\mathcal H_E),
\nonumber\\&&\makebox[1cm]{}
 L=\mathcal L(i\bar J\otimes \{Id_S\}\cup J\otimes {\rm su}({\rm dim}\mathcal H_S))\;\land\;i\mathcal L(\bar J\cup J)={\rm u}({\rm dim}\mathcal H_E)
\nonumber\\
&&\makebox[.7cm]{}
\Longleftrightarrow \forall U,\:\exists g\in L,\:\forall \rho_E,\:
{\rm Tr}_S e^g \rho_E\otimes \ket{\phi_S}\bra{\phi_S} e^{-g}=U\rho_EU^\dagger.
\label{th:alessandro_3}
\end{eqnarray}

The right arrows in (\ref{th:alessandro_1}) and  (\ref{th:alessandro_2}) are trivial.
The right arrow in (\ref{th:alessandro_3}) can be justified as follows.
From the condition in the LHS of (\ref{th:alessandro_3}), for any unitary operator $U$ on $\mathcal H_E$, there is an element $g=\alpha_1\otimes Id + \alpha_2\otimes (\ket{\phi_S}\bra{\phi_S}-\ket{\phi_S^\perp}\bra{\phi_S^\perp}) \in L$ such that $e^g=U\otimes \ket{\phi_S}\bra{\phi_S}+V\otimes \ket{\phi_S^\perp}\bra{\phi_S^\perp}$, where
$U=e^{\alpha_1+\alpha_2}$ and $V=e^{\alpha_1-\alpha_2}$ are unitary operators on $\mathcal H_E$ . 

So, it is sufficient to prove the left arrows in these propositions to obtain the theorems.
To this end, the following two additional lemmas will be useful to use our result for them.


\begin{lemma}
\label{lem:lemma_1_for_ex}
If a positive matrix $\rho_E$ is non-zero and not proportional to the identity operator,
\begin{equation}
\forall U,\:\exists g\in L,\:
{\rm Tr}_S e^g \rho_E\otimes \rho_S e^{-g}\propto U\rho_EU^\dagger
\Longrightarrow
{\rm u}({\rm dim}\mathcal H_E)=\{{\rm Tr}_S \tilde{g}\}_{\tilde{g}\in Ad_L^\infty(\rho_E\otimes \rho_S)},
\label{th:alessandro_0}
\end{equation}
where $L$ is a set of skew-Hermitian operators, $\rho_S$ is a positive semi-definite operator, and $U$ is a unitary operator on $\mathcal H_E$.
Also, 
\begin{eqnarray*}
 Ad_L^\infty(\rho)&:=& \lim_{j\rightarrow \infty}Ad_L^j(\rho),\\
 Ad_L^0(\rho)&:=& \{i\rho\}, \\
 Ad_L^j(\rho)&:=& \mathcal L(Ad_L^{j-1}(\rho)\cup [Ad_L^{j-1}(\rho),L])\;\makebox{ for $j\geq1$}.
\end{eqnarray*}
\end{lemma}
Here, we just sketch an outline of its proof, while it was proved in the paper~\cite{AR12} as Theorem 1.
\begin{proof}
The $\Rightarrow$ in Eq. (\ref{th:alessandro_0}) can be shown as 
\begin{equation}
{\rm u}({\rm dim}\mathcal H_E)
\supseteq
\{{\rm Tr}_S g\}_{g\in Ad_L^\infty(\rho_E\otimes \rho_S)}
\supseteq
i\mathcal L(\{{\rm Tr}_Se^g \rho_E\otimes \rho_S e^{-g}\}_{g\in L})
\supseteq
i\mathcal L(\{U \rho_E U^\dagger\}_{U})
={\rm u}({\rm dim}\mathcal H_E).
\end{equation}
Each inclusion in the above expression can be justified as follows:
The first inclusion is guaranteed by definition of $Ad_L^\infty(\rho)$.
The second one is a result of $\forall g\in L,\; e^g \rho e^{-g}\in iAd_L^\infty(\rho)$, which can be seen by using the Taylor expansion of $e^g$ and $e^{-g}$ for $e^g \rho e^{-g}$. The third one comes from the LHS of (\ref{th:alessandro_0}) and the fact that ${\rm Tr}_S e^g \rho_E\otimes \rho_S
e^{-g}\neq 0$. The last equality is due to the assumption for $\rho_E$.
\end{proof}

\begin{lemma}
\label{lem:lemma_2_for_ex}
If $\rho_S'$ is a full rank positive operator,
\begin{equation}
\forall \rho_E,\:
{\rm Tr}_S e^g \rho_E\otimes \rho_S' e^{-g}= U\rho_EU^\dagger
\Longrightarrow
\forall \rho_E, {\rm Tr}_S e^g \rho_E\otimes Id_S e^{-g}\propto U\rho_EU^\dagger,
\end{equation}
where $g$ is a skew-Hermitian operator, $U$ is a unitary operator on $\mathcal H_E$, and $\rho_E$ is a density matrix on $\mathcal H_E$.
\end{lemma}
\begin{proof}
In can be shown directly by following the chain of relations:
\begin{eqnarray}
&&\forall \rho_E,\:{\rm Tr}_S e^g \rho_E\otimes \rho_S' e^{-g}= U\rho_EU^\dagger
\\
&\Longrightarrow&
\forall \ket{\phi_E},\:{\rm Tr}_S
e^g \ket{\phi_E}\bra{\phi_E}\otimes \rho_S' e^{-g}= U\ket{\phi_E}\bra{\phi_E}U^\dagger \label{208}
\\
&\Longrightarrow&
\forall \ket{\phi_S},\:\forall \ket{\phi_E},\:{\rm Tr}_S e^g \ket{\phi_E}\bra{\phi_E}\otimes \ket{\phi_S}\bra{\phi_S} e^{-g}= U\ket{\phi_E}\bra{\phi_E}U^\dagger
\\
&\Longrightarrow&
\forall \ket{\phi_E},\:{\rm Tr}_S e^g \ket{\phi_E}\bra{\phi_E}\otimes Id_S e^{-g}\propto U\ket{\phi_E}\bra{\phi_E}U^\dagger
\\
&\Longrightarrow&
\forall \rho_E,\:{\rm Tr}_S e^g \rho_E\otimes Id_S e^{-g}\propto U\rho_EU^\dagger
\end{eqnarray}
The first and the third arrows are trivial, and others can be justified as follows.
The second one can be seen by decomposing $\rho_S^\prime$ as $\rho_S'=\rho_S''+\delta \ket{\phi_S}\bra{\phi_S}$ with any $\ket{\phi_S}$, a positive operator $\rho_S^{\prime\prime}$ and an appropriate positive number $\delta$. Therefore,
if the relation (\ref{208})
\begin{eqnarray}
{\rm Tr}_S e^g \ket{\phi_E}\bra{\phi_E}\otimes \rho_S' e^{-g}
&=&
{\rm Tr}_S e^g \ket{\phi_E}\bra{\phi_E}\otimes \rho_S'' e^{-g}+\delta
 {\rm Tr}_Se^g \ket{\phi_E}\bra{\phi_E}\otimes \ket{\phi_S}\bra{\phi_S} e^{-g}
\nonumber\\
&=&U\ket{\phi_E}\bra{\phi_E}U^\dagger
\end{eqnarray}
holds, both terms in the middle must be proportional to $U\ket{\phi_E}\bra{\phi_E}U^\dagger$ since they are both positive and $U\ket{\phi_E}\bra{\phi_E}U^\dagger$ is of rank 1.
Combining this and the fact that $U$ and $e^g$ are unitary operators, we know that
${\rm Tr }_Se^g \ket{\phi_E}\bra{\phi_E}\otimes \ket{\phi_S}\bra{\phi_S} e^{-g}$ must be $U\ket{\phi_E}\bra{\phi_E}U^\dagger$.
The last arrow holds simply because any positive operator $\rho_E$ can be written as a linear combination of rank 1 projection operators.
\end{proof}

Now, we can give a simple proof for the left arrow of (\ref{th:alessandro_1}). Consider a density operator $\rho_E$ that is proportional to $Id_{B_1}\otimes \ket{0}_{R_1}{}_{R_1}\bra{0}$, where the tensor product structure is the one shown in Theorem \ref{th2:str_ds3_dem}, i.e., $\bigoplus_j\mathcal H_{B_j}\otimes \mathcal H_{R_j}$. Then, if the RHS of Eq. (\ref{th:alessandro_1}) holds, $\rho_E$ cannot be the identity operator in $\mathcal{H}_E$. This is because if $\rho_E \propto Id_E$ there exists a single $j$, say 1, in the direct sum above and the dimension of $R_1$ is one. This implies, due to Theorem \ref{th2:str_ds3_dem}, that $L$ is a subset of $\mathcal{L}(i\cdot \mathrm{u(dim}\mathcal{H}_E) \otimes \{Id_S\} \cup i\cdot Id_E\otimes \mathrm{su(dim}\mathcal{H}_S))$, and this contradicts with the assumption (i) above. Therefore, from the condition in the RHS of (\ref{th:alessandro_1}) and Lemma \ref{lem:lemma_1_for_ex},  ${\rm u}({\rm dim}\mathcal H_E)=\{{\rm Tr}_S g\}_{g\in Ad_L^\infty(Id_{B_1}\otimes \ket{0}_{R_1}{}_{R_1}\bra{0}\otimes \rho_S)}$ must hold. This relation and the structure of $L$, i.e.,
$L=\mathcal L(L_d' \cup \bigoplus_j \{Id_{B_j}\}\otimes {\rm su}({\rm dim}\mathcal H_{R_j}\cdot {\rm dim}\mathcal H_S) )$
with $L_d^\prime\subseteq L_d=\bigoplus_j  {\rm u}({\rm dim}\mathcal H_{B_j})\otimes Id_{R_j}\otimes Id_S $, tell us that the index $j$ can take only one value $1$ and ${\rm dim}\mathcal H_{B_1}=1$. 
Further, the assumption (iii), stating that all elements are traceless, implies that $L_d'$ can contain only $0$. Hence, the left arrow of (\ref{th:alessandro_1}) is verified.

Next, let us give a simple proof for the left arrow in (\ref{th:alessandro_2}).
From the condition in the RHS of (\ref{th:alessandro_2}) and Lemma \ref{lem:lemma_2_for_ex}, 
\begin{equation}
\forall U,\:\exists g\in L,\:\forall \rho_E,\:
{\rm Tr}_S e^g \rho_E\otimes Id_S e^{-g}\propto U\rho_EU^\dagger
\label{eq:alesssandro_01}
\end{equation}
must hold.
From Theorems \ref{th:con_discon_Lie_alg} and \ref{th3:str_ds2_dem},
we know that $L$ has a structure such that
\begin{eqnarray}
L&=&\mathcal L(L_d^\prime \cup  \bigoplus_j \mathcal L(i\bar J_j\otimes \{Id_S\}\cup J_j\otimes {\rm su}({\rm dim}\mathcal H_S)) ),
\label{eq:alesssandro_02a}
\\
L_d^\prime &\subseteq& \bigoplus_j i\hat J_j\otimes \{ Id_S\}
\label{eq:alesssandro_02b}
\end{eqnarray}
where candidates of the triple $(J_j, \bar J_j, \hat J_j)$ are given in Eqs. (\ref{eq:JordanAlg_1})-(\ref{eq:str_hat_6}).
From the assumption (iii), we can pick a density matrix  $\rho_E$ proportional to $Id_{E_1}+h$ where $h$ is an element in the set $\bar J_1$
so that $\rho_E$ is not proportional to $Id_E$.
Therefore, from (\ref{eq:alesssandro_01}) and Lemma \ref{lem:lemma_1_for_ex}, ${\rm u}({\rm dim}\mathcal H_E)=\{{\rm Tr}_S g\}_{g\in Ad_L^\infty((Id_{E_1}+h)\otimes Id_S)}$ must hold.
Since $i(Id_{E_1}+h)\otimes Id_S$ is in $\mathcal L(L\cup i\{Id_{E_1}\otimes Id_S\})$ and
the latter is obviously closed under the commutation relation, the relation $Ad_L^\infty((Id_{E_1}+h)\otimes Id_S)\subseteq  \mathcal L(L\cup i\{Id_{E_1}\otimes Id_S\})$ holds.
These two relations allow us to have 
\begin{equation}\label{eq:215.1}
{\rm u}({\rm dim}\mathcal H_E)=\{{\rm Tr}_S g\}_{g\in \mathcal L(L\cup i\{Id_{E_1}\otimes Id_S\})}.
\end{equation}
This relation and the structure of $L$ written above enforce us that the index $j$ can take only one value $1$. Then, we can define a set $\hat{J}_1^\prime$ such that $\hat{J}_1^\prime\subseteq \hat{J}_1$ and $L_d^\prime = i\hat{J}_1^\prime\otimes \{Id_S\}$. Equation (\ref{eq:215.1}) now implies 
\begin{equation}\label{eq:216.1}
i\hat{J}^\prime_1\cup i\bar{J}_1\cup \{i Id_{E_1}\}=\mathrm{u(dim}\mathcal{H}_{E_1})=\mathrm{u(dim}\mathcal{H}_E),
\end{equation}
which means that the dimension of $\mathcal H_{A_1}$ 
is equal to 1.  Thus, $J_1^\prime\subset \mathcal{L}(\{Id_{E_1}\})$, and $\bar{J}_1$ must be sandwiched as $\mathrm{u(dim}\mathcal{H}_{E_1})\supseteq \bar{J}_1\supseteq \mathrm{su(dim}\mathcal{H}_{E_1})$, from which we can deduce $\hat{J}_1$ should be either $\hat{\mathfrak S}_{4}$ or $\hat{\mathfrak M}_\gamma^{(2)}$ with an appropriate integer $\gamma\geq 3$.
Note that the case of $J_1=\mathfrak R$ is ruled out from the assumption (i).
Moreover, the assumption (iii) indicates that $L_d^\prime$ can contain only $0$, and thus the left arrow in (\ref{th:alessandro_2}) is shown.

Finally, we give a simple proof for the left arrow in (\ref{th:alessandro_3}).
Similarly to the above case, we pick a density matrix $\rho_E \propto Id_{E}+h$ where $h$ is an element in the set $\bigoplus_j \bar{J}_j$ so that $\rho_E$ is not proportional to $Id_E$.
From the right relation in (\ref{th:alessandro_3}) and Lemma \ref{lem:lemma_1_for_ex}, the relation
\begin{eqnarray}
 {\rm u}({\rm dim}\mathcal H_E)=\{{\rm Tr}_S g\}_{g\in Ad_L^\infty((Id_{E_1}+h)\otimes \ket{\phi_S}\bra{\phi_S})} \label{eq:alesssandro_03}
\end{eqnarray}
must hold. Here, we recycle the definitions of $L$ and $L_d^\prime$ in Eqs. (\ref{eq:alesssandro_02a}) and (\ref{eq:alesssandro_02b}).
Since any element in $L$ and $(Id_{E_1}+h)\otimes \ket{\phi_S}\bra{\phi_S}$ is block diagonalized into the subspaces $\mathcal H_{E_j}\otimes \mathcal H_S$, in order to satisfy the above relation, the index $j$ can take only a single value $1$.
Since $h$ is in $\bar{J}_1$, its form is one of those in Eqs. (\ref{eq:Rbar})-(\ref{eq:Sbar_2n}). Together with other forms of operators, i.e., Eqs. (\ref{eq:str_hat_1})-(\ref{eq:str_hat_6}) and (\ref{eq:JordanAlg_1})-(\ref{eq:JordanAlg_6}), it can be shown that the components of $Ad_L^\infty((Id_{E_1}+h)\otimes \ket{\phi_S}\bra{\phi_S})$ in the space $\mathcal{H}_{A_1}$ should be either $Id$ or $Z^*$. Since both $Id$ and $Z^*$ are obviously diagonalized in $\mathcal H_{A_1}$, Eq. (\ref{eq:alesssandro_03}) means that  ${\rm dim}\mathcal H_{A_1}$ must be $1$ so that \{Tr$_S g$\} can span u(dim$\mathcal{H}_E$), which then implies $L_d^\prime \subseteq i\{Id_E\otimes Id_S\}$. Taking the assumption (iii) into account, we can conclude $L_d^\prime= \{0\}$.

With the help of forms of $J$ and $\hat{J}$ in Eqs. (\ref{eq:JordanAlg_1})-(\ref{eq:JordanAlg_6}) and (\ref{eq:str_hat_1})-(\ref{eq:str_hat_6}), we can now verify whether each type of $J$ in Eqs. (\ref{eq:JordanAlg_1})-(\ref{eq:JordanAlg_6}) satisfies the requirement Eq. (\ref{eq:alesssandro_03}). First, let us have a look at $J_1=\mathfrak S_{2n'-1}$. The set $iAd_L^\infty((Id_{E_1}+h)\otimes \ket{\phi_S}\bra{\phi_S})$ of operators on $ \mathcal H_{Q_1^{(n'-1)}}\otimes  \mathcal H_{Q_1^{(n'-2)}}\otimes \cdots \otimes \mathcal H_{Q_1^{(1)}}\otimes \mathcal H_S=\mathcal H_{E_1}\otimes \mathcal H_S=\mathcal H_E\otimes \mathcal H_S$ can be written as
\begin{eqnarray}
&&\mathcal L(
\{
           Id^{ \otimes \Delta_1}\otimes W_1
\otimes Y^{\otimes  \Delta_2}\otimes W_2
\otimes Id^{ \otimes \Delta_3}\otimes W_3
\otimes Y^{\otimes  \Delta_4}\otimes W_4
\otimes Id^{ \otimes \Delta_5}\}_{W_k\in \{X,Z\}, \Delta_k\in \mathbb Z_{\geq 0}\;s.t. \sum_{k=1}^5\Delta_k=n'-4},
\nonumber\\&&\cup \{
           Id^{ \otimes \Delta_1}\otimes W_1
\otimes Y^{\otimes  \Delta_2}\otimes Id
\otimes Y^{\otimes  \Delta_3}\otimes W_2
\otimes Id^{ \otimes \Delta_4},
\nonumber\\&&
           Id^{ \otimes \Delta_1}\otimes Y
\otimes Id^{\otimes  \Delta_2}\otimes W_1
\otimes Y^{\otimes  \Delta_3}\otimes W_2
\otimes Id^{ \otimes \Delta_4},
\nonumber\\&&
           Id^{ \otimes \Delta_1}\otimes W_1
\otimes Y^{\otimes  \Delta_2}\otimes W_2
\otimes Id^{\otimes  \Delta_3}\otimes Y
\otimes Id^{ \otimes \Delta_4},
\nonumber\\&&
           Id^{ \otimes \Delta_1}\otimes W_1
\otimes Y^{\otimes  \Delta_2}\otimes W_2
\otimes Id^{\otimes  \Delta_3}\otimes W_3
\otimes Y^{ \otimes \Delta_4}
\}_{W_k\in \{X,Z\}, \Delta_k\in \mathbb Z_{\geq 0}\;s.t. \sum_{k=1}^4\Delta_k=n'-3},
\nonumber\\&&\cup \{
           Id^{ \otimes \Delta_1}\otimes Y
\otimes Id^{\otimes  \Delta_2}\otimes Y
\otimes Id^{ \otimes \Delta_3},
           Id^{ \otimes \Delta_1}\otimes Y
\otimes Id^{\otimes  \Delta_2}\otimes W_1
\otimes Y^{ \otimes \Delta_3},
\nonumber\\&&
           Id^{ \otimes \Delta_1}\otimes W_1
\otimes Y^{\otimes  \Delta_2}\otimes Id
\otimes Y^{ \otimes \Delta_3},
           Id^{ \otimes \Delta_1}\otimes W_1
\otimes Y^{\otimes  \Delta_2}\otimes W_2
\otimes Id^{ \otimes \Delta_3}
\}_{W_k\in \{X,Z\}, \Delta_k\in \mathbb Z_{\geq 0}\;s.t. \sum_{k=1}^3\Delta_k=n'-2},
\nonumber\\&&\cup \{
           Id^{ \otimes \Delta_1}\otimes Y
\otimes Id^{\otimes  \Delta_2},
           Id^{ \otimes \Delta_1}\otimes W
\otimes Y^{\otimes  \Delta_2}
\}_{W\in\{X,Z\},\Delta_k\in \mathbb Z_{\geq 0}\;s.t. \Delta_1+\Delta_2=n'-1}
\nonumber\\&&
\cup \{Id^{\otimes n'}\})=:\Sigma.
\label{eq:219}
\end{eqnarray}
Here, we have omitted the space $\mathcal H_{A_1}$ since its dimension is $1$.
We can see from Eq. (\ref{eq:219}) that Eq. (\ref{eq:alesssandro_03}) cannot be satisfied when $n'\geq 3$. 
Thus, $J_1=\mathfrak S_{2n'-1}$ is not allowed when $n'\geq 3$.

Second, we repeat a similar check for $J_1=\mathfrak S_{2n'}$. The set $iAd_L^\infty((Id_{E_1}+h)\otimes \ket{\phi_S}\bra{\phi_S})$ can now be written as
\begin{eqnarray}
&&\mathcal L(
\{Y^{ \otimes \Delta_1}\otimes W_1
\otimes Id^{\otimes  \Delta_2}\otimes W_2
\otimes Y^{ \otimes \Delta_3},
           Y^{ \otimes \Delta_1}\otimes W_1
\otimes Id^{\otimes  \Delta_2}\otimes Y
\otimes Id^{ \otimes \Delta_3}\}_{W_k\in \{X,Z\}, \Delta_k\in \mathbb Z_{\geq 0}\;s.t. \sum_{k=1}^3\Delta_k=n'-2},
\nonumber\\&&\cup \{
           Y^{ \otimes \Delta_1}\otimes Id
\otimes Y^{\otimes  \Delta_2},
           Y^{ \otimes \Delta_1}\otimes W
\otimes Id^{\otimes  \Delta_2}
\}_{W\in\{X,Z\},\Delta_k\in \mathbb Z_{\geq 0}\;s.t. \Delta_1+\Delta_2=n'-1}
\cup
\{Y^{ \otimes n'}\}
\cup\Sigma
).
\end{eqnarray}
From this, we again see that the requirement (\ref{eq:alesssandro_03}) cannot be fulfilled when $n'\geq 4$. Therefore, $J_1=\mathfrak S_{2n'}$
is ruled out for $n'\geq 4$.

Combining all these results, we can conclude that $L$ should have the form $L=\mathcal L(i\bar J_1\otimes \{Id_S\}\cup J_1\otimes {\rm su}({\rm dim}\mathcal H_S) )$, where dim($\mathcal H_{A_1})=1$ and $(\bar J_1,J_1)$ is equal to either $(\hat {\mathfrak S}_{n},\mathfrak S_{n})$ or $(\bar {\mathfrak M}_\gamma^{(k)}, \mathfrak M_\gamma^{(k)})$ with $n\in\{3,4,6\}$, $k\in\{1,2,4\}$ and $\gamma\in\mathbb Z_{\geq 3}$.
Also, it is straightforward to check $\mathcal L(\bar J_1\cup J_1)={\rm u}({\rm dim}\mathcal H_E)$ for any choice of $(\bar J_1,J_1)$. Note, however, that the choice $(\bar J_1,J_1)=(\bar {\mathfrak R},\mathfrak R)$ is ruled out because of the assumption (ii). Hence, the left arrow in (\ref{th:alessandro_3}) is proved.

\bibliographystyle{apsrev}

\begin{thebibliography}{31}
\bibitem{LLS04}
S. Lloyd, A.J. Landahl, J-J.E. Slotine,
Universal quantum interfaces,
Physical Review A \textbf{69}, 012305 (2004)

\bibitem{BMM10}
D. Burgarth, K. Maruyama, M. Murphy, S. Montangero, T. Calarco, F. Nori, M.B. Plenio,
Scalable quantum computation via local control of only two qubits,
Physical Review A \textbf{81}, 040303(R) (2010)

\bibitem{KP10}
A. Kay, P.J. Pemberton-Ross,
Computation on spin chains with limited access,
Physical Review A \textbf{81}, 010301(R) (2010)

\bibitem{GLST17}
C. Gokler, S. Lloyd, S. Shor, K. Thompson,
Efficiently Controllable Graphs,
Physical Review Letters \textbf{118}, 260501 (2017)

\bibitem{LARR18}
J. Lee, C. Arenz, H. Rabitz, B. Russell, 
Dependence of the quantum speed limit on system size and control complexity,
New Journal of Physics \textbf{20}, 063002 (2018)

\bibitem{JS72}
V. Jurdjevi\'{c}, H. Sussmann,
Control systems on Lie groups,
Journal of Differential Equations \textbf{12}, 313-329 (1972)

\bibitem{RSDRP95}
V. Ramakrishna,  M.V. Salapaka, M. Dahleh, H. Rabitz, A. Peirce,
Controllability of molecular systems,
Physical Review A \textbf{51}, 960-966 (1995)

\bibitem{DAlessandroBook}
D. D'Alessandro,  
Introduction to Quantum Control and Dynamics,
Taylor and Francis, Boca Raton (2008)

\bibitem{MBBookChapter}
K. Maruyama, D. Burgarth,
Gateway schemes of quantum control for spin networks,
in: 
T. Takui, L. Berliner, G. Hanson,  (eds.) 
Electron Spin Resonance (ESR) Based Quantum Computing,
 Springer, Heidelberg  pp. 167-192, (2016)

\bibitem{ZSH11}
R. Zeier, T. Schulte-Herbr\"{u}ggen,
Symmetry principles in quantum systems theory,
Journal of Mathematical Physics \textbf{52}, 113510 (2011)

\bibitem{JNW34}
P. Jordan, J. von Neumann, E. Wigner,
On an algebraic generalization of the quantum mechanical formalism,
Annals of Mathematics \textbf{35}(1), 29 (1934)

\bibitem{BMN09}
D. Burgarth, K. Maruyama, F. Nori,
Coupling strength estimation for spin chains despite restricted access,
Physical Review A \textbf{79}, 020305(R) (2009)

\bibitem{FPK09}
C.D. Franco, M. Paternostro, M.S. Kim,
Hamiltonian Tomography in an Access-Limited Setting without State Initialization,
Physical Review Letters \textbf{102}, 187203 (2009)

\bibitem{BM09}
D. Burgarth, K. Maruyama,
Indirect Hamiltonian identification through a small gateway,
New Journal of Physics \textbf{11}, 103019 (2009)

\bibitem{BMN11}
D. Burgarth, K. Maruyama, F. Nori,
Indirect quantum tomography of quadratic Hamiltonians,
New Journal of Physics \textbf{13}, 013019 (2011)

\bibitem{Shabani11}
A. Shabani, R.L. Kosut, M. Mohseni, H. Rabitz, M.A. Broome, M.P. Almeida, A. Fedrizzi, A.G. White,
Efficient Measurement of Quantum Dynamics via Compressive Sensing,
Physical Review Letters \textbf{106}, 100401 (2011)

\bibitem{JS14}
J. Zhang, M. Sarovar,
Quantum Hamiltonian Identification from Measurement Time Traces,
Physical Review Letters \textbf{113}, 080401 (2014)

\bibitem{KY14}
Y. Kato, N. Yamamoto,
Structure identification and state initialization of spin networks with limited access,
New Journal of Physics \textbf{16}, 023024 (2016)

\bibitem{SC17a}
A. Sone, P. Cappellaro,
Hamiltonian identifiability assisted by a single-probe measurement,
Physical Review A \textbf{95}, 022335 (2017)

\bibitem{SC17b}
A. Sone, P. Cappellaro,
Exact dimension estimation of interacting qubit systems assisted by a single quantum probe,
Physical Review A \textbf{96}, 062334 (2017)

\bibitem{DD09}
D. D'Alessandro,
General methods to control right-invariant systems on compact Lie groups and multilevel quantum systems,
Journal of Physics A: Mathematical and Theoretical, \textbf{42}, 395301 (2009)

\bibitem{AAR15}
D. D'Alessandro, F. Albertini, R. Romano,
Exact algebraic conditions for indirect controllability of quantum systems,
 SIAM Journal on Control and Optimization \textbf{53}, 1509 (2015)

\bibitem{AR12}
D. D'Alessandro, R. Romano,
Indirect Controllability of Quantum Systems; A Study of Two Interacting Quantum Bits,
IEEE Transactions on Automatic Control, \textbf{57}, 2009-2020 (2012)

\bibitem{BY12}
D. Burgarth, K. Yuasa,
Quantum System Identification,
Physical Review Letters \textbf{108}, 080502 (2012)

\bibitem{OMTK15}
M. Owari, K. Maruyama,  T. Takui, Kato, G,
Probing an untouchable environment for its identification and control,
Physical Review  A \textbf{91}, 012343 (2015)

\bibitem{TBG12}
K.W. Moore Tibbetts, C. Brif, M.D. Grace, A. Donovan, , D.L. Hocker, T.-S. Ho, R.-B. Wu, H. Rabitz,
Exploring the tradeoff between fidelity and time optimal control of quantum unitary transformations,
Physical Review A \textbf{86}, 062309 (2012)

\bibitem{TGLS16}
K.F. Thompson, C. Gokler, S. Lloyd, P.W. Shor,
Time independent universal computing with spin chains: quantum plinko machine,
New Journal of Physics \textbf{18}, 073044 (2016)

\bibitem{HC18}
M. Hirose, P. Cappellaro,
Time-optimal control with finite bandwidth,
Quantum Information Processing \textbf{17}, 88 (2018) 


\end{thebibliography}

\end{document}